\documentclass[]{article}

\usepackage[left=3cm,right=3cm,top=2cm,bottom=2cm]{geometry}
\usepackage{amssymb}
\usepackage{amsthm}
\usepackage{amsmath}
\usepackage{graphicx}
\usepackage{color}
\usepackage{relsize}

\newlength{\figurewidth}
\newlength{\figureheight}
\newtheorem{property}{Proposition}[section]
\newtheorem{definition}{Definition}[section]
\newtheorem{theorem}{Theorem}[section]
\newtheorem{lemma}{Lemma}[section]

\begin{document}
	\setcounter{page}{1}
	
	\title{Optimal market making with persistent order flow}

	\author{
		Paul Jusselin\footnote{paul.jusselin@polytechnique.edu} \\
		\'Ecole Polytechnique, CMAP
		}

	\maketitle
	
	\begin{abstract}
\noindent We address the issue of market making on electronic markets when taking into account the clustering and long memory properties of market order flows. We consider a market model with one market maker and order flows driven by general Hawkes processes. We formulate the market maker's objective as a stochastic control problem. We characterize an optimal control by proving existence and uniqueness of a viscosity solution to the associated Hamilton-Jacobi-Bellman equation. Finally we propose a fully consistent numerical method allowing to implement this optimal strategy in practice.
	\end{abstract}

\noindent \textbf{Keywords:} Hawkes processes, market making, high frequency trading, stochastic control, partial differential equations, viscosity solutions.

\section{Introduction}
\label{sec:introduction}

Most electronic exchanges are organized as anonymous continuous double auction systems. Market participants can send limit orders to a central limit order book (LOB for short) displaying the volume of shares and the price at which they stand ready to buy or sell those shares. Market participants can also use market orders specifying a volume to buy or sell instantaneously at the best available price. In a very stylized view we can consider that there are two types of market participants: market takers seeking to buy or sell shares for strategic purposes using market orders and market makers filling the LOB with limit orders. Market makers play the role of intermediaries between buyers and sellers market takers.\\

In practice one of the main risks faced by a market maker is the inventory risk. For example if he has a large positive inventory, price may decrease to his disadvantage. Market makers thus design their strategies in order to mitigate this risk. Basically we expect a market maker with a large positive inventory to set attractive ask prices and less competitive bid prices to attract more buy than sell market orders. More generally he must adapt his strategy to the main statistical features of order flows. Two key stylized facts market makers should take into account in their trading strategies are the clustering and long memory properties of market order flows. The clustering property refers to the fact that buy and sell market orders are not distributed homogeneously in time but tend to be clustered, see \cite{hewlett2006clustering}. In practice it means that after a buy (for say) market order it is likely that a new one is going to be sent shortly. The long memory of market order flows means the autocorrelation function of trades signs ($+1$ for a buy order and $-1$ for a sell order) has a power-law tail, see \cite{lillo2004long}. In this paper our goal is to propose a method to design market making strategies that take into account those two features of market order flows.\\

The issue of market making while managing inventory risk has been notably addressed in \cite{avellaneda2008high, gueant2013dealing} where market order flows are modeled using Poisson processes, see also the books \cite{cartea2015algorithmic, gueant2016financial}. But Poisson processes neither reproduce the clustering nor the long memory property of market order flows. Dealing with the same topic the authors of \cite{cartea2014buy, chen2020optimal} use a refined model based on Hawkes processes with exponential kernels. The same modeling is also used in \cite{alfonsi2016dynamic, hewlett2006clustering} to design optimal liquidation strategies. With such kernels Hawkes processes reproduce the clustering property of market order flows but not its long memory. However when the kernel has a power law tail, both properties are reproduced, see \cite{bacry2016estimation, jaisson2016rough}. Hence in this paper we extend the works \cite{avellaneda2008high, cartea2014buy, gueant2013dealing} to market order flows driven by Hawkes processes with general kernels.\\

We now make precise the market model we use. We consider a market with one market maker controlling the best bid and ask prices and with market takers sending only market orders of unit volume. We denote by $N^a_t$ (resp. $N^b_t$) the total number of buy (resp. sell) market orders sent between time $0$ and time $t$ and $i_t=N_t^b - N_t^a$ the market maker's inventory, which is null at time $0$. As in \cite{avellaneda2008high} the market maker controls the bid and ask spreads, denoted by $\delta^a$ and $\delta^b$. The corresponding best ask and bid prices are $P+\delta^a$ and $P-\delta^b$, where $P$ is the fundamental price of the underlying asset. The set of admissible controls is then 
$$
\mathcal{A} = \{\delta=(\delta^a, \delta^b) \in \mathbb{R}_+^2,\text{ s.t. }\delta\text{ is predictable}\},
$$
where predictability is relative to the natural filtration generated by $(P,N^a, N^b)$, see Section \ref{sec:formal_definition} for more details. Since market takers are seeking for low transaction costs, their trading intensity is decreasing with the spreads. More precisely we know, from classical financial economics results, see \cite{dayri2015large,madhavan1997security, wyart2008relation}, that the average number of trades per unit of time is a decreasing function of the ratio between spread and volatility. To model this we consider that market order intensities are given by
$$
\lambda^{a, \delta}_t = e^{-\frac{k}{\sigma}\delta^a_t}\lambda^{a, 0}_t\text{ and }\lambda^{b, \delta}_t = e^{-\frac{k}{\sigma}\delta^b_t}\lambda^{b, 0}_t,
$$
where $k$ is a positive constant, $\sigma$ the price volatility and
$$
\lambda^{a,0}_t = \Phi\big( \int_0^t K(t-s) \mathrm{d}N^a_s \big),~\lambda^{b,0}_t = \Phi\big( \int_0^t K(t-s) \mathrm{d}N^b_s \big),
$$
where $\Phi$ is a continuous function and $K$ a completely monotone $L^1$ function\footnote{In this paper we consider complete monotony on $\mathbb{R}_+$.}. Note that if the spreads are null, the processes $N^a$ and $N^b$ are, as intended, generalized Hawkes processes with kernel $K$. Regarding the dynamics of $P$ we assume it is given by
	\begin{equation}
	\label{eq:price_dynamic}
	\mathrm{d}P_t = d(t, P_t)\mathrm{d}t + \sigma \mathrm{d}W_t
	\end{equation}
where $d$ is a Lipschitz function.\\

Inspired by \cite{avellaneda2008high, cartea2014buy, gueant2013dealing} we consider that the market maker's problem is equivalent to solve the following stochastic control problem:
	\begin{equation}
	\label{eq:optimal_control_intro}
	\underset{\delta \in \mathcal{A}}{\sup}~ 	\mathbb{E}^{\delta}\Big[G(i_T, P_T) e^{-rT} + \int_0^T e^{-rs}\big(  g(i_s, P_s)\mathrm{d}s + \delta^a_s \mathrm{d}N^a_s+ \delta^b_s\mathrm{d}N^b_s\big)\Big],
	\end{equation}
where $r$ is a positive constant and $g$ and $G$ are two continuous functions with at most quadratic growth. The former represents a continuous reward received by the market maker (besides its P\&L) and the latter is a final lump sum payment received at the end of the trading horizon. Typical choices would be $G(x, y) = xy$ and $g(x, y) = -x^2$. The notation $\mathbb{E}^{\delta}$ denotes the expectation under the law corresponding to the control $\delta$ (see Section \ref{sec:formal_definition} for details).\\

When the order flows are Poisson processes (corresponding to $\phi$ constant) the process $(P, i, N^a, N^b)$ is Markovian. Hence to solve the market maker's optimization problem the authors of \cite{avellaneda2008high} study the associated Hamilton-Jacobi-Bellman equation (HJB for short). However when $N^a$ and $N^b$ are general Hawkes processes $(P,i, N^a, N^b)$ is not Markovian\footnote{In the exponential case the process $(P,i,N^a, N^b, \lambda^b, \lambda^a)$ is Markovian.}. In order to circumvent this issue we need to consider auxiliary state variables enabling us to work in a Markovian setting. More precisely we consider the process $X^K=(P, i, \theta^{K,a}, \theta^{K,b})$ where
	$$
	\theta^{K,a}_t(u) = \int_0^t K(u-s)\mathrm{d}N^a_s\text{ and }\theta^{K,b}_t(u) = \int_0^t K(u-s)\mathrm{d}N^b_s.
	$$
Note\footnote{To define $\theta^{K,a}_t$ and $\theta^{K,b}_t$ we consider that $K$ is extended to $\mathbb{R}$ with value $0$ on $\mathbb{R}_-^*$.} that $\theta^{K,a}_t$ and $\theta^{K,b}_t$ are random functions from $\mathbb{R}_+$ into $\mathbb{R}_+$ and that the process $(t, X^K_t)_{t\geq 0}$ is Markovian. Studying the HJB equation associated with this representation, we prove in Section \ref{sec:solve_problem} that the stochastic control problem \eqref{eq:optimal_control_intro} admits a solution of the form $\delta^{K,*}_t = \delta^K(t, X^K_t)$, where $\delta^K$ is a feedback control function.\\

The HJB equation associated with \eqref{eq:optimal_control_intro} and the representation $(t, X^K_t)_{t\geq 0}$ is defined on a subset of an infinite dimensional vector space. So we cannot rely on classical numerical methods to approximate $\delta^K$. To tackle this issue we propose the following strategy.
\begin{itemize}
\item[1.]We show that if $(K_n)_{n\geq 0}$ converges towards $K$ in $L^1$ and uniformly on $[0, T]$ then $(\delta^{K_n,*})_{n\geq 0}$ converges almost surely towards $\delta^{K,*}$.
\item[2.]We show that when $K(t)= \sum_{i = 1}^n \alpha_i e^{-\gamma_i t } $, there exists a Markovian representation of the model in dimension $2n + 2$. Therefore in this case the feedback control $\delta^{K}$ can be approximated numerically.
\item[3.]Inspired by \cite{jaber2018lifting}, we prove that for any completely monotone kernel $K$ in $L^1$, we can find a sequence $(K_n)_{n\geq 0}$, converging towards $K$ in $L^1$ and uniformly on $[0, T]$, such that for any $n$, $K_n$ is a linear combination of decreasing exponential functions.
\end{itemize}
Those three points show that if $K_n$ is close from $K$ then $\delta^{K_n,*}$, which is related to a finite dimensional HJB equation, is good approximation of $\delta^{K,*}$. However when $n$ is large we cannot rely on finite differences methods to compute the feedback control $\delta^{K_n}$ since the dimension of the associated HJB equation is too large. Hence for numerical experiments we use the probabilistic representation of semi-linear partial differential equations (PDEs for short) introduced in \cite{henry2019branching}. \\

The paper is organized as follows. In Section \ref{sec:solve_problem} we prove existence of a solution to Problem \eqref{eq:optimal_control_intro} based on the study of its associated HJB equation. In Section \ref{sec:approach_optimal_control} we explain how to approximate the optimal control obtained in Section \ref{sec:solve_problem}. Finally in Section \ref{sec:numerical}, we present some numerical experiments. Main proofs are relegated to Section \ref{sec:proofs}.

\section{Solving the market maker problem}
\label{sec:solve_problem}
In this section we prove existence of a solution to Problem \eqref{eq:optimal_control_intro}. First we define an appropriate domain for the process $X^K$. Then we show that the HJB equation associated to \eqref{eq:optimal_control_intro} has a unique viscosity solution with polynomial growth on this domain. Finally we prove existence and characterize an optimal control solving \eqref{eq:optimal_control_intro}.

\subsection{Appropriate set for the process $X^K$}
\label{subsec:state_space}

To study uniqueness of solution to a PDE in the sense of viscosity, it is convenient to work with locally compact domain. We have $X^K = (P, i, \theta^{K,a}, \theta^{K,b})\in \mathbb{R} \times \mathbb{Z}\times L^1 \times L^1$, but $L^1$ is not locally compact. Hence we need to make precise the set in which the processes $\theta^{K,a}$ and $\theta^{K,b}$ belong. Obviously for $j=a$ or $b$ we have
$$
\theta^{K,j}_t \in \Theta_t^K = \{ \sum_{i = 1}^nK(\cdot - T_i), ~ n\in \mathbb{N}, ~ T_1\leq \dots \leq T_n\leq t  \}\subset \Theta^K_T.
$$
We naturally endow $\Theta^K_T$ with the $L^1$ topology and prove in Appendix \ref{proof:lemma_topology_a} that it has the following topological properties.
\begin{lemma}
	\label{lemma:topology_a}~
	\begin{itemize}
		\item[(i)]The set $\Theta^K_T$ is a locally compact closed subset of $L^1$.
		\item[(ii)]For any sequence $(s_n, \theta_n)_{n\geq 0}$ with values in $[0, T]\times \Theta^K_T$ such that for any $n$, $\theta_n \in \Theta^K_{s_n}$, if $(s_n, \theta_n)_{n\geq 0}$ converges towards $(s, \theta)$ then we have $\theta \in \Theta_s^K$ and $\theta_n(s_n) \rightarrow \theta(s)$ when $n\rightarrow + \infty$.
		\item[(iii)]Moreover if $K$ is a sum of exponential functions then we have for any $l\geq 0,$ 
		$$
		\theta_n^{(l)}(T) \rightarrow \theta^{(l)}(T)\text{, when }n\rightarrow + \infty.
		$$
	\end{itemize}
	\end{lemma}
Points $(ii)$ and $(iii)$ are purely technical and are used in Section \ref{sec:approach_optimal_control}. We now define a locally compact domain for the process $X^K$. More precisely for any $t\in [0, T]$ we consider
	\begin{equation*}
	\mathcal{Z}^K_t = \{(i, \theta^a, \theta^b) \in  \mathbb{Z} \times \Theta^K_t \times \Theta^K_t \} \text{ and }	\mathcal{X}^K_t = \{(p, i, \theta^a, \theta^b) \text{ s.t. }  p\in \mathbb{R}\text{ and } (i, \theta^a, \theta^b)\in \mathcal{Z}_t^K \} .
	\end{equation*}
According to Lemma \ref{lemma:topology_a} $(i)$ the set $\mathcal{Z}^K_t$ (resp. $\mathcal{X}^K_t$) is a locally compact closed subset of $\mathbb{Z}\times L^1 \times L^1$ (resp. $\mathbb{R} \times \mathbb{Z}\times L^1 \times L^1$). We also define
	\begin{equation*}
	\mathcal{E}^K = \{(t, x)\in [0, T]\times\mathcal{X}^K_T\text{ s.t } x\in \mathcal{X}^K_t \},
	\end{equation*}
which is a locally compact closed subset of $[0, T]\times \mathbb{R} \times \mathbb{Z} \times L^1 \times L^1$. Obviously for any $t\in [0, T]$ we have $(t, X^K_t)\in \mathcal{E}^K$. Hence $\mathcal{E}^K$ is the appropriate domain to use the theory of viscosity solutions.\\

Before going to the next section we give additional definitions that we use later on. For $ x = (p, i,\theta^a,\theta^b) \in \mathbb{R} \times \mathbb{Z}\times L^1 \times L^1$ we define the norm $$\|x\| = \sqrt{ p^2 + i^2 + \|\theta^a \|^2_{1} + \|\theta^b \|^2_{1} }.$$ Then for any positive $R$ the following set is a compact subset of $\mathcal{E}^K$, as consequence of Lemma \ref{lemma:topology_a} $(i)$,
$$
\mathcal{E}^K_R = \{ (t, x)\in \mathcal{E}^K\text{, s.t. } \|x\|\leq R\}.
$$
Finally in order to lighten the notations from now on when we consider $x \in \mathcal{E}^K$ (resp. $x\in \mathcal{X}^K_t$, $z\in \mathcal{Z}^K_t$) we implicitly assume that $x = (t,p, i,\theta^a, \theta^b)$ $\big($resp. $x = (p, i,\theta^a, \theta^b)$, $z = (i, \theta^a, \theta^b)$ $\big)$.\\

Now that we have defined an adapted domain for PDE analysis we derive in the next section the HJB equation related to the stochastic control problem \eqref{eq:optimal_control_intro}.

\subsection{Hamilton-Jacobi-Bellman equation associated to the control problem}
\label{subsec:HJB_equation}

We start by rewriting the stochastic control problem \eqref{eq:optimal_control_intro}. We note that up to a $\mathbb{P}^\delta$-local martingale the integrals
$$
\int_0^T \delta^a_s \mathrm{d}N^a_s\text{ and }\int_0^T \delta^b_s \mathrm{d}N^b_s.
$$
are respectively equal to $ \int_0^T \delta^a_s \lambda^{a, \delta}_s \mathrm{d}s$ and $\int_0^T \delta^b_s \lambda^{b, \delta}_s\mathrm{d}s$ . Hence as a consequence of Appendix \ref{appendix:apriori_inequalities_integral} for any $\delta \in \mathcal{A}$ we have
\begin{align*}
\mathbb{E}^{\delta}[G(i_T,P_T) e^{-rT} + &\int_0^T e^{-rs}\Big(  g(i_s, P_s)\mathrm{d}s + \delta^a_s \mathrm{d}N^a_s+ \delta^b_s\mathrm{d}N^b_s\Big)]\\
& = 	\mathbb{E}^{\delta}[G(i_T,P_T) e^{-rT} + \int_0^T e^{-rs}\Big(  g(i_s,P_s) + \delta^a_s \lambda^{a,\delta}_s + \delta^b_s\lambda^{b, \delta}_s\Big) \mathrm{d}s] .
\end{align*}
Thus \eqref{eq:optimal_control_intro} is equivalent to the stochastic control problem
\begin{equation}
\label{eq:optimal_control_pb}
\underset{\delta \in \mathcal{A}}{\sup} ~ \mathbb{E}^{\delta}[G(i_T, P_T) e^{-rT} + \int_0^T e^{-rs}\Big(  g(i_s, P_s) + \delta^a_s \lambda^{a, \delta}_s+ \delta^b_s\lambda^{b, \delta}_s\Big)\mathrm{d}s].
\end{equation}
In order to give intuition on the HJB equation related to this stochastic control problem we write the Ito formula related to $X^K$ for a fixed control $\delta\in \mathcal{A}$. We consider a function $\varphi$ defined on $[0, T]\times \mathbb{R} \times \mathbb{Z} \times L^1 \times L^1$ that is $C^{2, 2, 0, 0, 0}$. We call any function with such regularity a \textit{test function}. For any $s<t\in [0, T]$ we have
\begin{align*}
\varphi(t, X^K_t) - \varphi(s, X^K_s) = \int_s^t& \Big( \partial_t\varphi(u, X^K_{u-}) + \mathcal{L}^P\varphi(u, X^K_{u-}) + \sum_{j = a, b} D^K_j \varphi(u, X^K_{u-})e^{-\frac{k}{\sigma}\delta^j_u}\Phi\big(\theta_u^{K,j}(u)\big) \Big)\mathrm{d}u\\
& + \partial_p \varphi(u, X^K_{u-}) \sigma \mathrm{d}W_u + D^K_a\varphi(u, X^K_{u-})\mathrm{d}M^{a;\delta}_u+ D^K_b\varphi(u, X^K_{u-})\mathrm{d}M^{b;\delta}_u,
\end{align*}
where
$$
M^{a;\delta}_t = N^a_t -\int_0^t \lambda^{a,\delta}_s \mathrm{d}s \text{ and }M^{b;\delta}_t = N^b_t -\int_0^t \lambda^{b,\delta}_s \mathrm{d}s
$$
are $\mathbb{P}^{\delta}$-uniformly integrable martingales, see Appendix \ref{appendix:subsec_martingale_hawkes} for details. The operator $\mathcal{L}^P$ is the infinitesimal generator related to the diffusion of $P$ and is defined for any test function $\varphi$ and $(t, x)\in \mathcal{E}^K$ by
$$
\mathcal{L}^P\varphi(t, x) = d(t, p)\partial_p \varphi(t, x) + \frac{1}{2} \sigma^2 \partial^2_{pp}\varphi(t, x).
$$
The operators $D^K_a$ and $D^K_b$ correspond to the infinitesimal generators related to the diffusion of $N^a$ and $N^b$. They are defined for $(t, x)\in \mathcal{E}^K$ by
\begin{align*}
D^K_a \varphi(t, x) &= \varphi\big(t, p, i-1, \theta^a + K(\cdot-t), \theta^b\big) - \varphi(t, p, i, \theta^a, \theta^b),\\
D^K_b \varphi(t, x) &= \varphi\big(t, p, i+1, \theta^a, \theta^b + K(\cdot-t)\big) - \varphi(t, p, i, \theta^a, \theta^b).
\end{align*}
Hence the HJB equation associated to the control problem \eqref{eq:optimal_control_pb} is
\begin{equation*}
(\mathbf{HJB})_K:\left\{ \begin{array}{ll}
&F\big( x, U(x), \nabla U(x), \partial^2_{pp}U(x),  D^KU(x)\big) =0\text{ for }x\in \mathcal{E}^K,\\
&U(T, y) = G(i, p)\text{ for }y \in \mathcal{X}_T^K,
\end{array}\right.
\end{equation*}
with  $\nabla U = \big( \partial_t U, \partial_p U \big)$, $D^KU = \big(D^K_aU, D^K_bU\big)$ and where the function $F$ is defined for $(x, u, q, A, I)\in  \mathcal{E}^K_t \times \mathbb{R}\times \mathbb{R}^2 \times \mathbb{R} \times \mathbb{R}^2$ by
\begin{align*}
F(x, u, q, A, I) = & ru -  q_1  - d(t, p)q_2 - \frac{1}{2} \sigma^2 A - g(i, p)\\
& - \underset{\delta\in \mathbb{R}_+}{\sup} \Phi\big(\theta^a(t)\big)e^{-\frac{k}{\sigma}\delta}(\delta + I_1) - \underset{\delta \in \mathbb{R}_+}{\sup}  \Phi\big(\theta^b(t)\big)e^{-\frac{k}{\sigma}\delta}(\delta + I_2).
\end{align*}
A straightforward computation gives the maximizers
\begin{equation}
\label{eq:maximizers}
\delta^{*a} = \big(\sigma/k - I_1\big)_{+}\text{ and }\delta^{*b} = \big(\sigma/k - I_2\big)_{+}.
\end{equation}
Note that the dependence in $K$ of $(\mathbf{HJB})_K$ lies in the operator $D^K$.\\

For such general partial-integro differential equation (PIDE for short) it is a priori impossible to prove existence of a smooth solution. Therefore in the next section we look for viscosity solutions.

\subsection{Viscosity solutions: definitions}
\label{subsec:definitions}
Since we are dealing with a PIDE defined on an unusual domain and in order to make things precise we define the notion of viscosity solution in our framework. First we give the classical definition and then its counterparts based on semi jets.

\begin{definition}
\label{def:definition}~
\begin{itemize}
	\item[-]A locally bounded function $U\in USC(\mathcal{E}^K)$ (the set of upper semi-continuous functions on $\mathcal{E}^K$) is a viscosity sub-solution of $(\mathbf{HJB})_K$ if for all $x \in \mathcal{E}^K$ and test function $\phi$ such that $x$ is a maximum of $U-\phi$ we have
	\begin{equation*}
	F\big(x, \phi(x), \nabla \phi(x), \partial^2_{pp}\phi(x), D^KU(x) \big) \leq 0.
	\end{equation*}
	\item[-]A locally bounded function $U\in LSC(\mathcal{E}^K)$ (the set of lower semi-continuous functions on $\mathcal{E}^K$) is a viscosity super-solution of $(\mathbf{HJB})_K$ if for all $x \in \mathcal{E}^K$ and test function $\phi$ such that $x$  is a minimum of $U-\phi$ we have
	\begin{equation*}
	F\big( x, \phi( x), \nabla \phi( x), \partial^2_{pp}\phi(x), D^KU(x) \big) \geq 0.
	\end{equation*}
	\item[-]A continuous function $U$ defined on $\mathcal{E}^K$ is a viscosity solution of $(\mathbf{HJB})_K$ if it is a viscosity super-solution and a viscosity sub-solution.
\end{itemize}
\end{definition}

Note that in the above definition it is equivalent to consider local (or local strict) extrema. Also note that we have not replaced $U$ by $\phi$ for the last operator $D^K$. This is because $D^K U$ only requires finiteness of $U$ to be defined. Of course it is equivalent to replace $D^K U$ by $D^K \phi$ in Definition \ref{def:definition}. Indeed consider a sub-solution $U$ and a test function $\phi$ at point $x$. Since $D^K$ is a non local operator we can always build a sequence of test functions $(\phi_n)_{n\geq 0}$ satisfying $U \leq \phi_n$ with equality at point $x$ and such that
	$$
	(\nabla \phi_n(x) , \partial^2_{pp}\phi_n(x) ) = (\nabla \phi(x) , \partial^2_{pp}\phi(x) ) \text{ with } D^K\phi_n(x) \underset{n\rightarrow + \infty}{\rightarrow} D^KU(x).
	$$
By continuity of $F$ we get the equivalence. This also holds for super-solution.\\

We now introduce the notions of semi super and sub-jets in our framework. For $U$ in $USC(\mathcal{E}^K)$ and $x=(t, p, z) \in \mathcal{E}^K$, the super-jet of $U$ at point $x$ is the set
	\begin{align*}
	\mathcal{J}^{+}U(x) = &\{ (g, A, h )\in \mathbb{R}^2  \times\mathbb{R}  \times C^{0}(\mathcal{Z}^K_T),\text{ s.t. for any } y=(s, q, v) \in \mathcal{E}^K \text{ we have }\\
	& ~U(s, y) \leq U(t, x) +   g_1 (t -s)+ g_2(p-q) + \frac{1}{2}A(p-q)^2 + h(z-v)  +o(|t-s| + |p-q|^2)\\
	&~ \text{and } h(0) = 0   \}
	\end{align*}
	and the semi super-jet of $U$ at point $x$ is
	\begin{align*}
	\overline{\mathcal{J}}^{+}U(x) = &\{(g, A, h)\in \mathbb{R}^2 \times\mathbb{R} \times C^{0}(\mathcal{Z}^K_T)\text{ s.t. there exists a sequence}\\
	&~(x_n,  g_n, A_n, h_n)_{n\geq 0} \text{ with for any }n\geq0 ~ (g_n, A_n, h_n) \in \mathcal{J}^{+}U(x_n)\\
	&~ \text{ and such that} \big( x_n, U(x_n), g_n, A_n, h_n\big) \underset{n \rightarrow +\infty}{\rightarrow} \big(x, U( x), g, A, h\big) \}.
	\end{align*}
	In the above definition the convergence of $h_n$ is taken in the sense of locally uniform convergence at point $0$. By analogy we define the sub-jet $\mathcal{J}^{-}U(x)$ and the semi sub-jet $\overline{\mathcal{J}}^{-}U(x)$ for $U$ in $LSC(\mathcal{E}^K)$.\\
	
	 We can now give another characterization of viscosity sub and super-solutions relying on the notions of semi jets.
	\begin{definition}
	\label{def:equivalent_definition}~
	\begin{itemize}
		\item[-]A locally bounded function $U\in USC(\mathcal{E}^K)$ is a viscosity sub-solution of $(\mathbf{HJB})_K$ if for all $ x \in\mathcal{E}^K$, and $ (g, A, h) \in \overline{\mathcal{J}}^{+}U(x)$  we have
		\begin{equation*}
		F\big( x, U( x), g, A, D^KU(x) \big) \leq 0.
		\end{equation*}
		\item[-]A locally bounded function $U\in LSC(\mathcal{E}^K)$ is a viscosity super-solution of $(\mathbf{HJB})_K$ if for all $x \in\mathcal{E}^K$, and $ (g, A, h) \in \overline{\mathcal{J}}^{-}U(x)$  we have
		\begin{equation*}
		F\big( x, U(x), g, A, D^KU(x) \big)\geq 0.
		\end{equation*}
	\end{itemize}
	\end{definition}
	We show in Appendix \ref{appendix:section:equivalent_definition} that Definition \ref{def:definition} and \ref{def:equivalent_definition} are equivalent.\\
	
	In the next section based on the study of $(\mathbf{HJB})_K$ we prove that the control problem \eqref{eq:optimal_control_intro} admits a solution.

\subsection{Existence of an optimal control}
\label{subsec:existence_uniqueness}

In this section we prove existence of a solution to Problem \eqref{eq:optimal_control_pb}, which is equivalent to \eqref{eq:optimal_control_intro}. Before stating the result we give a sketch of the proof.\\

We start by proving uniqueness of a viscosity solution with polynomial growth to $(\mathbf{HJB})_K$ using a comparison result. The main difficulty is to adapt the Crandall-Ishi's lemma to our framework, which is done in Appendix \ref{appendix:crandall}. Using a verification argument we then check that the continuation utility function $U^K$ associated to Problem \eqref{eq:optimal_control_pb} is actually this unique solution. The maximizers of the Hamiltonian given in Equation \eqref{eq:maximizers} then naturally provide a control solving \eqref{eq:optimal_control_pb} and therefore Problem \eqref{eq:optimal_control_intro}. Full proof is given in Section \ref{proof:optimal_control}.
\begin{theorem}
\label{th:optimal_control}~
\begin{itemize}
	\item[(i)]There exists a unique viscosity solution $U^K$ with polynomial growth to $(\mathbf{HJB})_K$.
	\item[(ii)]This solution satisfies:
	$$
	U^K(0) = \underset{\delta \in \mathcal{A} }{\sup} 	~\mathbb{E}^{\delta}[G(i_T, P_T) e^{-rT} + \int_0^T e^{-rs}\Big(  g(i_s, P_s) + \delta^a_s \lambda^{a, \delta}_s+ \delta^b_s\lambda^{b, \delta}_s\Big)\mathrm{d}s].
	$$
	\item[(iii)]Problem \eqref{eq:optimal_control_intro} admits a solution given by 
	$$
	\delta^{K,*}_t = \delta^{K}(t, X^K_t), \text{ with }\delta^K = (\delta^K_a, \delta^K_b),
	$$
	where
	\begin{equation}
	\label{eq:feedback_control}
	\delta_a^K = \big(\sigma/k - D^K_aU^K\big)_{+} \text{ and }\delta_b^K = \big(\sigma/k - D^K_bU^K\big)_{+}.
	\end{equation}
\end{itemize}
\end{theorem}

It is important to note that in order to obtain existence of an admissible optimal control we have benefited from the fact that we are dealing with counting processes, whose infinitesimal generators are defined for any finite functions independently of their regularity. From a practical point of view Theorem \ref{th:optimal_control} implies that if we manage to compute $U^K$ we can implement the optimal control $\delta^{K,*}$ by monitoring the processes $\theta^a$ and $\theta^b$. Note that this is equivalent to monitor the list of arrival times of buy and sell market orders. However $\mathcal{E}^K$ is a subset of an infinite dimensional vector space. So we cannot compute $U^K$ using classic numerical methods. Therefore we need to find another way to approximate the control $\delta^{K,*}$. We deal with this issue in the next section.

\section{How to approximate the optimal control}
\label{sec:approach_optimal_control}

In this section we explain how to approximate numerically the feedback control $\delta^{K,*}$. We proceed in three steps: 
\begin{itemize}
\item[1.]We show that if $(K_n)_{n\geq 0}$ converges towards $K$ in $L^1$ and uniformly on $[0, T]$ then $(\delta^{K_n,*})_{n\geq 0}$ converges almost surely towards $\delta^{K,*}$.
\item[2.]We prove that when $K(t)= \overset{n}{\underset{i=1}{\sum}} \alpha_i e^{-\gamma_i t } $ there exists a Markovian representation of the model in dimension $2n + 2$.
\item[3.]Inspired by \cite{jaber2018lifting}, we show that for any completely monotone function $K$ in $L^1$ we can find a sequence $(K_n)_{n\geq 0}$ converging towards $K$ in $L^1$ and uniformly on $[0, T]$ such that for any $n$, $K_n$ is a linear combination of $n$ decreasing exponential functions.
\end{itemize}

Those three points give a simple method to compute an approximate version of the control $\delta^{K,*}$: choose $\tilde{K}$,  a sum of decreasing exponential functions, close enough to $K$. Use the finite dimensional representation to compute $U^{\tilde{K}}$ and implement $\delta^{\tilde{K}, *}$ instead of $\delta^{K,*}$.

\subsection{Convergence of solutions and optimal controls}
\label{subsec:convergence_control}
Consider a completely monotone function $K$ in $L^1$. We show that if a sequence of continuous $L^1$ functions $(K_n)_{n\geq 0}$ converges towards $K$ in $L^1$ and uniformly on $[0, T]$ then the sequence $(\delta^{K_n,*})_{n\geq 0}$ converges almost surely towards $\delta^{K,*}$.\\

From Theorem 5.8 in \cite{touzi2012optimal} we observe that the notion of viscosity solution is perfectly adapted to prove the convergence of solutions to a sequence of PIDEs. Hence we prove in Section \ref{proof:convergence_solution} the following result which is an extension of Theorem 5.8 in \cite{touzi2012optimal} to our framework.
\begin{property}
\label{prop:convergence_viscosity_solution}
Consider a sequence $(K_n)_{n\geq 0}$ of continuous $L^1$ functions converging towards a completely monotone function $K$ in $L^1$ and uniformly on $[0, T]$, then for any $x\in \mathcal{E}^K$ we have 
	\begin{equation}
	\label{eq:limit_utility}
 	U^K(x) ~ = \underset{(y, n)\in \bar{\mathcal{E}}   \rightarrow (x, +\infty) \in \bar{\mathcal{E}} }{\lim}  ~ U^{K_n}(y)	
	\end{equation}
	where
	$$
	\bar{\mathcal{E}} = \big(  \bigcup_{n \geq 0} \mathcal{E}^{K_n}\times \{n\} \big) \cup \big( \mathcal{E}^{K}\times \{+\infty\} \big).
	$$
\end{property}

The main technical difficulty in the proof of Proposition \ref{prop:convergence_viscosity_solution}, compared to Theorem 5.8 in \cite{touzi2012optimal}, is that the functions $(U^{K_n})_{n\geq 0}$ are defined on different domains.\\

We now consider a fixed sequence $(K_n)_{n\geq 0}$ of continuous $L^1$ functions converging towards $K$ in $L^1$ and uniformly on $[0, T]$. We have that almost surely
$$
\big( (t, X^{K_n}_t), n \big) \in \bar{\mathcal{E}} \rightarrow \big( (t, X^K_t), +\infty \big)\in \bar{\mathcal{E}}  \text{ when }n\rightarrow+\infty.
$$
From now on, when we consider a similar limit result we forget to write $ \bar{\mathcal{E}}$ to lighten notations.\\

We first recall that $\delta^{K,*}_t = (\delta^{K,a,*}_t, \delta^{K,b,*}_t)$ with
\begin{eqnarray*}
	\delta^{K,a,*}_t &= D^K_aU^K(t, X^K_t) = U^K(t, X^{K,+a}_t) - U^K(t, X^{K}_t),\\
	\delta^{K,b,*}_t &= D^K_bU^K(t, X^K_t) = U^K(t, X^{K,+b}_t) - U^K(t, X^{K}_t).
\end{eqnarray*}
where 
$$
	X^{K,+a}_t = \big( P_t, i+1, \theta^{K,a}_t+K(\cdot - t), \theta^{K,b}_t \big) \text{ and }
	X^{K,+b}_t = \big( P_t, i+1, \theta^{K,a}_t, \theta^{K,b}_t+K(\cdot - t)\big).
$$
Obviously we have the following almost sure convergences
$$
\big( (t, X^{K_n,+a}_t), n \big) \underset{n\rightarrow +\infty}{\rightarrow} \big( (t, X^{K,+a}_t), +\infty \big)  \text{ and } \big( (t, X^{K_n,+b}_t), n \big) \underset{n\rightarrow +\infty}{\rightarrow} \big( (t, X^{K,+b}_t), +\infty \big).
$$
So according to Proposition \ref{prop:convergence_viscosity_solution} we get that almost surely
$$
\big( U^{K_n}(t, X^{K_n}_t), U^{K_n}(t, X^{K_n,+a}_t), U^{K_n}(t, X^{K_n,+b}_t) \big) \underset{n\rightarrow +\infty}{\rightarrow} \big( U^K(t, X^{K}_t), U^K(t, X^{K,+a}_t), U^K(t, X^{K,+b}_t) \big).
$$
Hence we obtain the following result.
\begin{property}
	\label{prop:convergence_control}
	Consider a sequence $(K_n)_{n\geq 0}$ of continuous $L^1$ functions converging towards a completely monotone function $K$ in $L^1$ and uniformly on $[0, T]$, then for any $t$ we have almost surely
	$$
	\lim\limits_{n\rightarrow +\infty} \delta^{K_n,*}_t = \delta^{K,*}_t.
	$$
\end{property}

Proposition \ref{prop:convergence_control} perfectly fits our purpose of approximating $\delta^{K,*}$. Indeed suppose we manage to find a dense\footnote{Here dense is intended in the sense of convergence in $L^1$, together with uniform convergence on $[0, T]$.} subset of the completely monotone $L^1$ functions such that for any $K'$ in this subset, the control $\delta^{K',*}$ can be approximated numerically. Then Propositions \ref{prop:convergence_viscosity_solution} and \ref{prop:convergence_control} guarantee that for any completely monotone function $K$ in $L^1$ we can approximate numerically $U^K$ and $\delta^{K,*}$.\\

 We show in the next two sections that the set
$$
\mathcal{SE} = \bigcup_{n\geq 0}\{ \sum_{i=1}^n \alpha_i e^{-\gamma_i \cdot}\mathbf{1}_{\mathbb{R}_+} \text{ s.t. } \alpha \in \mathbb{R}_+^n \text{ and } \gamma \in \mathbb{R}_+^n\}
$$
satisfies those two conditions. Note that $\mathcal{SE}$ is simply the set of positive linear combinations of decreasing exponential functions. In the next two sections we study Problem \eqref{eq:optimal_control_pb} when the function $K$ is in $\mathcal{SE}$ and then show that $\mathcal{SE}$ is dense in the set of completely monotone functions in $L^1$.

\subsection{Solving the market maker's problem when $K\in \mathcal{SE}$}
\label{subsec:control_problem_exponential}
In this section we explain how to solve Problem \eqref{eq:optimal_control_pb} when the kernel function $K$ is in $\mathcal{SE}$.\\

We consider that the kernel of the Hawkes processes $N^a$ and $N^b$ is $$ K_{\alpha, \gamma}(t)=\sum_{i = 1}^n \alpha_i e^{-\gamma_i t}\mathbf{1}_{\mathbb{R}_+}(t), $$ where $n$ is a positive integer, $\alpha \in \mathbb{R}_{+}^n$ and $\gamma \in \mathbb{R}^n_+$. For $i\in\{1, \dots, n\}$ and $j = a$ or $b$ we define the process
$$
c^{j, i}_t = \int_{0}^t \alpha_ie^{-\gamma_i (t-s)}\mathrm{d}N^{j}_s.
$$
Then $Y^{\alpha, \gamma}_t = \big(t, P_t, i_t, (c_t^{a,i})_{1\leq i \leq n}, (c_t^{b,i})_{1\leq i \leq n} \big)$ is a Markovian process since for $j=a$ or $b$
$$
\lambda^{j,0}_t = \Phi(\sum_{i=1}^n c^{j,i}_t)\text{ and }\mathrm{d}c^{j, i}_t = -\gamma_i c^{j, i}_t \mathrm{d}t + \alpha_i \mathrm{d}N^j_t.
$$
The domain associated with this representation is $\mathcal{E}^{n} =  [0, T]\times \mathbb{R} \times \mathbb{Z}\times \mathbb{R}_+^n \times \mathbb{R}_+^n$, which is locally compact. As for $\mathcal{E}^K$, when we have $(t, x)\in \mathcal{E}^n$ we implicitly consider that $x = (p,i, c^a, c^b)$. Note that we can naturally go from the first representation to this one. More precisely we prove in Appendix \ref{appendix:existence_change_representation} that there exists a continuous function $R^{\alpha, \gamma}$ from $\mathcal{E}^{K_{\alpha, \gamma}}$ into $\mathcal{E}^{n}$ such that for any $t>0$ we have $R^{\alpha, \gamma}(t, X^{K_{\alpha, \gamma}}_t) =(t, Y^{\alpha, \gamma}_t)$. However notice that the second representation is somehow larger than the first one.\\

The infinitesimal generators associated to the processes $N^a$ and $N^b$ for the new representation are denoted by $D^{\alpha}_a$ and $D^{\alpha}_b$. They are defined for any function $U$ on $\mathcal{E}^{n}$ and $x\in \mathcal{E}^{n}$ by
\begin{align*}
D_{a}^{\alpha}U(x) &= U(t, p, i-1, c^a+\alpha, c^b ) - U(t, p, i, c^a, c^b ),\\
D_{b}^{\alpha}U(x) &= U(t, p, i+1, c^a, c^b+\alpha ) - U(t, p, i, c^a, c^b ).
\end{align*}
The HJB equation related to Problem \eqref{eq:optimal_control_pb} in this new representation is therefore
\begin{equation*}
(\mathbf{HJB})_{\alpha, \gamma}:~\left\{\begin{array}{ll}
& G_{\alpha, \gamma}\big(x, U(x), \nabla^c U(x),\nabla U( x), \partial^2_{pp}U(x), D^{\alpha}U(x) \big) = 0,\text{ for }x\in \mathcal{E}^{n},\\
& U(T, y) = G(i, p)\text{ for }(T, y)\in \mathcal{E}^{n}
\end{array}   \right.
\end{equation*}
with  $\nabla^c U =  ( \nabla^c_aU, \nabla^c_bU )$ where for $j=a$ or $b$, $ \nabla^c_jU = (\partial_{c^{j, i}}U)_{1\leq i \leq n}$, $ \nabla U(t, x) = \big(\partial_t U(t, x), \partial_pU(t, x) \big) $,
$$
D^{\alpha}U(t, x) = \big(D^{\alpha}_{a}U(t, x), D^{\alpha}_{b}U(t, x) \big)
$$
and where the function $G_{\alpha, \gamma}$ is defined for $(x, u, h, q, A, I)\in \mathcal{E}^{n} \times \mathbb{R} \times (\mathbb{R}^n)^2 \times \mathbb{R}^2\times \mathbb{R} \times \mathbb{R}^2$ by
\begin{align*}
G_{\alpha, \gamma}\big(x, u, q, h, A, I \big) =&~  ru - h_1 - d(t, p)h_2 - \frac{1}{2}\sigma^2 A - \langle \gamma, q_1 \rangle- \langle \gamma, q_2 \rangle -g(i, p) \\
& -\underset{\delta \in \mathbb{R}_+}{\sup}  \Phi(\sum_{i = 1}^n c^{a, i})e^{-\frac{k}{\sigma}\delta}(\delta + I_1) -\underset{\delta \in \mathbb{R}_+}{\sup}  \Phi(\sum_{i = 1}^n c^{b, i})e^{-\frac{k}{\sigma}\delta}(\delta + I_2).
\end{align*}

We easily adapt the proof of Theorem \ref{th:optimal_control} to $(\mathbf{HJB})_{\alpha, \gamma}$ and prove the following result.

\begin{theorem}
\label{th:verification_exp}~
\begin{itemize}
	\item[(i)]There exists a unique continuous viscosity solution with polynomial growth $U^{\alpha, \gamma}$ to $(\mathbf{HJB})_{\alpha, \gamma}$.
	\item[(ii)]The solution $U^{\alpha, \gamma}$ satisfies
	$$
	U^{\alpha, \gamma}(0) = \underset{\delta \in \mathcal{A}}{\sup}~ 	\mathbb{E}^{\delta}[G(i_T, P_T) e^{-rT} + \int_0^T e^{-rs}\big(  g(i_s, P_s) + \delta^a_s\lambda^{a, \delta}_s + \delta^b_s\lambda^{b, \delta}_s  \big) \mathrm{d}s].
	$$
	\item[(iii)]The stochastic control problem \eqref{eq:optimal_control_pb} admits a solution $\delta^{\alpha, \gamma,*}_t$ satisfying
	$$
	\delta^{\alpha, \gamma,*}_t = \delta^{\alpha, \gamma}(t, Y^{\alpha, \gamma}_t), \text{ with }\delta^{\alpha, \gamma} = (\delta^{\delta, \gamma}_a, \delta^{\delta, \gamma}_a)
	$$
	where 
	$$
	\delta_{a}^{\alpha, \gamma} = \big(\sigma/k - D^{\alpha}_aU^{\alpha, \gamma}\big)_{+} \text{ and } \delta_{b}^{\alpha, \gamma} = \big(\sigma/k - D^{\alpha}_bU^{\alpha, \gamma}\big)_{+}.
	$$
	\item[(iv)]We have $U^{K_{\alpha, \gamma}} = U^{\alpha, \gamma} \circ \mathcal{R}^{\alpha, \gamma}$.
\end{itemize}
\end{theorem}

The proof of the three first points is exactly the same as the proof of Theorem \ref{th:optimal_control}. We deal with point $(iv)$ in Section \ref{proof:verification_exp}. Points $(iii)$ and $(iv)$ of Theorem \ref{th:verification_exp} imply that for any $\alpha$ and $\gamma$ in $\mathbb{R}_+^n$ we can approximate numerically $\delta^{K_{\alpha, \gamma}}$. We just need to approximate $U^{\alpha, \gamma}$ using any numerical method, which is possible because the domain of $(\mathbf{HJB})_{\alpha, \gamma}$ is a subset of a finite dimensional vector space. Then using the change of variable $\mathcal{R}^{\alpha,\gamma}$ one gets
$$
\delta^{K_{\alpha, \gamma}} = \delta^{\alpha, \gamma} \circ \mathcal{R}^{\alpha, \gamma}.
$$
Note that this shows that the controls $\delta^{K_{\alpha, \gamma},*}$ given in Theorem \ref{th:optimal_control} (iii) and $\delta^{\alpha, \gamma, *}$ given in Theorem \ref{th:verification_exp} (iii) are actually the same.

\subsection{Density of $\mathcal{SE}$ in the set of completely monotone functions}
\label{subsec:density_function}

In this section we show that $\mathcal{SE}$ is dense in the set of completely monotone functions in $L^1$. Before giving the result we present a short sketch of the proof.\\

The key idea is that any completely monotone function can be written as the Laplace transform of a positive measure $m$, see Lemma 2.3 in \cite{merkle2014completely}:
\begin{equation}
\label{eq:completely_monotone_rep}
K(x) = \int_0^{+\infty} e^{-u x} m(\mathrm{d}u).
\end{equation}
Moreover if $K(0)<+\infty$ then $m$ is $L^1$ and if $K$ is in $L^1$ then $\int_{0}^{+\infty} \frac{m(\mathrm{d}u)}{u}<\infty$. Hence using Riemann sums to approximate the integral \eqref{eq:completely_monotone_rep} we get a natural way of approximating $K$ by a sequence of functions in $\mathcal{SE}$. Based on this idea we prove the following result in Appendix \ref{appendix:existence_converging_sequence}.
\begin{lemma}
\label{lemma:existence_converging_sequence}
	For any completely monotone function $K$ in $L^1$ there exists a sequence $(\alpha_n, \gamma_n)_{n\geq 0}$, such that:
	\begin{itemize}
		\item[$(i)$] For any $n$, $(\alpha_n, \gamma_n) \in \mathbb{R}^n_+ \times \mathbb{R}^n_+$,
		\item[$(ii)$]$(K_{\alpha_n, \gamma_n})_{n\geq 0}$ converges towards $K$ in $L^1$ and uniformly on every compact set of $\mathbb{R}_+$,
		\item[$(iii)$] $	\|K_{\alpha_n, \gamma_n}\|_{1} = \|K\|_1 \text{ and } K_{\alpha_n, \gamma_n}(0) = K(0).$
	\end{itemize}
\end{lemma}
	
Lemma \ref{lemma:existence_converging_sequence} together with Proposition \ref{prop:convergence_control} and Theorem \ref{th:verification_exp} allows us to conclude on the existence of a procedure to approximate $\delta^{K,*}$. In the next section we sum up our results and explain how to use them in practice.

\subsection{Conclusion on approximating the optimal control}
\label{subsec:conclusion_approach}

For a completely monotone function $K$ in $L^1$ consider $(\alpha_n, \gamma_n)_{n\geq 0}$ a sequence given by Lemma \ref{lemma:existence_converging_sequence}. We write $K_n$ instead of $K_{\alpha_n, \gamma_n}$ to lighten notations. According to Proposition \ref{prop:convergence_control}, we have the following almost sure convergence for any $t$ 
$$
\delta^{\alpha_n, \gamma_n}(t, Y^{\alpha_n, \gamma_n}_t) \underset{n\rightarrow +\infty}{\rightarrow} \delta_t^{K,*}.
$$
Hence to implement an approximated version of the optimal control $\delta^{K,*}$ we simply have to implement the control $\delta^{\alpha_n, \gamma_n,*}_t$ for $n$ large enough. This approximated control can be computed by solving numerically the finite dimensional PIDE $(\mathbf{HJB})_{\alpha, \beta}$.\\

In conclusion the recipee to implement an approximated version of the optimal control $\delta^{K,*}$ is the following:
\begin{itemize}
	\item[1.] Fix $n$ positive and find $\alpha, \gamma \in \mathbb{R}_+^n$ such that $K_{\alpha, \gamma}$ is the closest possible from $K$. See Appendix \ref{appendix:existence_converging_sequence} for a method to choose such $\alpha$ and $\gamma$.
	\item[2.] Approximate numerically $U^{\alpha, \gamma}$, the solution of $(\mathbf{HJB})_{\alpha, \gamma}$, which is equivalent to approximate numerically the feedback $\delta^{\alpha, \gamma}$.
	\item[3.] Monitor $Y^{\alpha, \gamma}$ and apply the control $\delta^{\alpha, \gamma}(t, Y^{\alpha, \gamma})$.	
\end{itemize}
The only flaw of this method is that the set $\mathcal{E}^{n}$ is a subset of a vector space of dimension $2n+2$. Hence when $n$ is larger than $2$ it is very unlikely that simple finite differences methods can be used to solve numerically $(\mathbf{HJB})_{\alpha, \gamma}$. To tackle this issue we need to use other numerical methods such as neural networks, see \cite{bachouch2018deep} for example, or probabilistic method, see \cite{henry2019branching}. In this article we propose to use the later method for numerical applications.

\section{Numerical applications}
\label{sec:numerical}

In this section we present some numerical experiments illustrating our results. We consider a simplified version of the market maker's problem:
$$
(N): \underset{\delta\in \mathcal{A}}{\sup}~\mathbb{E}^{\delta}[\int_0^{T} \delta^a_s\mathrm{d}N^a_s + \delta^b_s\mathrm{d}N^b_s -\mu i_s^2\mathrm{d}s].
$$
This corresponds to $G=0$ and $g(i, p) = -\mu i^2$. We take $k/\sigma = 20$ and $\mu = 0.1$. We note $U^K$ the unique viscosity solution (with polynomial growth) of the HJB equation associated to $(N)$ when the Hawkes processes' kernel is $K$. In this section we discard the price variable from the PIDEs since it does not appear in the optimization problem.\\

We first consider in Section \ref{subsec:low_dimension_method} the cases of kernels in $\mathcal{SE}$ with $n=2$ and show the importance for market makers to take into account the clustering and long memory properties of market order flows in their trading strategies . Then in Section \ref{subsec:high_dimension_method} we deal with more complex functions $K$ and illustrate the convergence of the method described in Section \ref{subsec:conclusion_approach}. In this last Section to solve the PIDEs we use the probabilistic representation introduced in \cite{henry2019branching} which is described in Appendix \ref{sec:branching_method}.

\subsection{The impact of taking into account the self exciting property of market order flows}
\label{subsec:low_dimension_method}

In this section we consider that $\Phi(x) = \mu + x$ for $\mu$ a positive constant and that the kernel $K$ is of the form:
$$
K(t) = \big( \alpha_1 e^{-\gamma_1 t}  + \alpha_2 e^{-\gamma_2 t} \big)\mathbf{1}_{\mathbb{R}_+}(t).
$$
This means that for $j=a$ or $b$
$$
\lambda^{j, 0}_t = \mu + \int_0^t K(t-s)\mathrm{d}N^j_s.
$$
In order to illustrate the interest of taking into account the clustering and long memory properties of market order flows we are going to compare three trading strategies corresponding to three controls $\delta^0,~ \delta^1$ and $\delta^2$. Each of those controls is computed in the following way:
\begin{itemize}
	\item[$(a)$]$\delta^0$ is the optimal control of a market maker believing that buy and sell market order flows are Poisson processes with intensity $\mu_0$.
	\item[$(b)$]$\delta^1$ is the optimal control of a market maker believing buy and sell market order flows are Hawkes processes with kernel $ K_1(t) = \alpha^1 e^{-\gamma^1 t}$ and that the value of $\mu$ is $\mu_1$.
	\item[$(c)$]$\delta^2$ is the optimal control $\delta^{K,*}$.
\end{itemize}
The first and second market makers are misleading on the dynamics of the market order flows so their strategies are suboptimal.\\

To compute the controls and the associated value functions we solve the corresponding HJB equations using finite differences methods. We use the following parameters settings:
\begin{itemize}
	\item[$(a)$]$\mu_0 = 1$,
	\item[$(b)$]$\mu_ 1 = 0.1 $, $\gamma^1 = 1$ and $\alpha^1 = 0.9$,
	\item[$(c)$]$\mu =  0.1$, $\gamma^2 = (1, 1)$ and $\alpha^2 = (0.45, 0.45)$.
\end{itemize}

We compare the different value functions associated to each controls in Figures \ref{fig:two_a}, \ref{fig:two_b} and \ref{fig:two_c}. As expected the control $\delta^2$ is optimal and $\delta^0$ is sub-optimal compared $\delta^1$. Moreover we observe in Figure \ref{fig:two_a} that considering an exponential kernel Hawkes model for the order flows leads to a $10\%$ gain compared with a strategy considering that market order flows is a Poisson process. Using two exponentials leads to another $10\%$ gain. This shows the large gain that can arise from taking into account the clustering and long memory properties of market order flows when designing a market making strategy.

\begin{figure}[tbph]
\centering
\includegraphics[width = \figurewidth, height = \figureheight]{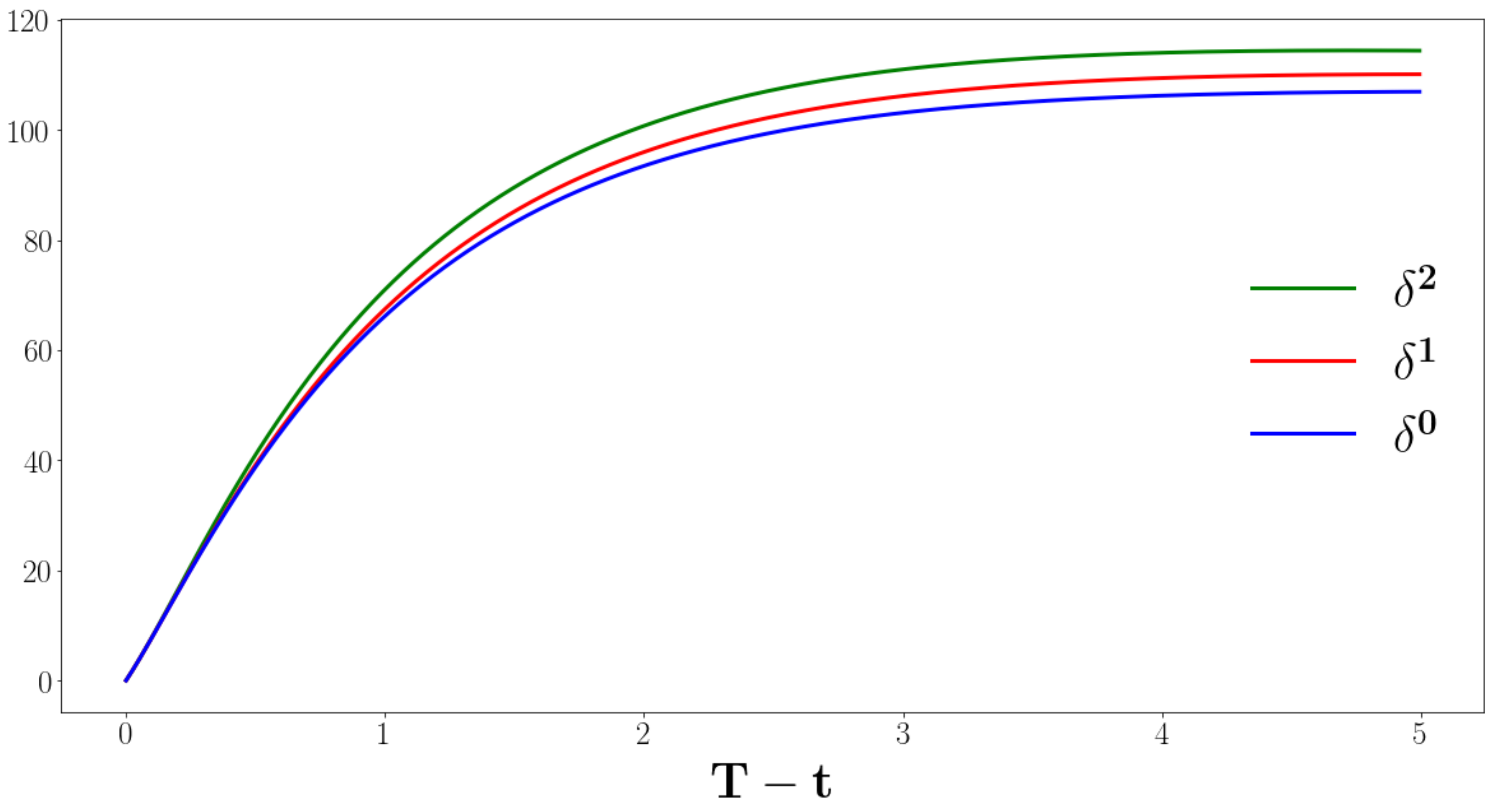}
	\caption{Value functions along the time for controls $\delta^0$, $\delta^1$ and $\delta^2$ with initial condition $c^a = (0, 10)$, $c^b = (0, 10)$ and $i=-10$.}
	\label{fig:two_a}
\end{figure} 

\begin{figure}[tbph]
\centering
\includegraphics[width = \figurewidth, height = \figureheight]{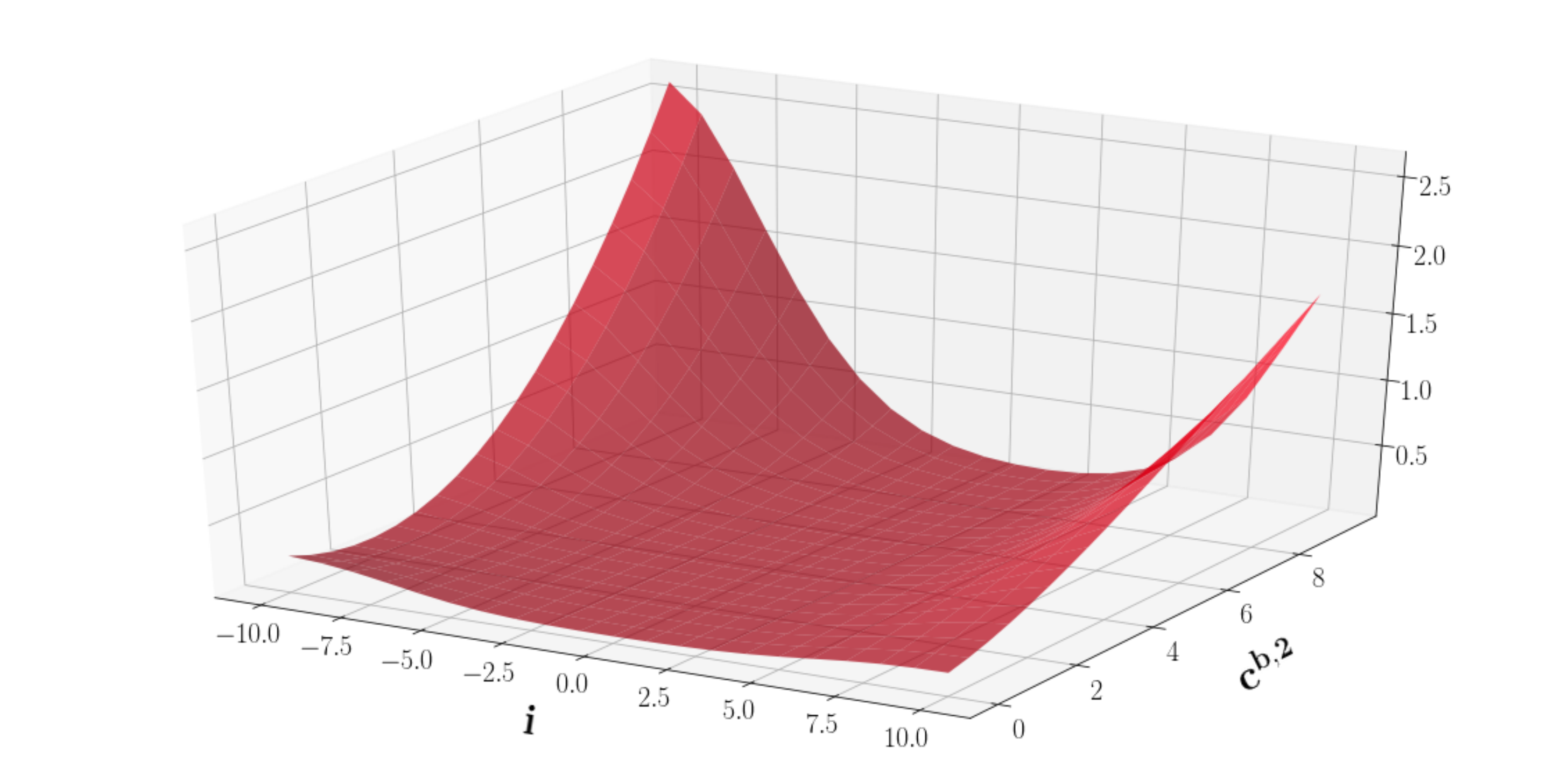}
	\caption{Difference between the value functions associated to controls $\delta^2$ and $\delta^1$ for $c^{a} = (10, 0)$, $c^{b, 1} = 10$.}
	\label{fig:two_b}
\end{figure} 

\begin{figure}[tbph]
\centering
\includegraphics[width = \figurewidth, height = \figureheight]{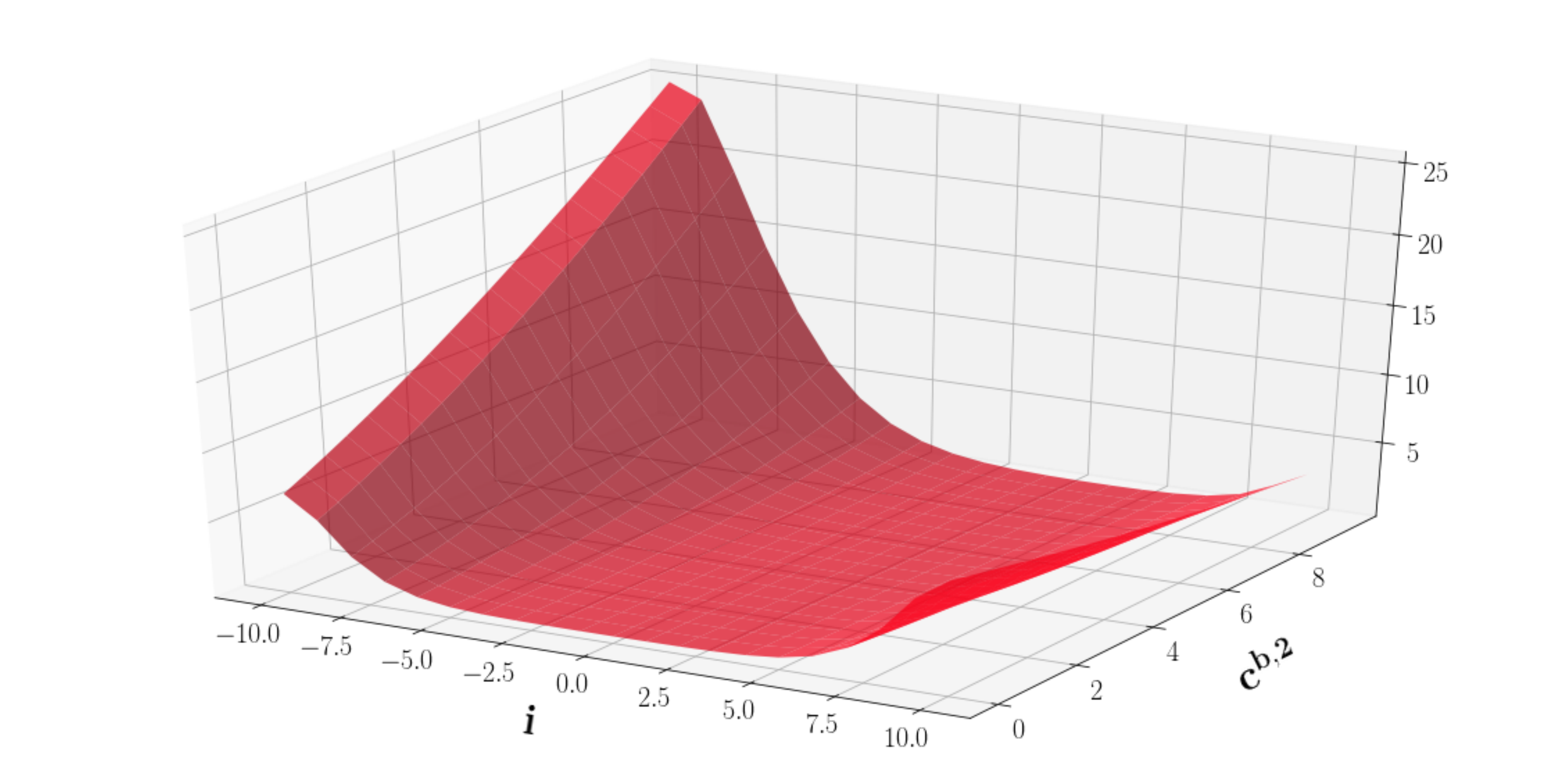}
	\caption{Difference between the value functions associated to controls $\delta^2$ and $\delta^0$ for $c^{a} = (10, 0)$, $c^{b, 1} = 10$.}
	\label{fig:two_c}
\end{figure}

\subsection{Numerical illustration of Section \ref{subsec:conclusion_approach}}
\label{subsec:high_dimension_method}

In this section we use the method presented in Section \ref{subsec:conclusion_approach} to estimate $U^K$ at several points when the function $K$ is the following completely monotone $L^1$ function:
$$
K(t) = \frac{\lambda}{\lambda + (t+\varepsilon)^{\alpha}}\frac{1}{(t+\varepsilon)^{\beta}} \mathbf{1}_{\mathbb{R}_+}(t),
$$
for $\lambda = 0.1$, $\alpha = 0.7$, $\beta=0.4$ and $\varepsilon = 0.01$ is a small shift used for numerical purposes.\\

We explain in Appendix \ref{appendix:large_dimension} how to build in this case the sequence $(\alpha_n,  \beta_n)_{n\geq 0}$ given by Lemma \ref{lemma:existence_converging_sequence}. As in Section \ref{subsec:conclusion_approach} we note $K_n = K^{\alpha_n, \beta_n}$.\\

For any $n$ we consider the following elements of $\mathcal{E}^{K_n}$
$$
x^n_0=(0, 0, 5K^n, 0),~x^n_1=(0, 5, 5K^n, 0) \text{ and }x^n_2=(0, -5, 5K^n, 0)
$$
and the following elements of $\mathcal{E}^K$
$$
x_0=(0, 0, 5K, 0),~x_1=(0, 5, 5K, 0) \text{ and }x^n_2=(0, 5, 5K, 0).
$$
According to Proposition \ref{prop:convergence_viscosity_solution} we have
$$
U_{K_n}(x^n_i) \underset{n \rightarrow +\infty}{\rightarrow} U_{K}(x_i)\text{ for any }i\in \{0,1,2\}.
$$ 
This convergence is illustrated in Figure \ref{fig:convergence}. We used the probabilistic representation of \cite{henry2019branching} to compute the $U^{K_n}(x^{n}_i)$ and $U^K(x_i)$, see Appendix \ref{sec:branching_method} for more details. This proves the tractability of the method presented in Section \ref{subsec:conclusion_approach}.

\begin{figure}[tbph]
\centering
\includegraphics[width = \figurewidth, height = \figureheight]{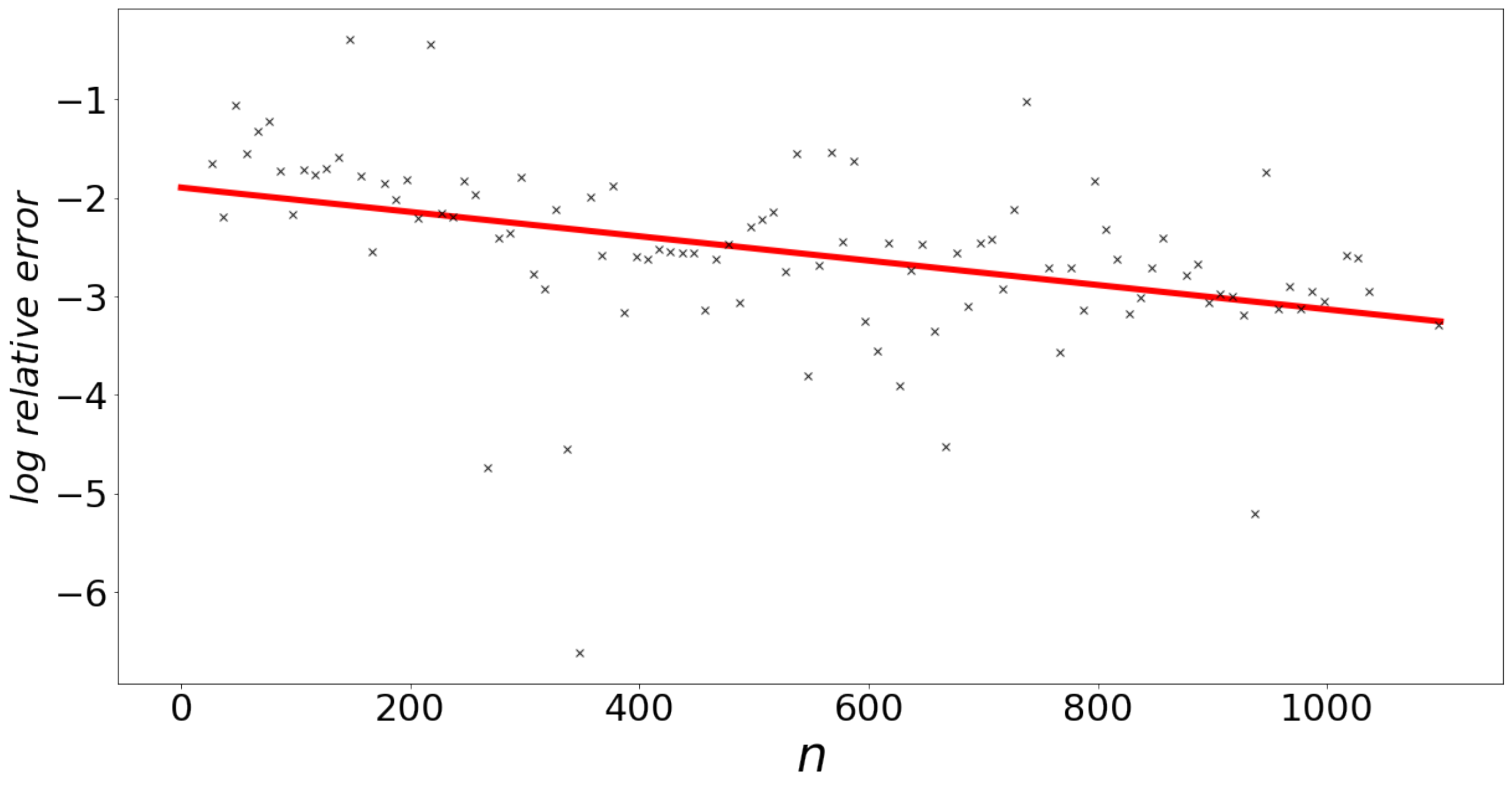}
	\caption{Log relative difference between $U^{K_{n}}(x^n_i)$ and $U^{K}(x_i)$ for $i\in \{0, 1, 2\}$. We used $T=1$ for the computations. In red is the linear regression.}.
\label{fig:convergence}
\end{figure}

\section*{Acknowledgments}
We thank Thibaut Mastrolia, Mathieu Rosenbaum and Nizar Touzi for many interesting discussions. The author gratefully acknowledges the financial support of the \textit{ERC Grant 679836 Staqamof} and of the chair \textit{Analytics and Models for Regulation}.

\newpage

\section{Proofs}
\label{sec:proofs}

In this section we drop the superscript $K$ for the processes $X^K$, $\theta^{K,a}$ and $\theta^{K,b}$ to lighten the notations.

\subsection{Formal definition of the probability space}
\label{sec:formal_definition}

In this section we make precise the probability space we are working on. In particular we give a proper definition to $\mathbb{E}^{\delta}$. First we define the canonical process and the probability space associated to our stochastic control problem.
\begin{itemize}
	\item[-] Consider $\Omega_d$ the set of increasing piecewise constant c\`adl\`ag functions from $[0, T]$ into $\mathbb{N}$ with jumps equal to $1$ and $\Omega_p$ the set of continuous functions from $[0, T]$ into $\mathbb{R}$. We define $\Omega = \Omega_p \times \Omega_d^2$.
	\item[-] We let $(W_t, N^a_t, N^b_t)_{t\in [0, T]}$ be the canonical process on $\Omega$.
	\item[-] The associated filtration is $\mathbb{F} = ( \mathcal{F}^p_t \otimes \mathcal{F}^d_t\otimes \mathcal{F}^d_t)_{t\in [0, T]}$ where $ (\mathcal{F}^d_t)_{t\in [0, T]}$ (resp. $ (\mathcal{F}^p_t)_{t\in [0, T]}$) is the right continuous completed filtration associated with $N^a$ (or $N^b$) (resp. $W$).
	\item[-]We denote by $\mathbb{P}_0$ the probability measure on $(\Omega, \mathbb{F})$ such that $\big(M^a_s = N^a_s - s\lambda_0,~M^b_s = N^b_s  - s \lambda_0\big)_{s\in [0, T]}$, for $\lambda_0>0$, are local martingales and $(W_s)_{s\in [0, T]}$ is a Brownian motion .
\end{itemize}  

We now introduce some processes we use later in this paper. For a fixed $(t, x) \in \mathcal{E}^K$ we define $X^{t, x} = (P^{t, x}, i^{t,x}, \theta^{t,x;a}, \theta^{t,x;b})$ that is the state of the system after time $t$ when starting from point $(t, x)$. The dynamics of $X^{t, x}$ is given on $[t, T]$ by
\begin{align*}
&\mathrm{d}P^{t, x}_s =d(s, P_s^{t, x})\mathrm{d}s + \sigma\mathrm{d}W_s ,~~ P^{t, x}_t = p,\\ 
&\mathrm{d}i^{t, x}_s = \mathrm{d}N^a_s - \mathrm{d}N^b_s,~~ i^{t, x}_t = i,\\
&\mathrm{d}\theta^{t, x;a}_s =K(\cdot - s) \mathrm{d}N^a_s ,~~ \theta^{t,x;a}_t = \theta^a,\\
&\mathrm{d}\theta^{t, x;b}_s =K(\cdot - s) \mathrm{d}N^b_s ,~~ \theta^{t,x;b}_t = \theta^b.
\end{align*}
Using those processes we explicit the change of measure associated to each control process. For this we consider the functions
$$
\lambda^{a}(t,x, \delta) = e^{-\frac{k}{\sigma}\delta^a}\Phi\big(\theta^a(t)\big)\text{ and }\lambda^{b}(t,x, \delta) = e^{-\frac{k}{\sigma}\delta^b}\Phi\big(\theta^b(t)\big),
$$
that represent the ask and bid intensity in the state $(t, x)\in \mathcal{E}^K$ when the control is $\delta$. For any $\delta\in \mathcal{A}$ we define $\mathbb{P}^{t, x;\delta}$ by
$$
\frac{\mathrm{d}\mathbb{P}^{t,x;\delta}}{\mathrm{d} \mathbb{P}_0} = L_T^{t,x; \delta}
$$
where $L_T^{t, x;\delta}$ is the Dol\'eans-Dade exponential of 
$$
Z^{t, x;\delta}_s = \int_{0}^s \frac{\lambda^a(s, X^{t, x}_s, \delta_s)-\lambda_0}{\lambda_0}\mathbf{1}_{s\geq t}\mathrm{d}M^a_s + \frac{\lambda^b(s, X^{t, x}_s, \delta_s)-\lambda_0}{\lambda_0}\mathbf{1}_{s\geq t} \mathrm{d}M^b_s.
$$
Since $\lambda^a(t, x, \delta)\leq C(1 + \|x\|)$ and $\lambda^b(t, x, \delta)\leq C(1 + \|x\|)$, by the Corrolary 2.6 in \cite{sokol2013optimal}, for any $(t, x)\in \mathcal{E}^K$, $(L_s^{t, x;\delta})_{s\in [t, T]}$ is a true $\mathbb{P}_0$ martingale. Moreover by Theorem III-3.11 in \cite{jacod2013limit} the processes
$$
M^{t, x;a, \delta} = N^a - \int_t^{\cdot} \lambda^a(u, \delta_u, X^{t, x}_u) \mathrm{d}u \text{ and }M^{t, x;b, \delta} = N^b - \int_t^{\cdot} \lambda^b(u, \delta_u, X^{t, x}_u) \mathrm{d}u
$$
are $\mathbb{P}^{t,x;\delta}$-local martingales on $[t, T]$. Actually they are true martingales, see Appendix \ref{appendix:subsec_martingale_hawkes}.\\

For $(t, x)\in \mathcal{E}^K$ and $\delta\in \mathcal{A}$ we note $\mathbb{E}^{\delta}_{t, x}$ the expectation under the law $\mathbb{P}^{t,x;\delta}$ and note $\mathbb{E}^{\delta}$ instead of $\mathbb{E}^{\delta}_{0, 0}$.\\

Finally, for any $F$ bounded continuous function, $\delta\in \mathcal{A} $ and $\theta$ stopping time with values in $[t, T]$ we have:
\begin{equation}
\label{eq:markovian_property}
\mathbb{E}^{\delta}_{t, x}[F(X^{t, x}_T)|\mathcal{F}_{\theta}] = \mathbb{E}^{\delta^{\theta}}_{\theta, X^{t, x}_{\theta}}[F(X^{\theta, X^{t, x}_{\theta}}_T)]
\end{equation}
where, $\delta^{\theta}$ is the restriction to $[\theta, T]$ of $\delta$. This prove that for any $(t, x)\in \mathcal{E}^K$ the process $(s, X^{t, x}_s)_{s\geq t}$ is Markovian.

\subsection{Proof of Theorem \ref{th:optimal_control}}
\label{proof:optimal_control}

We proceed in 5 steps.
\begin{enumerate}
\item Section \ref{proof:comparison}: Using a comparison result we show that $(\mathbf{HJB})_{K}$ admits a unique viscosity solution with polynomial growth.
\item Section \ref{proof:utility}: For any $K$ we define $U^K$ the continuation utility function associated to \eqref{eq:optimal_control_pb}.
\item Section \ref{proof:ppd}: We prove a dynamic programming principle for $U^K$.
\item Section \ref{proof:verification}: Using a verification argument we show that $U^K$ is the unique viscosity solution (with polynomial growth) of $(\mathbf{HJB})_K$.
\item Section \ref{proof:control}: We show that the control given in Equation \eqref{eq:feedback_control} solves the control problem \eqref{eq:optimal_control_pb}.
\end{enumerate}

\subsubsection{Comparison result for $(\mathbf{HJB})_K$}
\label{proof:comparison}

We start by proving a comparison result for bounded solutions, then we extend it to functions with polynomial growth.
\begin{property}
\label{prop:comparison_bounded}
Let $U\in USC(\mathcal{E}^K)$ be a bounded from above viscosity sub-solution of $(\mathbf{HJB})_K$ and $V\in LSC(\mathcal{E}^K)$ be a bounded from below viscosity super-solution of $(\mathbf{HJB})_K$ such that $U(T, \cdot)\leq V(T, \cdot)$ then
$$
U\leq V \text{ on } \mathcal{E}^K.
$$
\end{property}
\begin{proof}
	We suppose that there exists some $(t_0, x_0)\in \mathcal{E}^K$ such that 
	$$
	U(t_0, x_0) - V(t_0, x_0) = \delta >0.
	$$
	By hypothesis, necessarily $t_0\in [0, T)$. We show that this implies a contradiction. We consider the following quantities
	$$
	N_{\varepsilon} = \underset{(t,x) \in \mathcal{E}^K}{\sup}~ U(t,x) - V(t,x)  - 2\varepsilon\|x\|^2 
	$$
	and
	$$
N^{\alpha}_{\varepsilon} =  \underset{(t,x),(t, y)\in \mathcal{E}^K}{\sup} U(t,x) - V(t, y) -  \varepsilon(\|x\|^{2} + \|y\|^2) - \alpha \|x-y\|^2.
	$$
	The functions $U$ and $-V$ being bounded from above we have
	$$
	\underset{\|x\| + \|y\| \rightarrow +\infty}{\lim}~U(t, x) -V(t, y) - \alpha \|x-y\|^2 - \varepsilon\|x\|^2 - \varepsilon\|y\|^2 = -\infty
	$$
	uniformly in $t$. Thus we can restrict the supremums to bounded sets that depends only on $\varepsilon$. More precisely
	\begin{eqnarray}
		N_{\varepsilon} &=& \underset{(t,x)\in \mathcal{E}^K_{R}}{\sup}~ U(t,x) - V(t,x) - 2\varepsilon\|x\|^2\\
		N^{\alpha}_{\varepsilon} &=& \underset{(t,x),(t,y)\in  \mathcal{E}^K_R}{\sup}~ U(t,x) - V(t,y) - \varepsilon(\|x\|^2 + \|y\|^2) -  \alpha \|x-y\|^2
	\end{eqnarray}
	where $R$ only depends on $\varepsilon$. Since the set $\mathcal{E}^K_{R}$ is compact the supremum $N_{\varepsilon}^{\alpha}$ is achieved at some $(t_{\varepsilon}^{\alpha}, x_{\varepsilon}^{\alpha}, y_{\varepsilon}^{\alpha})$. We show at the end of the proof that when $\alpha \rightarrow + \infty$, up to a subsequence, we have
	\begin{equation}
	\label{eq:statement_c}
	\underset{\alpha \rightarrow +\infty}{\lim}(t_{\varepsilon}^{\alpha},x_{\varepsilon}^{\alpha}, y_{\varepsilon}^{\alpha}) = (t_{\varepsilon}, x_{\varepsilon}, x_{\varepsilon})	
	\end{equation}
	where $(t_{\varepsilon}, x_{\varepsilon})$ achieves the supremum $N_{\varepsilon}$. We also prove that
	\begin{equation}
	\label{eq:statement_a}
	\underset{\alpha\rightarrow + \infty}{\lim}~\alpha  \|x_{\varepsilon}^{\alpha}-y_{\varepsilon}^{\alpha}\|^2 = 0 \text{, }\underset{\alpha \rightarrow 0}{\lim}~N_{\varepsilon}^{\alpha} = N_{\varepsilon} ,~\underset{\varepsilon\rightarrow 0}{\lim}~\varepsilon\|x_{\varepsilon}\|^2 = 0
	\end{equation}
	and that
	\begin{equation}
	\label{eq:statement_b}
	\underset{\varepsilon\rightarrow 0}{\lim}~N_{\varepsilon} = N = \underset{(t,x)\in  \mathcal{E}^K}{\sup}~ U(t,x) - V(t,x).	
	\end{equation}
	A consequence of Equation \eqref{eq:statement_b} is that
	\begin{equation}
	\label{eq:proof_comparison_a}
	\underset{\alpha \rightarrow +\infty}{\lim }\big( U(t^{\alpha}_{\varepsilon}, x^{\alpha}_{\varepsilon}), V(t^{\alpha}_{\varepsilon}, x^{\alpha}_{\varepsilon}) \big) = \big( U(t_{\varepsilon}, x_{\varepsilon}), V(t_{\varepsilon}, x_{\varepsilon})\big).
	\end{equation}
	We use the notations $x_{\varepsilon}^{\alpha} = ( P_{\varepsilon}^{\alpha}, i_{\varepsilon}^{\alpha}, \theta_{\varepsilon}^{a, \alpha}, \theta_{\varepsilon}^{b, \alpha})\text{ and }	y_{\varepsilon}^{\alpha} = ( Q_{\varepsilon}^{\alpha}, j_{\varepsilon}^{\alpha}, \beta_{\varepsilon}^{a, \alpha}, \beta_{\varepsilon}^{b, \alpha})$. \\

	With respect to Lemma \ref{lemma:crandall_ishi}, which is an adaptation of the Crandall-Ishi's lemma to our framework, for any $\beta>0$ there exists $\big((\lambda^{\alpha}_{\varepsilon}, p^{\alpha}_{\varepsilon}), A^{\beta, \alpha}_{\varepsilon}, h) \in \overline{\mathcal{J}}^{+}U(t^{\alpha}_{\varepsilon}, x^{\alpha}_{\varepsilon} \big) $ and $\big((\hat{\lambda}^{\alpha}_{\varepsilon}, q^{\alpha}_{\varepsilon}), B^{\beta, \alpha}_{\varepsilon}, g) \in \overline{\mathcal{J}}^{-}V(t^{\alpha}_{\varepsilon}, y^{\alpha}_{\varepsilon}\big)$ such that 
	$$
	-(\beta^{-1} + 2\varepsilon + 4\alpha )I_2 \leq  \begin{pmatrix}A^{\beta, \alpha}_{\varepsilon} & 0\\ 0& -B^{\beta, \alpha}_{\varepsilon}\end{pmatrix} \leq (2\varepsilon+\beta 4\varepsilon^2)I_2  + \big(2\alpha + 8\beta  (\alpha\varepsilon + \alpha^2)\big) \begin{pmatrix}
	1 & -1 \\ -1 & 1
	\end{pmatrix}.
	$$
	with
	$$
	p^{\alpha}_{\varepsilon} =  2\varepsilon P^{\alpha}_{\varepsilon} + 2\alpha(P^{\alpha}_{\varepsilon} -Q^{\alpha}_{\varepsilon}),~q^{\alpha}_{\varepsilon} =  -2\varepsilon Q^{\alpha}_{\varepsilon} - 2\alpha(Q^{\alpha}_{\varepsilon} - P^{\alpha}_{\varepsilon}),~ \lambda^{\alpha}_{\varepsilon} = 0\text{ and }\hat{\lambda}^{\alpha}_{\varepsilon} = 0.
	$$
	Remark that for $\varepsilon$ small enough
	$$
U(t^{\alpha}_{\varepsilon}, x^{\alpha}_{\varepsilon}) - V(t^{\alpha}_{\varepsilon}, y^{\alpha}_{\varepsilon})\geq \delta - \varepsilon\|x_0\|^2> \frac{\delta}{2}
	$$
	We now walk towards a contradiction by showing that
	$$
	\underset{\varepsilon \rightarrow 0}{\lim \sup}~\underset{\alpha\rightarrow + \infty}{\lim \sup}~U(t^{\alpha}_{\varepsilon}, x^{\alpha}_{\varepsilon}) - V(t^{\alpha}_{\varepsilon}, y^{\alpha}_{\varepsilon}) \leq 0.
	$$
	According to the definition of sub-solution and super-solution we have
	$$
	F\big(t_{\varepsilon}^{\alpha},x_{\varepsilon}^{\alpha}, U(t_{\varepsilon}^{\alpha}, x_{\varepsilon}^{\alpha}),(\lambda^{\alpha}_{\varepsilon}, p^{\alpha}_{\varepsilon}), A^{\beta, \alpha}_{\varepsilon}, D^KU(t_{\varepsilon}^{\alpha},x_{\varepsilon}^{\alpha})\big) \leq 0
	$$
	and
	$$
	F\big(t_{\varepsilon}^{\alpha},y_{\varepsilon}^{\alpha}, V(t_{\varepsilon}^{\alpha},y_{\varepsilon}^{\alpha}),(\hat{\lambda}^{\alpha}_{\varepsilon}, q^{\alpha}_{\varepsilon}) , B^{\beta, \alpha}_{\varepsilon}, D^KV( t_{\varepsilon}^{\alpha},y_{\varepsilon}^{\alpha})\big) \geq 0.
	$$
	By definition of $F$:
	\begin{align*}
	r\big( U(t_{\varepsilon}^{\alpha}, x^{\alpha}_{\varepsilon}) - V(t_{\varepsilon}^{\alpha}, y^{\alpha}_{\varepsilon})\big) &\leq F\big(t_{\varepsilon}^{\alpha}, x^{\alpha}_{\varepsilon}, U(t_{\varepsilon}^{\alpha}, x^{\alpha}_{\varepsilon}), (\lambda^{\alpha}_{\varepsilon}, p^{\alpha}_{\varepsilon}), A^{\beta, \alpha}_{\varepsilon}, D^KU(t_{\varepsilon}^{\alpha}, x^{\alpha}_{\varepsilon})\big)\\
	&~~~~~~ - F\big(t_{\varepsilon}^{\alpha},x^{\alpha}_{\varepsilon}, V(t_{\varepsilon}^{\alpha},y^{\alpha}_{\varepsilon}), (\lambda^{\alpha}_{\varepsilon}, p^{\alpha}_{\varepsilon}), A^{\beta, \alpha}_{\varepsilon}, D^KU(t_{\varepsilon}^{\alpha},x^{\alpha}_{\varepsilon})\big),	
	\end{align*}
	thus
	\begin{align*}
	r\big(U(t_{\varepsilon}^{\alpha},x^{\alpha}_{\varepsilon}) - V(t_{\varepsilon}^{\alpha},y^{\alpha}_{\varepsilon})\big) &\leq F(t_{\varepsilon}^{\alpha},y^{\alpha}_{\varepsilon}, V(t_{\varepsilon}^{\alpha}, y^{\alpha}_{\varepsilon}), \hat{\lambda}^{\alpha}_{\varepsilon}, q^{\alpha}_{\varepsilon}, B^{\beta, \alpha}_{\varepsilon}, D^KV(t_{\varepsilon}^{\alpha},y^{\alpha}_{\varepsilon}))\\
	&~~~~~~ - F(t_{\varepsilon}^{\alpha}, x^{\alpha}_{\varepsilon}, V(t_{\varepsilon}^{\alpha},y^{\alpha}_{\varepsilon}), \lambda^{\alpha}_{\varepsilon}, p^{\alpha}_{\varepsilon}, A^{\beta, \alpha}_{\varepsilon}, D^KU(t_{\varepsilon}^{\alpha},x^{\alpha}_{\varepsilon}))\\
	&\leq d(t_{\varepsilon}^{\alpha}, P^{\alpha}_{\varepsilon})p^{\alpha}_{\varepsilon} -  d(t_{\varepsilon}^{\alpha}, Q^{\alpha}_{\varepsilon})q^{\alpha}_{\varepsilon} + \frac{1}{2} \sigma^2 A^{\alpha}_{\varepsilon} - \frac{1}{2} \sigma^2B^{\alpha}_{\varepsilon}\\
	&~~~~~~ +  H(t^{\alpha}_{\varepsilon},y^{\alpha}_{\varepsilon},D^KV(t^{\alpha}_{\varepsilon}, y^{\alpha}_{\varepsilon})) - H(t^{\alpha}_{\varepsilon},x^{\alpha}_{\varepsilon},D^KU(t^{\alpha}_{\varepsilon}, x^{\alpha}_{\varepsilon}))
	\end{align*}
	where
	$$
	H(t,x, I) = - \underset{\delta\in \mathbb{R}_+}{\sup} \Phi\big(\theta^a(t)\big)e^{-\frac{\sigma}{k}\delta}(\delta + I_1) - \underset{\delta \in \mathbb{R}_+}{\sup}  \Phi\big(\theta^b(t)\big)e^{-\frac{\sigma}{k}\delta}(\delta + I_2) + g(i, p).
	$$
	Note that the function $H$ is Lipschitz continuous. Taking $\beta = \alpha^{-1}$	we get  	
	$$
	\sigma^{2} A^{\beta, \alpha}_{\varepsilon} - \sigma^2 B^{\beta, \alpha}_{\varepsilon} \leq2 (2\varepsilon+\alpha^{-1} 4\varepsilon^2)\sigma^{2} .
	$$
	The RHS can be taken arbitrarly small when $\alpha \rightarrow + \infty$  and $\varepsilon\rightarrow 0$ by Equation \eqref{eq:statement_a}. Using the Lipschitz property of $d$ we have
	\begin{align*}
	d(t_{\varepsilon}^{\alpha}, P^{\alpha}_{\varepsilon})p^{\alpha}_{\varepsilon} - d(t_{\varepsilon}^{\alpha}, Q^{\alpha}_{\varepsilon})q^{\alpha}_{\varepsilon}  &= 2\varepsilon \big( d(t_{\varepsilon}^{\alpha}, P^{\alpha}_{\varepsilon})P^{\alpha}_{\varepsilon} + d(t_{\varepsilon}^{\alpha}, Q^{\alpha}_{\varepsilon})Q^{\alpha}_{\varepsilon}\big)+ 2\alpha (P^{\alpha}_{\varepsilon}-Q^{\alpha}_{\varepsilon})\big( d(t_{\varepsilon}^{\alpha}, P^{\alpha}_{\varepsilon}) - d(t_{\varepsilon}^{\alpha}, q^{\alpha}_{\varepsilon}) \big)\\
	&\leq 2\varepsilon C(1 + \|y^{\alpha}_{\varepsilon}\|^2 + \|x^{\alpha}_{\varepsilon}\|^2) + C\alpha \|x^{\alpha}_{\varepsilon} - y^{\alpha}_{\varepsilon}\|^2.
	\end{align*}
		Here the RHS goes to zero when $\alpha \rightarrow + \infty$ and $\varepsilon\rightarrow 0$ due to Equation \eqref{eq:statement_a}. Finally, by Equation \eqref{eq:proof_comparison_a} and since $U$ (resp. $V$) is a USC (resp. LSC) function and $H$ is continuous and decreasing with respect to its last variable we have 
	$$
	\underset{\alpha \rightarrow + \infty}{\lim \sup }~H(t_{\varepsilon}^{\alpha}, y_{\varepsilon}^{\alpha}, D^KV(t_{\varepsilon}^{\alpha}, y_{\varepsilon}^{\alpha})) - H(t_{\varepsilon}^{\alpha}, x_{\varepsilon}^{\alpha}, D^KU(t_{\varepsilon}^{\alpha}, x_{\varepsilon}^{\alpha}))  \leq H\big(t_{\varepsilon}, x_{\varepsilon}, D^KV(t_{\varepsilon}, x_{\varepsilon})  \big) - H\big(t_{\varepsilon}, x_{\varepsilon}, D^KU(t_{\varepsilon}, x_{\varepsilon}) \big).
	$$
Remark that for any $z$ such that $(t_{\varepsilon}, x_{\varepsilon} + z)\in \mathcal{E}^K$ we have by definition of $(t_{\varepsilon}, x_{\varepsilon})$
	$$
	U(t_{\varepsilon},x_{\varepsilon}) - V(t_{\varepsilon}, x_{\varepsilon}) -2\varepsilon\|x_{\varepsilon}\|^2  \geq 	U(t_{\varepsilon},x_{\varepsilon}+z) - V(t_{\varepsilon}, x_{\varepsilon}+z)  -2\varepsilon\|x_{\varepsilon} + z\|^2.
	$$
	Consequently we have
	$$	V(t_{\varepsilon},x_{\varepsilon}+z) - V(t_{\varepsilon}, x_{\varepsilon}) \geq U(t_{\varepsilon}, x_{\varepsilon}+z) - U(t_{\varepsilon}, x_{\varepsilon}) -2 \varepsilon\big( \|x_{\varepsilon} + z\|^2 - \|x_{\varepsilon}\|^2  \big)
	$$
	and so
	$$
	D^KV(t_{\varepsilon}, x_{\varepsilon}) \leq D^K(U - 2\varepsilon\|\cdot\|)(t_{\varepsilon}, x_{\varepsilon}).
	$$
	The monotony and Lipschitz regularity of $H$ implies
	\begin{align*}
	\underset{\alpha \rightarrow + \infty}{\lim \sup }~&H(t_{\varepsilon}^{\alpha}, y_{\varepsilon}^{\alpha}, D^KV(t_{\varepsilon}^{\alpha}, y_{\varepsilon}^{\alpha})) - H(t_{\varepsilon}^{\alpha}, x_{\varepsilon}^{\alpha}, D^KU(t_{\varepsilon}^{\alpha}, x_{\varepsilon}^{\alpha}))\\
	 &~~~\leq H\big(x_{\varepsilon}, D^K(U - 2\varepsilon\|\cdot\|)(t_{\varepsilon}, x_{\varepsilon}) \big) - H\big(x_{\varepsilon}, D^KU(t_{\varepsilon}, x_{\varepsilon}) \big)\\
	&~~~\leq C \varepsilon \|x_{\varepsilon}\| \Big\| D^K\|\cdot \|^2(t_{\varepsilon},x_{\varepsilon})\Big\|.
	\end{align*}

	Notice that for any $x\in \mathcal{E}^K$
	$$
	\begin{array}{ll}
	D^K \|\cdot\|^2 (x) & = \begin{pmatrix}
	\|\theta^a + K(\cdot - t)\|^2_1 - \|\theta^a\|^2_1 + |i + 1|^2 - |i|^2\\ \|\theta^b + K(\cdot - t)\|^2_1 - \|\theta^b\|^2_1 + |i - 1|^2 - |i|^2  
	\end{pmatrix}\\
	& = \begin{pmatrix}
	\|K(\cdot - t)\|^2_1 + 2 \|K(\cdot - t)\|_1\|\theta^a\|_1 +  1 + 2i \\ \|K(\cdot - t)\|^2_1 + 2 \|K(\cdot - t)\|_1\|\theta^b\|_1 +  1 - 2i  
	\end{pmatrix},
	\end{array}
	$$
	thus there exists $C>0$ such that $ \big{\|} D^K \|\cdot\|^2 (t, x) \big{\|} \leq C (1 + \|x\|)$. Consequently we get
	\begin{align*}
	\underset{\alpha \rightarrow + \infty}{\lim \sup }~H(t_{\varepsilon}^{\alpha}, y_{\varepsilon}^{\alpha}, D^KV(t_{\varepsilon}^{\alpha}, y_{\varepsilon}^{\alpha})) - H(t_{\varepsilon}^{\alpha}, x_{\varepsilon}^{\alpha}, D^KU(t_{\varepsilon}^{\alpha}, x_{\varepsilon}^{\alpha})) \leq C\varepsilon (1 + \|x_{\varepsilon}\|^2)  	\end{align*}
	that goes to zero when taking the limit $\varepsilon\rightarrow 0$. Finally we have shown that
	\begin{equation*}	\underset{\varepsilon\rightarrow 0}{\lim \sup}~ \underset{\alpha \rightarrow + \infty}{\lim \sup}~ U(t^{\alpha}_{\varepsilon},x^{\alpha}_{\varepsilon}) - V(t^{\alpha}_{\varepsilon},y^{\alpha}_{\varepsilon})\leq 0,
	\end{equation*}
	which is a contradiction.\\
	
	We finally prove the statements \eqref{eq:statement_c}, \eqref{eq:statement_a} and \eqref{eq:statement_b}. We consider $(t_{\varepsilon}, x_{\varepsilon}, y_{\varepsilon})\in \overline{(t^{\alpha}_{\varepsilon}, x^{\alpha}_{\varepsilon}, y^{\alpha}_{\varepsilon})_{\alpha \geq 0}}$ that exists since $\mathcal{E}^K$ is compact. Since $N^{\alpha}_{\varepsilon} \geq N_{\varepsilon}$ then necessarily $x_{\varepsilon} = y_{\varepsilon}$. We now prove the first limit of \eqref{eq:statement_a} and that $(t_{\varepsilon}, x_{\varepsilon})$ corresponds to a point where the supremum $N_{\varepsilon}$ is achieved. Passing to the lower limit we get
	$$
	U(t_{\varepsilon}, x_{\varepsilon}) - V(t_{\varepsilon}, x_{\varepsilon}) - 2\varepsilon \|x_{\varepsilon}\|^2 - \underset{\alpha\rightarrow + \infty}{\lim \sup}~ \alpha \|x^{\alpha}_{\varepsilon} - y^{\alpha}_{\varepsilon}\|^2 \geq N_{\varepsilon}.
	$$
	Hence by definition of $N_{\varepsilon}$ we necessarily have that $	\underset{\alpha\rightarrow + \infty}{\lim } ~\alpha \|x^{\alpha}_{\varepsilon} - y^{\alpha}_{\varepsilon}\|^2 = 0 $
	and that 
	$$
	N_{\varepsilon} = U(t_{\varepsilon}, x_{\varepsilon}) - V(t_{\varepsilon}, x_{\varepsilon}) - 2\varepsilon \|x_{\varepsilon}\|^2.
	$$
	To conclude we show that $N_{\varepsilon} \rightarrow N$ and that $\varepsilon \|x_{\varepsilon}\|^2 \rightarrow 0$. For $\xi>0$ consider $(t, x)$ that is $\xi$-optimal in the definition of $N$:
	$$
	U(t, x) - V(t, x) \geq N - \xi.
	$$
	For $ \varepsilon $ small enough $2\varepsilon \|x\|^2$ is lower than $\xi$, and we get
	$$
	N \geq N_{\varepsilon} \geq U(t, x) - V(t, x) - 2\varepsilon \|x\|^2 \geq N - 2\xi.
	$$
	Therefore we get convergence of $N_{\varepsilon}$ towards $N$ and as consequence
	$$
	U(t_{\varepsilon}, x_{\varepsilon}) - V(t_{\varepsilon}, x_{\varepsilon}) - 2\varepsilon \|x_{\varepsilon}\|^2 \rightarrow N .
	$$
	Since for any $\varepsilon$ we have
	$$
	N \geq U(t_{\varepsilon}, x_{\varepsilon}) - V(t_{\varepsilon}, x_{\varepsilon}) \geq  U(t_{\varepsilon}, x_{\varepsilon}) - V(t_{\varepsilon}, x_{\varepsilon}) - 2\varepsilon \|x_{\varepsilon}\|^2 = N_{\varepsilon},
	$$
	we get that $\varepsilon \|x_{\varepsilon}\|^2 \rightarrow 0$. This concludes the proof.
	
	\end{proof}

We now extend Proposition \ref{prop:comparison_bounded} to the case of functions with polynomial growth.

\begin{property}
\label{prop:comparison_polynomial}
Let $U\in USC( \mathcal{E}^K)$ with polynomial growth be a viscosity sub-solution of Equation $(\mathbf{HJB})_K$ and $V\in LSC( \mathcal{E}^K)$ with polynomial growth be a viscosity super-solution of Equation $(HJB)_K$ such that $U(T, \cdot)\leq V(T, \cdot)$. Then 
$$
U\leq V \text{ on } \mathcal{E}^K.
$$
\end{property}
\begin{proof}
	There exists $k>0$ such that
	$$
	\underset{\|x\|\rightarrow+\infty}{\lim } ~ \frac{|U(t, x)|+|V(t, x)|}{1 + \|x\|^k} = 0.
	$$
	We introduce the following function
	$$
	w(t, x) = e^{M(T-t)}(1 + \|x\|^{2 k}),
	$$
	where $M$ is a positive constant. We have
	$$
	D^Kw(t, x) = e^{M(T-t)} \begin{pmatrix}
	P^{11}_{2k - 1}(\|\theta^a\|_1) & P^{12}_{2k - 1}(i) \\ P^{21}_{2k - 1}(\|\theta^b\|_1) & P^{22}_{2k - 1}(i)
	\end{pmatrix}
	$$
	with $(P^{ij}_{2k - 1})_{i, j\in \{1, 2\}}$ polynomials with degree $2k-1$. Consequently for some $C>0$
	$$
	\|x\| \|D^Kw(t, x)\| \leq Cw(t, x).
	$$
	We have
	$$
	\sigma^2 \partial^2_P w(t, x) \leq  C( 1 + \|x\|^2) e^{M(T-t)} Q_{2k - 2}(\|x\|) \leq C w(t, x)
	$$
	and
	$$
	d(t, x) \partial_P w(t, x) \leq e^{M(T - t)}C(1 + \|x\|) Q_{2k-1}(\|x\|) \leq C w(t, x)
	$$
	where $Q_{2k-2}$ and $Q_{2k-1}$ are two polynomials with respective degree $2k-2$ and $2k-1$. Consequently for any constant $B$
	$$
	-\partial_t w(t, x) - d(t, x) \partial_P w(t, x) - \frac{1}{2}\sigma^2 \partial^2_P w(t, x) - B\|x\|\|D^Kw(t, x)\| \geq w(t, x)(M - C) 
	$$
	which is positive for $M$ large enough. Hence for any $\varepsilon>0$ the function $U - \varepsilon w$ is a bounded from above viscosity sub-solution of Equation $(\mathbf{HJB})_K$. Indeed if $U - \varepsilon w < \phi$ then $U<\phi + \varepsilon w $ consequently
	$$
	F\big(t, x, U(t, x), \nabla (\phi +\varepsilon w)(t, x), \partial^2_{pp}(\phi +\varepsilon w)(t, x), D^KU(t, x)\big)\leq 0.
	$$
	We have for $M$ large enough
	\begin{align*}
	 F\big(&t, x, U(t, x)-\varepsilon w(t, x), \nabla \phi (t, x), \partial_{pp}^2 \phi(t, x), D^K(U-\varepsilon w)(t, x) \big)\\
 -& F\big(t, x, U(t, x), (\nabla(\phi +\varepsilon w)(t, x), \partial^2_{pp} (\phi+\varepsilon w )(t, x),  D^KU(t, x)\big)\\
	 &\leq   \varepsilon \big( -rw(t,x) +  (\partial_t  + d \partial_P  + \frac{1}{2}\sigma^2 \partial_{pp}^2)w(t, x) + C\|x\|  \|D^K w(t, x)\| \big) \\
	 & <  0.
	\end{align*}
	It implies that
	$$
	F\big(t, x, U(t, x)-\varepsilon w(t, x), (\partial_t \phi (t, x), \partial_p \phi (t, x)),\partial_{pp}^2 \phi(t, x), D^K(U-\varepsilon w)(t, x) \big) \leq 0.
	$$
	We show in the same way that $V+\varepsilon w$ is a bounded from below viscosity super-solution. Then from Proposition \ref{prop:comparison_bounded} we have
	$$
	U-\varepsilon w \leq V + \varepsilon w
	$$
	and taking $\varepsilon$ to $0$ we get the stated result.
	\end{proof}
	
An immediate consequence from Proposition \ref{prop:comparison_polynomial} is that there exists a unique viscosity solution with polynomial growth to $(\mathbf{HJB})_K$. We now prove the existence of such solution using a verification argument.

\subsubsection{Definition of the continuation utility function}
\label{proof:utility}

For $(t, x)\in \mathcal{E}^K$ and $\delta \in \mathcal{A}$ we define
$$
J^K(t, x;\delta) = \mathbb{E}_{t, x}^{\delta}[G(i^{t, x}_T, P^{t, x}_T) e^{-r(T-t)} + \int_t^T e^{-r(s-t)}  \tilde{g}(s, X^{t, x}_s, \delta_s) \mathrm{d}s]
$$
where
$$
\tilde{g}(s, x, \delta) = g(i, P) + \delta^a \lambda^{a}(s,x, \delta)+ \delta^b \lambda^{b}(s,x, \delta).
$$
We also define
\begin{equation}
\label{eq:utility}
U^K(t, x) = \underset{ \delta \in \mathcal{A} }{\sup}~J^K(t, x;\delta)
\end{equation}
that is the maximal utility than can expect a market maker starting its trading from time $t$ with initial market condition given by $x$. By Lemma \ref{lemma:apriori_estimates} we get that $U^K$ has polynomial growth. More precisely there exists a positive constant $\kappa$ such that
$$
U^K(t, x) \leq \kappa(1 + \|x\|^2).
$$
We also define
$$
\mathcal{A}_t = \{\delta \in \mathcal{A}\text{ s.t. }\delta \text{ is independent of }\mathcal{F}_t\},
$$
the set of controls starting from $t$ and independent from the past. Since under $\mathbb{P}_0$ the processes $N^a$ and $N^b$ have independent increments, using the same arguments than in Remark 2.2-$(iv)$ in \cite{touzi2012optimal} we get 
$$
U^K(t, x) = \underset{ \delta \in \mathcal{A}_t }{\sup}J^K(t, x;\delta).
$$

In the next sections we show that the function $U^K$ is the unique viscosity solution with polynomial growth to $(\mathbf{HJB})_K$ . For this we prove a dynamic programming principle for $U^K$ and then conclude using a verification argument.

\subsubsection{Dynamic programming principle}
\label{proof:ppd}

Consider the lower and upper semi-continuous version of $U^K$:
$$
U_*^K(x) =\underset{y \rightarrow x}{\lim\inf}~U^K(y)\text{ and }U^{K*}(x) =\underset{y \rightarrow x}{\lim\sup}~U^K(y).
$$
Inspired by \cite{touzi2012optimal} we prove the following dynamic programming principle.
\begin{theorem}
\label{th:dynamic_programming_principle}
Let $(t, x)\in \mathcal{E}^K$ be fixed and $\{ \theta^{\delta},~ \delta \in \mathcal{A}_t \}$ be a family of finite stopping times with values in $[t, T]$. Assume that for any $\delta$, $(X^{t,x}_{s}\mathbf{1}_{s\in [t, \theta^{\delta}]})_{s\in [0,T]}$ is $L^{\infty}$-bounded. Then we have
$$
U^K(t, x) \geq \underset{\delta \in \mathcal{A}_t}{\sup} ~ \mathbb{E}^{\delta}[ e^{-r(\theta^{\delta} - t)}U^K_{*}(\theta^{\delta}, X^{t, x}_{\theta^{\delta}}) +  \int_t^{\theta^{\delta}}e^{-r(s-t)}   \tilde{g}(s, X^{t, x}_s,\delta_s)\mathrm{d}s]
$$
and
$$
U^K(t, x) \leq \underset{\delta \in \mathcal{A}_t}{\sup} ~ \mathbb{E}^{\delta}[ e^{-r(\theta^{\delta} - t)}U^{K*}(\theta^{\delta}, X^{t, x}_{\theta^{\delta}}) + \int_t^{\theta^{\delta}}e^{-r(s-t)} \tilde{g}(s, X^{t, x}_s,\delta_s)\mathrm{d}s].
$$	
\end{theorem}

The proof of Theorem \ref{th:dynamic_programming_principle} is the same as the one of Theorem 2.3 in \cite{touzi2012optimal}. However since we are working on non-standard domains we write the proof for the sake of completeness.

\begin{proof}

We first show the first inequality. We consider a continuous function $\psi$ such that $U^K\geq \psi$. By definition of $U^K$ for any $(t, x) \in \mathcal{E}^K$ there is an admissible control $\delta^{t, x, \varepsilon}\in \mathcal{A}_t$ that is $\varepsilon$ optimal:
$$
J^K(t, x; \delta^{t, x, \varepsilon}) \geq U^K(t, x)-\varepsilon.
$$
Using Fatou's lemma and the fact that $G$ and $\tilde{g}$ are lower semi-continuous we get that $J^K(\cdot; \delta^{t, x, \varepsilon})$ is a lower semi-continuous function. Then $\psi$ being upper semi continuous we can find a family of positive real $\big(r_{t, x}\big)_{t, x\in \mathcal{E}^K}$ such that for any $(t, x)\in \mathcal{E}^K$ we have
$$
\psi(t, x) - \psi(s, y) \geq -\varepsilon \text{ and } J^K(t, x;\delta^{t, x, \varepsilon}) - J^K(s, y;\delta^{t, x, \varepsilon})\leq \varepsilon, \text{ for } (s, y)\in B(t, x;r_{t, x}) 
$$
where 
$$
B(t, x;r) = \{(s, y)\in \mathcal{E}^K\text{ s.t. }s\in (t-r, t), ~\|x-y\| < r \}.
$$
The system $(B(t,x;r_{t, x}))_{t, x\in \mathcal{E}}$ forms an open covering of $\mathcal{E}^K$. With the topology it is endowed $\mathcal{E}^K$ is second countable since $[0, T]\times \mathbb{R} \times \mathbb{Z} \times \mathbb{L}_1 \times \mathbb{L}_1 $ is second countable. Hence by the Lindel\"of covering Theorem we can extract from $\big(B(t, x;r(t, x))\big)_{(t, x)\in \mathcal{E}^K}$ a countable subfamily that covers $\mathcal{E}^K$. Thus we have $(t_i, x_i, r_i)_{i\in \mathbb{N}}$ such that 
$$
\mathcal{E}^K\subset \bigcup_{i \in \mathbb{N}}B(t_i, x_i;r_i).
$$
Set $A^n = \bigcup_{0\leq i \leq n}A_i$. Consider $A_{0} = \{T\}\times\mathcal{X}^K_T,~C_{-1}=\emptyset$ and define the sequence
$$
A_{i + 1} = B(t_{i +1}, x_{i +1};r_{i+1}) \backslash C_i, \text{ where }C_i = C_{i-1}\cup A_i,~i\geq 0.
$$
Now fix $\delta \in \mathcal{A}_t$. With the above construction, we have $(\theta^{\delta}, X^{t, x}_{\theta^{\delta}})\in \cup_{i\geq 0 }A_i$ and for $i \geq 1$, we have 
$$
J^K(\cdot; \delta^{t_i, x_i, \varepsilon})\geq \psi - 3\varepsilon \text{ on }A_i.
$$
We define the control process $\delta^{\varepsilon, n}$ by
$$
\delta^{\varepsilon, n}_{s} = \mathbf{1}_{[t, \theta^{\delta}]}(s)\delta_s + \mathbf{1}_{(\theta^{\delta}, T]}(s)\big( \mathbf{1}_{{A^n}^{c}}(\theta^{\delta}, X^{t, x}_{\theta^{\delta}})\delta_s + \sum_{i = 1}^n \mathbf{1}_{A_i}(\theta^{\delta}, X^{t, x}_{\theta^{\delta}})\delta^{t_i, x_i, \varepsilon}_s \big).
$$
The control $\delta^{\varepsilon, n}$ is in $\mathcal{A}_t$. By Equation \eqref{eq:markovian_property} we have
\begin{align*}
\mathbb{E}_{t, x}^{\delta^{\varepsilon, n}}[G(i_T^{t, x}, P^{t, x}_T) e^{-r(T-t)} &+\int_{t}^T e^{-r(s-t)}\tilde{g}(s, X^{t, x}_s, \delta_s)  \mathrm{d}s |\mathcal{F}_{\theta^{\delta}}]\mathbf{1}_{A^n}(\theta, X^{t, x}_{\theta^{\delta}})\\ 	
	=& \big( U^K(T, X_T^{t, x})e^{-r(T-t)} +\int_{t}^T e^{-r(s-t)}\tilde{g}(s, X^{t, x}_s, \delta_s) \mathrm{d}s \big)\mathbf{1}_{A_0}(\theta^{\delta}, X^{t, x}_{\theta^{\delta}})\\
	&+ \sum_{i = 1}^n \big( e^{-r(\theta^{\delta} - t)} J^K(\theta^{\delta}, X^{t, x}_{\theta^{\delta}}; \delta^{t_i, x_i, \varepsilon}) + \int_{t}^{\theta^{\delta}} e^{-r(s- t)} \tilde{g}(s, X^{t, x}_s, \delta_s) \mathrm{d}s \big) \mathbf{1}_{A_i}(\theta^{\delta}, X^{t, x}_{\theta^{\delta}})\\
	 \geq &  \sum_{i = 0}^n \big( e^{-r(\theta^{\delta} - t)} \psi(\theta^{\delta}, X^{t, x}_{\theta^{\delta}}) - 3\varepsilon + \int_{t}^{\theta^{\delta}} e^{-r(s - t)} \tilde{g}(s, X^{t, x}_s, \delta_s) \mathrm{d}s \big)  \mathbf{1}_{A_i}(\theta^{\delta}, X_{\theta^{\delta}}^{t, x})\\
	 \geq & \big(e^{-r(\theta^{\delta} - t)} \psi(\theta^{\delta}, X^{t, x}_{\theta^{\delta}}) - 3\varepsilon + \int_{t}^{\theta^{\delta}} e^{-r(s - t)} \tilde{g}(s, X^{t, x}_s, \delta_s) \mathrm{d}s\big) \mathbf{1}_{A^n}(\theta^{\delta}, X_{\theta^{\delta}}^{t, x}).
	\end{align*}
	Thus we get
	\begin{align*}
	U^K(t, x) \geq& J^K(t, x; \delta^{\varepsilon, n})\\
	\geq& \mathbb{E}_{t, x}^{\delta^{\varepsilon, n}}[ \mathbb{E}_{t, x}^{\delta^{\varepsilon, n}}[ G(i_T^{t, x}, P_T^{t, x}) e^{-r(T-t)} +\int_{t}^T e^{-r(s-t)}\tilde{g}(s,X^{t, x}_s, \delta^{\varepsilon, n}_s ) \mathrm{d}s|\mathcal{F}_{\theta^{\delta}}]]\\
	\geq& \mathbb{E}_{t, x}^{\delta^{\varepsilon, n}}[ \big( e^{-r(\theta^{\delta} - t)} \psi(\theta^{\delta}, X^{t, x}_{\theta^{\delta}}) - 3\varepsilon + \int_{t}^{\theta^{\delta}} e^{-r(s - t)} \tilde{g}(s, X^{t, x}_s, \delta_s ) \mathrm{d}s \big) \mathbf{1}_{A_n}(\theta^{\delta}, X_{\theta^{\delta}}^{t, x})]\\
	&+\mathbb{E}_{t, x}^{\delta^{\varepsilon, n}}[ \big( G(i_T^{t, x}, P_T^{t, x}) e^{-r(T-t)} +\int_{t}^T e^{-r(s-t)}\tilde{g}(s, X_s^{t, x},\delta_s) \mathrm{d}s\big) \mathbf{1}_{{A^n}^{c}}(\theta^{\delta}, X^{t, x}_{\theta^{\delta}})].
	\end{align*}
	Since $L^{t,x;\delta^{\varepsilon, n}}$ is a true martingale and $L_s^{t,x;\delta^{\varepsilon, n}} = L_s^{t,x;\delta}$ for $s\in [t, \theta^{\delta}]$ we have
	\begin{align*}
	U^K(t, x) \geq& \mathbb{E}_{t, x}^{\delta}[ \big( e^{-r(\theta^{\delta} - t)} \psi(\theta^{\delta}, X^{t, x}_{\theta^{\delta}}) - 3\varepsilon + \int_{t}^{\theta^{\delta}} e^{-r(s - t)} \tilde{g}(s, X_s^{t, x},\delta_s) \mathrm{d}s \big) \mathbf{1}_{A_n}(\theta^{\delta}, X_{\theta^{\delta}}^{t, x})]\\
	&+\mathbb{E}_{t, x}^{\delta^{\varepsilon, n}}[ \big( G(i_T^{t, x}, P_T^{t, x}) e^{-r(T-t)} +\int_{t}^T e^{-r(s-t)}\tilde{g}(s, X_s^{t, x},\delta_s) \mathrm{d}s\big) \mathbf{1}_{{A^n}^{c}}(\theta^{\delta}, X^{t, x}_{\theta^{\delta}})].	
	\end{align*}
	By dominated convergence letting $n\rightarrow +\infty$ we get 
	$$
	U^K(t, x) \geq -3\varepsilon + \mathbb{E}_{t, x}^{\delta}[ e^{-r(\theta^{\delta} - t)} \psi(\theta^{\delta}, X^{t, x}_{\theta^{\delta}})  + \int_{t}^{\theta^{\delta}}e^{-r(s-t)}\tilde{g}(s, X_s^{t, x}, \delta_s) \mathrm{d}s ].
	$$
	Since $\varepsilon$ is any positive real we have
	$$
	U^K(t, x) \geq \mathbb{E}_{t, x}^{\delta}[ e^{-r(\theta^{\delta} - t)} \psi(\theta^{\delta}, X^{t, x}_{\theta^{\delta}}) + \int_{t}^{\theta^{\delta}}e^{-r(s-t)}\tilde{g}(s, X_s^{t, x},\delta_s) \mathrm{d}s\big) ].
	$$
	We now explain how to pass from $\psi$ dominated by $U^K$ to $U^K_{*}$. By hypothesis, for any $\delta$ we can find $r$ such that almost surely $\| X^{t, x}_{s}\| \leq r$ for any $s\in [t, \theta^{\delta}]$. Then we can find an increasing sequence of continuous functions on $\mathcal{E}^K$, $(\Phi_n)_{n\geq 0}$ such that $\Phi_n\leq U^K_* \leq U^K$  and such that $\Phi_n$ converges pointwise towards $U^K_*$ on $ \big( [0, T]\times B_{r}(x) \big) \cap \mathcal{E}^K$ (see Lemma 3.5. in \cite{reny1999existence}), where
	$$
	B_r(x) = \{y \in \mathbb{R}\times \mathbb{Z} \times L^1\times L^1 \text{ s.t. } \|y-x\|\leq r\}.
	$$
	Consequently from monotone convergence Theorem we have
	$$
	U^K(t, x) \geq  ~\mathbb{E}_{t, x}^{\delta}[ e^{-r(\theta^{\delta} - t)} U^K_*(\theta^{\delta}, X^{t, x}_{\theta^{\delta}}) + \int_{t}^{\theta^{\delta}}e^{-r(s-t)}\tilde{g}(s, X_s^{t, x},\delta_s) \mathrm{d}s ].
	$$
	Then we can pass to the supremum in $\delta \in \mathcal{A}_t$ to get the result.\\

Now we show the first inequality. Take $\delta \in \mathcal{A}_t$ and consider $\tilde{\delta}$ the controlled process obtained after freezing the trajectory of $\delta$ up to time $\tau^{\delta}$. By definition of $U^K$ we have
$$
U^{K*}(\tau^{\delta}, X^{t, x}_{\tau^{\delta}}) \geq 	\mathbb{E}_{\tau, X^{t, x}_{\tau}}^{\tilde{\delta}}[e^{-r(T-\tau^{\delta})}G(i^{\tau^{\delta}, X^{t, x}_{\tau^{\delta}}}_{T}, P^{\tau^{\delta}, X^{t, x}_{\tau^{\delta}}}_{T}) + \int_{\tau^{\delta}}^{T} e^{-r(s-\tau^{\delta})}\tilde{g}(s, X^{\tau^{\delta}, X^{t, x}_{\tau^{\delta}}}_{s}, \delta_s) \mathrm{d}s].
$$
Using Equation \eqref{eq:markovian_property} this gives
\begin{align*}	
U^{K*}(\tau, X^{t, x}_{\tau}) e^{-r(\tau^{\delta} - t)} &+ \int_{t}^{\tau^{\delta}} e^{-r (s-t) }\tilde{g}(s, X^{t,  x}_{s}, \delta_s) \mathrm{d}s\\
& \geq \mathbb{E}_{t, x}^{\delta}[e^{-r (T - t) }G(i^{t,  x}_{T}, P^{t,  x}_{T}) + \int_{t}^{T} e^{-r (s-t) }\tilde{g}(s, X^{t,  x}_{s}, \delta_s) \mathrm{d}s |\mathcal{F}_{\tau^{\delta}}].	
\end{align*}
Now taking the average, by arbitrariness of $\delta$ we get the second inequality
$$
\underset{\delta \in \mathcal{A}_t}{\sup}~~\mathbb{E}_{t, x}^{\delta}[U^{K*}(\tau^{\delta}, X^{t, x}_{\tau^{\delta}}) e^{-r(\tau^{\delta} - t)} + \int_{t}^{\tau^{\delta}} e^{-r (s-t) }\tilde{g}(s, X^{t,  x}_{s}, \delta_s) \mathrm{d}s] \geq U^K(t, x).
$$
\end{proof}

In the next section we show that $U^{K}$ is a viscosity solution of $(\mathbf{HJB})_K$ using a verification argument based on Theorem \ref{th:dynamic_programming_principle}.

\subsubsection{Verification}
\label{proof:verification}

In this section using the dynamic programming principle proved previously we prove that $U^{K*}$ (resp. $U^K_*$) is a viscosity super (resp. sub)-solution of $(\mathbf{HJB})_K$. The proof is inspired from the proof of Propositions 6.2 and 6.3 in \cite{touzi2012optimal}.

\begin{property}	
\label{prop:verification}
The function $U^K_{*}$ (resp. $U^{K*}$) is a viscosity sub (resp. super)-solution of $(\mathbf{HJB})_K$. 
\end{property}
\begin{proof}
We first show that $U^K_{*}$ is a viscosity super-solution and then that $U^{K*}$ is a viscosity sub-solution.\\
	
		Let $(t, x)\in \mathcal{E}$ and $\phi $ be a test function such that
	$$
	(U^K_* - \phi)(t, x) = \underset{\mathcal{E}^K}{\min}~U^K_* - \phi=0
	$$
	and $(t_n, x_n)$ a sequence in $\mathcal{E}^K$ such that
	$$
	(t_n, x_n )\rightarrow (t, x) \text{ and } U^K(t_n, x_n)\rightarrow U^K_*(t, x).
	$$
	Since $\phi$ is continuous we have
	$$
	\eta_n = U^K(t_n, x_n) - \phi(t_n, x_n)\rightarrow 0.
	$$
	Let $\delta \in \mathbb{R}_+^2$ and consider the constant control process equal to $\delta$. We use the notation 
	$ X^n = X^{t_n, x_n}$ and $\mathbb{E}^{\delta}_n = \mathbb{E}^{\delta}_{t_n, x_n}$. Finally, for all $n>0$ we define the stopping time:
	$$
	\tau_n = \inf \{s>t_n~\text{s.t.}~(s-t_n, X^n_s - x_n)\notin [0, h_n)\times  B_{\alpha}\},
	$$
	where $B_{\alpha}$ the ball for $\|\cdot \|$, centered in $0$ with radius $\alpha$ positive and small enough such that if a jump occurs then the stopping time $\tau_n$ is immediatly reached. We take
	$$
	h_n = \sqrt{\eta_n}\mathbf{1}_{\eta_n\neq 0} + n^{-1}\mathbf{1}_{\eta_n = 0}.
	$$
	Notice that $\tau_n \rightarrow t$ almost surely.\\
	
	From the first inequality in the dynamic programming principle, we have
	$$
	0\leq \mathbb{E}_n^{\delta}\big[ U^K(t_n, x_n)- e^{-r(\tau_n-t_n)} U^K_{*}(\tau_n, X^n_{\tau_n}) - \int_{t_n}^{\tau_n}e^{-r(s-t_n)}\tilde{g}(s, X^n_{s}, \delta)\mathrm{d}s \big].
	$$
	Now using that $U^K_{*}\geq \phi$ we get
	$$
	0 \leq \eta_n + \mathbb{E}_n^{\delta}\big[\phi(t_n, x_n)- e^{-r(\tau_n-t_n)} \phi(\tau_n, X^n_{\tau_n}) - \int_{t_n}^{\tau_n}e^{-r(s-t_n)}\tilde{g}(s, X^n_{s}, \delta_s)\mathrm{d}s   \big].
	$$
	We can use the Ito formula since $\phi$ is smooth. Thus we get
	$$
	0 \leq  \eta_n - \mathbb{E}_n^{\delta}\big[ \int_{t_n}^{\tau_n} e^{-r(s - t_n)}\big( (-r\phi + \partial_t \phi + \mathcal{L}^{\delta}\phi )(s, X^n_s) + \tilde{g}(s, X^n_s, \delta_s) \big)  \mathrm{d}s   \big]- \mathbb{E}_n^{\delta}[  M^n_{\tau_n} ]
	$$ 
	where
	$$
	\mathcal{L}^{\delta}\phi(s, x) =  \mathcal{L}^P\phi(s, x) + \sum_{j = a, b} D^K_j \varphi(s, x)e^{-\frac{k}{\sigma}\delta^j_u}\Phi\big(\theta^j(s)\big)
	$$
	and with
	$$
	M_{s}^n = \int_{t_n}^{s }e^{-r(s-t_n)}\big(   D^a_{K} \phi(s, X^{n}_{s}) \mathrm{d}M^{\delta,a}_s+ D^b_{K} \phi(s, X^{n}_{s})\mathrm{d}M^{\delta,b}_s + \sigma \partial_p \phi(s, X^{n}_{s}) \mathrm{d}W_s\big)
	$$
	The function $\phi$ being continuous, the integrands in the term $M^n$ are all bounded so the expectation of $M^n$ under $\mathbb{P}^{\delta}_n$ is $0$. Consequently we have
	$$
	0 \leq \frac{\eta_n}{h_n} - \mathbb{E}_n^{\delta}\big[\frac{1}{h_n} \int_{t_n}^{\tau_n}e^{-r(s-t_n)}\big(  (-r\phi + \partial_t \phi + \mathcal{L}^{\delta}\phi)(s, X^n_s) + \tilde{g}(s, X^n_s, \delta)\big) \mathrm{d}s   \big].
	$$
	Taking $n\rightarrow + \infty$ using dominated convergence and arbitrariness of $\delta$ we get
$$
		0 \leq   (r\phi - \partial_t \phi - \mathcal{L}^{\delta}\phi)(t, x) - \tilde{g}(t, x, \delta).
$$
	The control $\delta$ being arbitrary we finally have that
	$$
	F(t, x, \phi(t, x), \nabla \phi(t, x), \partial_{pp}^2 \phi(t, x), D^K\phi(t, x)) \geq 0.
	$$
	Thus $U^K_*$ is a viscosity supersolution of  $(\mathbf{HJB})_K$.\\

	Now we suppose that $U^{K*}$ is not a viscosity subsolution of $(\mathbf{HJB})_K$ and exhibit a contradiction. According to the definition of viscosity subsolution we can find $\phi$ a test function and $(t_0, x_0)$ such that
	$$
	\begin{array}{ll}
	0 = (U^{K*} - \phi)(t_0, x_0) > (U^{K*} - \phi)(t, x),~ \forall~(t, x)\in \mathcal{E}^K\backslash \{(t_0, x_0)\}
	\end{array}
	$$
	and that
	\begin{equation}
	\label{eq:verification_proof_a}
	F(t_0, x_0, \phi(t_0, x_0), \nabla \phi(t_0, x_0), \partial_{pp}^2 \phi(t_0, x_0), D^K\phi(t_0, x_0)) > 0.	
	\end{equation}
	By continuity of $\phi$ and $F$ we have existence of a $r>0$ small enough such that on $B_{r}(t_0, x_0)\backslash \{(t_0, x_0)\}$ we have
	$$
	h  = -F(\cdot, \phi, \nabla \phi, \partial_{pp}^2 \phi, D^K\phi)  < 0.
	$$
	Moreover we can find some $\eta>0$ (up to a change of $r$), such that
	$$
	\underset{ \partial B_r(t_0, x_0) \cup \mathcal{J}(t_0, x_0)  }{ \sup } U^{K*} - \phi = -2 \eta e^{rT}
	$$
	where $\mathcal{J}(t_0, x_0)$ is the set of all values that can be reached if a jump occurs inside $B_r(t_0, x_0)$. Note that it is a compact set. We consider a sequence $(t_n, x_n)_{n\geq 0 }\in \mathcal{E}^K$ such that
	$$
	\underset{n\rightarrow + \infty}{\lim}(t_n, x_n) = (t_0, x_0) \text{ and } \underset{n\rightarrow + \infty}{\lim} U^K(t_n, x_n) = U^{K*}(t_0, x_0). 
	$$
	Since $U^{K}(t_n, x_n) - \phi(t_n, x_n)\rightarrow 0$ we can assume that 
	$$
	| U^{K}(t_n, x_n) - \phi(t_n, x_n) | \leq \eta \text{ for any }n\geq 1.
	$$
	For a fixed control $\delta \in \mathcal{A}_{t_n}$
	We define the stopping time
	$$
	\tau_n = \inf\{ t>t_n\text{ s.t }X_t^{t_n, x_n} \notin B_r(t_0, x_0)\}.
	$$
	At the stopping time, either the process $X^{t_n, x_n}$ has not jumped and so is on $\partial B_r(t_0, x_0)$ or has jumped and is in $\mathcal{J}(t_0, x_0)$. Thus
	$$
	e^{-r(\tau_n-t_n)} \phi(\tau_n, X^{t_n, x_n}_{\tau_n}) \geq 2\eta + e^{-r(\tau_n-t_n)}  U^K(\tau_n, X^{t_n, x_n}_{\tau_n}).
	$$
	We derive from the Ito formula
	\begin{eqnarray*}
	U^K(t_n, x_n) &\geq & -\eta + \phi(t_n, x_n)\\
	&=& -\eta + \mathbb{E}_n^{\delta}\big[ e^{-r(\tau_n -t_n )}\phi(\tau_n, X^n_{\tau_n}) - \int_{t_n}^{\tau_n}  e^{-r(s -t_n )} ( -r + \partial_t + \mathcal{L}^{\delta})\phi(s, X^n_{s^-}) \mathrm{d}s \big].
	\end{eqnarray*}
	Hence using Equation \eqref{eq:verification_proof_a} we get
	\begin{eqnarray*}
U^K(t_n, x_n)	&\geq & -\eta + \mathbb{E}_n^{\delta}\big[ e^{-r(\tau_n -t_n )}\phi(\tau_n, X^n_{\tau_n}) + \int_{t_n}^{\tau_n}  e^{-r(s -t_n )} \tilde{g}(s, X^n_s, \delta)  \mathrm{d}s \big]\\
	&\geq & \eta + \mathbb{E}_n^{\delta}\big[ e^{-r(\tau_n -t_n )}U^{K*}(\tau_n, X^n_{\tau_n}) + \int_{t_n}^{\tau_n}  e^{-r(s -t_n )} \tilde{g}(s, X^n_s, \delta)   \mathrm{d}s \big].
	\end{eqnarray*}
	Since $\delta$ is any control and $\eta$ is positive this contradict the second equation of Theorem \ref{th:dynamic_programming_principle}. Thus $U^{K*}$ is a viscosity sub-solution of $(\mathbf{HJB})_K$.
	
\end{proof}

A direct consequence of Proposition \ref{prop:verification} together with Proposition \ref{prop:comparison_polynomial} is that
$$
U^K_* \geq U^{K*}.
$$
But obviously we have $U^K_*\leq U^{K*}$, therefore $U^K_{*} = U^{K*} = U^K$. In particular $U^K$ is continuous and therefore is the unique continuous viscosity solution with polynomial growth to $(\mathbf{HJB})_K$.\\

\subsubsection{Proof of Theorem \ref{th:optimal_control} $(iii)$}
\label{proof:control}

To prove Theorem \ref{th:optimal_control} $(iii)$ we must show that $J^K(\cdot;\delta^*) = U^K$.\\

As we did previously we can show that $J^K(\cdot;\delta^*)$ is the unique viscosity solution with polynomial growth of 
$$
(\mathbf{LHJB})_K:~~\left\{ \begin{array}{ll}
&rU -\partial_tU -\mathcal{L}^P U - g - \sum_{j = a, b}  e^{-\frac{k}{\sigma}\delta^{K}_j }\Phi\big(\theta^j(u)\big) (D^K_j U + \delta^K_j)  = 0\text{, on  } \mathcal{E}^K \\
&U(T,x)=G(i, P)\text{ for }x\in \mathcal{X}^K_T
\end{array} \right. .
$$
But since $U^K$ is a viscosity solution of $(\mathbf{HJB})_K$ and by definition of $\delta^K$, $U^K$ is also a viscosity solution with polynomial growth of $(\mathbf{LHJB})_K$. This gives the result.

\subsection{Proof of Proposition \ref{prop:convergence_viscosity_solution}.}
\label{proof:convergence_solution}

We define the following functions on $\mathcal{E}^K$: 
$$
\overline{U}(x) =  \underset{(y, n)\in \bar{\mathcal{E}}   \rightarrow (x, +\infty)}{\lim\sup}  ~ U^{K_n}(y)\text{ and }\underline{U}(x) =  \underset{(y, n)\in \bar{\mathcal{E}}   \rightarrow (x, +\infty)}{\lim\inf}  ~ U^{K_n}(y),
$$
We show that $\underline{U}$ and $\overline{U}$ are respectively a viscosity super-solution and a viscosity sub-solution of $(\mathbf{HJB})_K$.\\

Consider $\phi$ a test function and $\underline{x}\in \mathcal{E}^K$ a strict minimizer of $\underline{U}-\phi$. We have existence of a sequence $(x_n, \sigma_n)_{n\in \mathbb{N}}$ in $\overline{\mathcal{E}}$ such that
	$$
	(x_n, \sigma_n)\rightarrow (+\infty, \underline{x})\text{ and }U^{K_n}(x_n)\rightarrow\underline{U}(\underline{x}).
	$$
	Consider $B_r(\underline{x})$ the closed ball of $[0, T]\times \mathbb{R} \times \mathbb{Z}\times L^1 \times L^1$ with radius $r>0$ centered in $\underline{x}$. Then we can always suppose that $x_n \in B_r(\underline{x}),~\forall n\geq 0$. Let $\underline{x}_n$ be a minimizer of the difference $U^{K_n} - \phi$ on $\mathcal{E}^{K_n} \cap B_r(\underline{x})$ (exists because $\mathcal{E}^{K_n}$ is locally compact). We note $\underline{x}_n = (t_n, p_n, i_n, \theta^{n,a}, \theta^{n,b})$. We show at the end of the proof that there exists $x \in \mathcal{E}^K$ such that $(x,+\infty)$ is the limit of a subsequence of $(\underline{x}_n, \sigma_n)_{n \geq 0}$ and that $\theta^{n,j}(t_n) \rightarrow \theta^{j}(t)$ for $j=a$ and $b$. Hence we can write
	\begin{align*}
	\underline{U}(\underline{x}) - \phi(\underline{x})  &= \underset{n\rightarrow +\infty}{\lim}~ U^{K_n}(x_n) - \phi(x_n) \\
	& \geq \underset{n \rightarrow +\infty}{\liminf}  ~ U^{K_n}(\underline{x}_{n}) - \phi(\underline{x}_{n}) \\
	& \geq  \underline{U}(x) - \phi(x).
	\end{align*}
	Thus by definition of $\underline{x}$ we get that $(\underline{x}_n)_{n\geq 0}$ converges towards $\underline{x}$ and that
	$$
	 U^{K_n}(\underline{x}_n) \underset{n \rightarrow +\infty}{\rightarrow} \underline{U}(\underline{x}).
	$$
	As a consequence when $n$ is large enough $\underline{x}_n$ is a local minimizer of $U^{K_n} - \phi$ (because it is in the interior of $\overline{B}_r(\underline{x})$) hence by definition of viscosity solutions
	$$
	F\big(\underline{x}_n, U^{K_n}(\underline{x}_n), \nabla\phi(\underline{x}_n), \partial^2_{pp}\phi(\underline{x}_n), D^{K_n}U^{K_n}(\underline{x}_n) \big) \geq 0.
	$$
	Then by definition of $\underline{U}$ and since $U^{K_n}(\underline{x}_n) \rightarrow \underline{U}(\underline{x})$:
	$$
	\underset{n \rightarrow + \infty}{\lim \inf }~D^{K_n}U^{K_n}(\underline{x_n}) \geq D^K \underline{U}(\underline{x}).
	$$
	Finally since $F$ is decreasing with respect to the last variable and since  $\theta^{n,j}(t_n)$ converges towards $\theta^j(t)$ for $j=a$ and $b$ we have
	$$
	F\big(\underline{x}, \underline{U}(\underline{x}), \nabla\phi(\underline{x}), \partial^2_{pp} \phi(\underline{x}), D_K \underline{U}(\underline{x})) \geq  \underset{n\rightarrow + }{\lim \sup } ~F\big( \underline{x}_n, U^{K_n}(\underline{x}_n), \nabla \phi(\underline{x}_n), \partial^2_{pp}\phi(\underline{x}_n), D^{K_n} U^{K_n}(\underline{x}_n)) \geq 0.
	$$
	So by Definition \ref{def:definition} $\underline{U}$ is a viscosity super-solution of $(\mathbf{HJB})_K$. In the same way we can show that $\overline{U}$ is a viscosity sub-solution of $(\mathbf{HJB})_K$. Moreover since the a priori inequalities on $U^{K_n}$ can be chosen uniform in $n$ (because $\|K_n\|_1 \rightarrow \|K\|_1$) they are true for $\underline{U}$ and $\overline{U}$. Therefore  Proposition \ref{prop:comparison_polynomial} implies that $\underline{U} \geq \overline{U}$. Because we have the other inequality by definition we get $\overline{U}=\underline{U}=U^K$, the unique viscosity solution with polynomial growth of $(\mathbf{HJB})_K$.\\

	To complete the proof we show that $(\underline{x_n})_{n\geq 0}$ admits a subsequence converging towards some $x \in \mathcal{E}^K$ and that for $j=a$ and $b$, $\theta^{n,a}(t_n)$ converges $\theta^a(t)$.\\
	
	 We have $ 	\underline{x}_n = (t_n, p_n, i_n, \theta^{n, a}, \theta^{n, b})$ with
	$$
	\theta^{n, a} = \sum_{j = 1}^{m^{n, a}} K_n(-T^{n, a}_j)\text{ and }\theta^{n, b} = \sum_{j = 1}^{m^{n, b}} K_n(-T^{n, b}_j)
	$$
	where $m^{n, a}$ and $m^{n, b}$ are non-negative integers, $(T^{n, a}_j)_{1\leq j \leq m^{n, a}}$ and $(T^{n, b}_j)_{1\leq j \leq m^{n, b}}$ are in $[0, t_n]$. We recall that $(\|x_n\|)_{n\geq 0}$ is bounded. Hence up to a subsequence $(t_n, p_n, i_n, \|\theta^{n,a}\|_1, \|\theta^{n, b}\|_1)_{n\geq 0}$ converges towards some $(t, p, i, l^a, l^b)$. Since we have assumed that $\|K\|_1$ is positive the convergence of $(\|\theta^{n,a}\|_1)_{n\geq 0}$ and $(\|\theta^{n,b}\|_1)_{n\geq 0}$ imply those of $(m^{n, a})_{n\geq}$ and $(m^{n,b})_{n\geq 0}$. Consequently those sequences are eventually constant and equal to $m^a$ and $m^b$ for $n$ large enough. Then up to a subsequence we have convergence of $\big((T^{n, a}_j)_{1\leq j \leq m^a}\big)_{n\geq 0}$ and $\big((T^{n, b}_j)_{1\leq j \leq m^b}\big)_{n\geq 0}$ since they take their values in $[0, T]^{m^a}$ and $[0, T]^{m^b}$ which are compact sets. We consider $(T^a_j)_{1\leq j \leq m^a}$ and $(T^b_j)_{1\leq j \leq m^b}$ their limits. We now show that $(\theta^{n, a})_{n\geq 0}$ converges in $L^1$ towards 
	$$
	\theta^a(\cdot) = \sum_{j = 1}^{m^a}K(\cdot - T^a_j) \in \Theta^K_t.
	$$
	We show that $K_n(\cdot - T^{n, a}_j)$ converges in $L^1$ towards $K(\cdot - T^a_j)$ to conclude. We write
	\begin{align*}
	\| K(\cdot - T^a_j) - K_n(\cdot - T^{n, a}_j) |_1 & \leq \| K_n(\cdot - T^{n,a}_j) - K(\cdot - T^{n,a}_j) \|_1 + \| K(\cdot - T^{n,a}_j) - K(\cdot - T^a_j)\|_1 \\
	& \leq  \| K_n - K \|_1 + \| K(\cdot - T^{n,a}_j) - K(\cdot - T^a_j)\|_1.	
	\end{align*}
	The first term goes to $0$ by hypothesis, the second by dominated convergence. Same results holds for $(\theta^{n, b})_{n\geq 0}$ and $\theta^b$. Consequently we have proved the convergence of $(x_n)_{n\geq 0}$ towards
	$$
	x = (t, p, i, \theta^a, \theta^b) \in \mathcal{E}^K.
	$$
	We finally show that $\theta^{n,a}(t_n)$ converges towards $\theta^{a}(t)$, the same methodology holds for $b$. We have for $n$ large enough
	$$
	\big| \theta^{n,a}(t_n)-\theta^{a}(t) \big| \leq \sum_{j = 1}^{m^a} \big|K_n(t_n - T^{n,a}_j) - K(t - T^{a}_j) \big|.
	$$
	The uniform convergence of $K_n$ towards $K$ implies that $K_n(t_n - T^{n,a}_j)$ converges towards $K(t - T^{a}_j)$. This concludes the proof.

\subsection{Proof of point $(iv)$ of Theorem  \ref{th:verification_exp}}
\label{proof:verification_exp}

We recall that the proof of Theorem \ref{th:verification_exp} is exactly the same of Theorem \ref{th:optimal_control}. Hence for any $(t, y)\in \mathcal{E}^n$ we define for $(t, y)\in \mathcal{E}^n$ the process $Y^{t, y} = (i^{t, y}, P^{t, y}, c^{t, y;a}, c^{t, y;b})\in \mathcal{E}^n$ by analogy with the process $X^{t, x}$ defined in Section \ref{proof:utility}. Note that by construction for any $(t, x) \in \mathcal{E}^{K_{\alpha, \gamma}}$ and for any $s\in [t, T]$ we have for $(t, y)= \mathcal{R}^{\alpha, \gamma}(t, x)$
\begin{equation}
\label{eq:proof_verif_exp_a}
(s, Y^{t, y}_s) = \mathcal{R}^{\alpha, \gamma}(s, X^{t, x}_s).
\end{equation}
Then as in Section \ref{proof:verification} we prove that the function
$$
U^{\alpha, \gamma}(t, y) = \underset{\delta \in \mathcal{A}_t}{\sup }~\mathbb{E}^{\delta}[G(i_T^{t, y}, P_T^{t, y}) + \int_t^T \Big( g(i^{t, y}_{s}, P^{t, y}_s) + \delta^a_s \lambda^{a, \delta}_s + \delta^b_s \lambda^{b, \delta}_s \Big)\mathrm{d}s]
$$
is the unique viscosity solution with polynomial growth of $(\mathbf{HJB})_{\alpha, \gamma}$. Moreover for any $(t, x)\in \mathcal{E}^{K_{\alpha, \gamma}}$ and $(t, y) = \mathcal{R}^{\alpha, \gamma}(t, x)$ by Equation \eqref{eq:proof_verif_exp_a} we have:
$$
U^{\alpha, \gamma}(t, y) = \underset{\delta \in \mathcal{A}_t}{\sup }~\mathbb{E}^{\delta}[G(i_T^{t, x}, P_T^{t, x}) + \int_t^T \Big( g(i^{t, x}_{s}, P^{t, x}_s) + \delta^a_s \lambda^{a, \delta}_s + \delta^b_s \lambda^{b, \delta}_s \Big)\mathrm{d}s] = U^{K_{\alpha, \gamma}}(t, x).
$$
Therefore for any $(t, x) \in \mathcal{E}^{K_{\alpha, \gamma}}$ we have $U^{K_{\alpha, \gamma}}(t, x) = U^{\alpha, \gamma}\circ \mathcal{R}^{\alpha, \gamma}(t, x)$. This concludes on the proof of point $(iv)$ of Theorem \ref{th:verification_exp}.

\appendix

\section{Proof of Lemma \ref{lemma:topology_a}}
\label{proof:lemma_topology_a}

We first prove (i). Consider $(\theta_k)_{k\geq 0}$ a sequence with values in $\Theta^K_t$ that converges towards some $\theta$ in $L^1$. We have 
	$$
	\theta_k = \sum_{j = 1}^{N_k} K(\cdot -T^k_j).
	$$
	The convergence of $\|\theta_k\|_1$ towards $\|\theta\|_1$ gives that $N_k$ is constant and equal to some $N$ up to a certain rank. Finally for any subsequence $\big( (T^{\sigma(k)}_j)_{1\leq j \leq N}\big)_{k\geq 0}$ converging to some $(T_j)_{1\leq j \leq N}$ we have :
	$$
	\theta_{\sigma(k)} \rightarrow \sum_{j = 1}^{N} K(\cdot -T_j) = \theta, \text{ in }L^1.
	$$
	so $\Theta^K_t$ is closed. Now with the same notation we consider a bounded sequence $(\theta_k)_{k \geq 0}$. We can find a subsequence $\sigma$ such that $N_{\sigma(k)}$ is constant and equal to some $N\in \mathbb{N}$ and such that for $j = 1\dots N$, $T^{\sigma(k)}_j \rightarrow T_j$. This implies that 
	$$
	\theta_{\sigma(k)}\rightarrow \sum_{j=1}^N K(\cdot - T_j).
	$$
	This shows that that $\Theta^K_t$ is locally compact.\\
	
	Now we prove (ii). Consider a converging sequence $(s_k, \theta_k)_{k\geq 0}$ such that $\theta_k\in \Theta^K_{s_k}$ for any $k$ and let $\theta = \sum_i^N K(\cdot-T_j)$ be the limit of $(\theta_k)_{k\geq 0}$. Then necessarily $((T^k_j)_{1\leq j \leq n})_{k\geq 0}$ converges towards $(T_j)_{1\leq j \leq n}$. Moreover by comparison we have $T_j \leq s$ and by continuity of $K$ that
	$$
	\theta_{k}(s_k)\rightarrow \sum_{j=1}^N K(s - T_j).
	$$ 
	Finally consider now that $K(t) = \alpha e^{-\gamma t}$, for $l\in \mathbb{N}$ we have
	$$
	(\theta_k^K)^{(l)}(T) = \sum_{i = 1}^N \alpha (-\gamma)^l e^{-\gamma (T - T^k_i)}.
	$$
	The convergence of $(T_j^k)_{k \geq0}$ thus imply that $\theta_k^{(l)}(T)\rightarrow \theta^{(l)}(T)$.

\section{A priori inequalities}
\label{appendix:apriori_inequalities}

In this section we prove some a priori inequalities.

\subsection{Hawkes processes}
\label{appendix:subsec_martingale_hawkes}
Consider a Hawkes process $N$ with kernel $K = c \mathbf{1}_{\mathbb{R}_+}$ and exogenous intensity $\mu$. The intensity of $N$ is given by
$$
\lambda_t = \mu + N_t c.
$$
The existence of such process is proved in \cite{jacod1975multivariate}. Consider $T_p = \inf\{s\text{ s.t. } N_s>p\}$, by to \cite{jacod1975multivariate}, $T_{\infty} = \underset{n \rightarrow + \infty}{\lim}T_n = +\infty$. To lighten the notations we write $N^p := N^{T_p}$. We have for any $t\in [0, T]$
$$
\mathbb{E}[N^p_t] = \mathbb{E}[\int_0^{t\wedge T_p}\lambda_s\mathrm{d}s]\leq \mathbb{E}[\int_0^t C(1 + N_s^p) \mathrm{d}s].
$$
thus using a Gr\"onwall's lemma we get $\mathbb{E}[N^p_T]\leq C T e^{CT}$. The RHS being independent of $p$ and using monotone convergence we get
$$
\mathbb{E}[N_T]<+\infty.
$$
We also have 
\begin{align*}
\mathbb{E}[(N^p_t)^2] &= \mathbb{E}[\int_0^{t\wedge T_p}(2N_{s^-}^p + 1) \mathrm{d}N_s] \\
&= \mathbb{E}[\int_0^{t\wedge T_p}(2N_{s^-}^p + 1)\lambda_s \mathrm{d}s] \leq \mathbb{E}[\int_0^{t\wedge T_p}C(2N_{s}^p + 1)(N_{s}^p + 1)s \mathrm{d}s].
\end{align*}
Using again a Gronwall lemma we deduce that $\mathbb{E}[(N^p_T)^2] \leq CT^2 e^{CT}$ with $C$ independent of $p$, so
$$
\mathbb{E}[(N_T)^2]<+\infty.
$$
Now consider a Hawkes process $N$ with kernel $K$ bounded and intensity given by
$$
\lambda_t = \Phi\big(\int_0^t K(t-s)\big)\mathrm{d}N_s
$$
with $\Phi$ non decreasing in its last variable and such that $|\Phi(x)|\leq C(1 + |x|)$ for some $C>0$. By the thinning property we can see $N$ as dominated by some Hawkes process $\tilde{N}$ with kernel $C\mathbf{1}_{\mathbb{R}_+}$ and exogenous intensity $C$. Remark that as a consequence $\tilde{\lambda}$ dominates $\lambda$. Hence we get
$$
\mathbb{E}[\tilde{N}_T + \int_0^T \tilde{\lambda}_s]<+\infty
$$
then consequently
$$
\mathbb{E}[N_T]<+\infty,~\mathbb{E}[N^2_T]<+\infty\text{ and } \mathbb{E}[\int_0^{T} \delta e^{-k\delta}\lambda_s\mathrm{d}s ]< +\infty.
$$
This ensures that the function $U^K$ defined in Equation \eqref{eq:utility} is well defined. This also implies that the martingales $M^{t,x;a, \delta}$ and $M^{t,x;b, \delta}$ are uniformly integrable martingales.

\subsection{A priori estimates on $X$}
\label{appendix:subsec:apriori_estimate_process}
	
We prove here that the value function $U^K$ defined in Equation \eqref{eq:utility} has polynomial growth in $x$. For this we show some inequalities on the norm of $(X^{t, x})_{(t, x)\in \mathcal{E}^K}$. More precisely we prove the following result:
	\begin{lemma}
	\label{lemma:apriori_estimates}
	There exists some positive constant $C$ depending only on $T$ and on the regularity constants of $G$ and $g$ such that for any $(t, x)\in \mathcal{E}$
	$$
	\mathbb{E}^{\delta}_{t, x}[\underset{s\in [t, T]}{\sup}~ \|X^{t, x}_s\|^2] \leq C(1 + \|x\|^2)
	$$
	and
	$$
	|U^K(t, x)|\leq C(1 + \|x\|^2).
	$$
	\end{lemma}
To prove Lemma \ref{lemma:apriori_estimates}, consider $(t, x)\in \mathcal{E}^K$ and $\delta \in \mathcal{A}_t$ with $x = (P, i, \theta^a, \theta^b)$. We show different a priori estimates on the subprocesses composing $X^{t, x}=(P^{t, x}, i^{t, x}, \theta^{t,x;a}, \theta^{t, x;b})$ under the probability measure $\mathbb{P}^{t, x;\delta}$.
\paragraph{A priori estimates on $\theta^{t,x;a}$ and $\theta^{t,x;b}$:} We have
$$
N^a_s - N^a_t = M^{t, x;a, \delta}_s + \int_{t}^s \lambda^a(u, X^{t, x}_u, \delta_u) \mathrm{d}u
$$
since $\lambda^a(t, x, \delta)\leq C(1 + \|\theta^a\|_1)$ and $\|\theta^a_u\|_1 = \|\theta^a\|_1 + (N^a_u - N^a_t)\|K\|_1$ we have
$$
N^a_s - N^a_t \leq  M^{t, x;a, \delta}_s +  \int_t^s C\big(1 + \|\theta^{t,x;a}_u\|_1 + \|K\|_1(N^a_u-N^a_t) \big) \mathrm{d}u.
$$
Taking the expected value over the probability measure $\mathbb{P}^{t,x;\delta}$ using the fact that $M^{t, x;d, \delta}$ is a true martingale and a Gr\"onwall lemma we get
\begin{equation}
\label{eq:apriori_eq_a}
\mathbb{E}^{\delta}_{t, x}[N^a_T - N^a_t] \leq C(1 + \|\theta^a \|_1)
\end{equation}
where $C$ only depends on $T$ and on the model constants. Consequently we have for some positive constant $C$
$$
\mathbb{E}_{t, x}^{\delta}[\| \theta^{t,x;a}_s \|_1] \leq C(1 + \|\theta^a \|_1).
$$
We now give an a priori estimate for the second order moment.
\begin{align*}
(N^a_s - N^a_t)^2 &=  \int_t^s \big(2(N^a_{u^-}- N^a_t)+1\big) \lambda_u\mathrm{d}u+ \int_{t}^s \big(2(N^a_{u^-}- N^a_t)+1\big) \mathrm{d}M^{t, x;a, \delta}_u \\
&\leq  \int_{t}^s  \big(2(N^a_{u^-}- N^a_t)+1\big) C(1 + \|\theta^{t,x;a}_u \|_1 + N^a_{u}- N^a_t)\mathrm{d}u +  \int_{t}^s \big(2(N^a_{u^-}- N^a_t)+1\big) \mathrm{d}M^{t, x;a, \delta}_u \\
& \leq  \int_{t}^s  C(N^a_{u^-}- N^a_t)^2 + (N^a_{u^-}- N^a_t)C(1 + \|\theta^{t,x;a}_u \|_1)\mathrm{d}u + \int_{t}^s \big(2(N^a_{u^-}- N^a_t)+1\big) \mathrm{d}M^{t, x;a, \delta}_u.
\end{align*}
The average of the last term of the right hand side is $0$ as consequence of Appendix \ref{appendix:subsec_martingale_hawkes}. Thus taking the average and using a Gr\"onwall lemma we get
\begin{equation}
\label{eq:apriori_eq_b}
\mathbb{E}_{t, x}^{\delta}[(N^a_T-N^a_t)^2] \leq C(1 + \|\theta\|^2_1).
\end{equation}

\paragraph{A priori estimates on $P^{t,x}$:} We have
$$
\mathrm{d}P^{t,x}_s = d(s, P^{t,x}_s)\mathrm{d}s + \sigma \mathrm{d}W_s, \text{ with }P^{t,x}_t = P.
$$
By assumption there exists $k>0$ such that :$|d(t, p) - d(t, q)|\leq k|p-q|$.	We have the classic apriori estimates (see for example Theorem 1.2 in \cite{touzi2012optimal}).
\begin{equation}
\label{eq:apriori_eq_c}
\mathbb{E}_{t, x}^{\delta}[\underset{s\leq T}{\sup }~P^2_s] \leq C(1 + P^2).
\end{equation}
Where $C$ only depends only on the Lipshitz constant $k$ and on $T$.
	
\paragraph{A priori estimates on $X^{t,x}$:}  We have
\begin{align*}
&i_s = i + N^a_s-N^a_t + N^b_s-N^b_t\\
&\|\theta^j_s \|_1 = \|\theta^j \|_1 + \|K\|_1 ( N^j_s-N^j_t),\text{ for }j = a,~b.
\end{align*}
Thus we have
$$
\|X_s^{t,x}\|^2 \leq C (1 + i^2 + \|\theta^a\|^2_1+ \|\theta^b\|^2_1 + (N^a_s-N^a_t)^2 + (N^b_s - N^b_t)^2 + P_s^2).
$$
Taking the average and using Equations \eqref{eq:apriori_eq_a}, \eqref{eq:apriori_eq_b} and \eqref{eq:apriori_eq_c} we get
\begin{equation}
\label{eq:apriori_eq_d}
\mathbb{E}_{t, x}^{\delta}\big[ \underset{s\leq T}{\sup } \|X_s^{t,x}\|^2 \big] \leq C (1 + i^2 + \|\theta^a\|^2_1+ \|\theta^b\|^2_1 + P^2) = C(1+\|x\|^2)
\end{equation}
where $C$ is independent of $\delta$.

\paragraph{A priori estimates on the value function:} By the quadratic growth of $G$ and $\tilde{g}$ we get
\begin{align*}
|J^K(t, x; \delta)| &\leq \mathbb{E}_{t, x}^{\delta}\big[ e^{-r(T-t} C(1+ \|X^{t,x}_T\|^2) + \int_t^T e^{-r(s-t) } C(1+ \|X^{t,x}_s\|^2 ) \big].
\end{align*}
Because of the a priori estimates \eqref{eq:apriori_eq_a}, \eqref{eq:apriori_eq_b}, \eqref{eq:apriori_eq_c} and \eqref{eq:apriori_eq_d} we have
	$$
	|J^K(t, x; \delta)| \leq C(1 + i^2 +\|\theta^a\|^2_1+\|\theta^b\|^2_1+P^2)  \leq C(1 + \|x\|^2)
	$$
	where $C$ only depends on $T$ and the regularity constants. We conclude by arbitrariness of $\delta$.

\subsection{Rewriting of the utility}
\label{appendix:apriori_inequalities_integral}
We show that for any $\delta \in \mathcal{A}$ we have
$$ 
\mathbb{E}^{\delta}[\int_0^T e^{-rs} \delta^a_s \mathrm{d}N^a_s]  = \mathbb{E}^{\delta}[\int_0^T e^{-rs} \delta^a_s \lambda^{a,\delta}_s \mathrm{d}s].
$$
To conclude it is enough to show that 
$$
\overline{M}_t = \int_0^t e^{-rs} \delta^a_s \mathrm{d}M^a_s
$$
is a true martingale. We have
$$
[\overline{M}]_t = \int_0^t e^{-2rs} (\delta^a_s)^2 \mathrm{d}N^a_s \text{ and } \langle \overline{M} \rangle_t = \int_0^t e^{-2rs} (\delta^a_s)^2 \lambda^{a, \delta}_s \mathrm{d}s
$$
and since $(\delta^a_t)^2 \lambda^{a, \delta}_t \leq C( 1 + \|X_t\|_1) $ we get 
$$
\langle \overline{M} \rangle_T \leq \int_0^T e^{-2rs} C(1 + \|X_s\|_1) \mathrm{d}s \leq T C(1 + \underset{s \in [0, T]}{\sup}\|X_s\|_1).
$$
The last term of the RHS is integrable by Lemma \ref{lemma:apriori_estimates}. By the monotone convergence $[\overline{M}]_T$ is also integrable so $\overline{M}$ is a uniformly integrable martingale. As consequence we get 
\begin{align*}
\mathbb{E}^{\delta}[G(i_T,P_T) e^{-rT} + &\int_0^T e^{-rs}\Big(  g(i_s, P_s)\mathrm{d}s + \delta^a_s \mathrm{d}N^a_s+ \delta^b_s\mathrm{d}N^b_s\Big)]\\
& = 	\mathbb{E}^{\delta}[G(i_T,P_T) e^{-rT} + \int_0^T e^{-rs}\Big(  g(i_s,P_s) + \delta^a_s \lambda^{a,\delta}_s + \delta^b_s\lambda^{b, \delta}_s\Big) \mathrm{d}s] .
\end{align*}

Using the sames arguments we get 
We show that for any $\delta \in \mathcal{A}$ we have
$$ 
\mathbb{E}^{\delta}[\int_0^T e^{-rs} \delta^b_s \mathrm{d}N^b_s]  = \mathbb{E}^{\delta}[\int_0^T e^{-rs} \delta^b_s \lambda^{b,\delta}_s \mathrm{d}s].
$$

\section{Equivalence between the two definitions of viscosity solutions}
\label{appendix:section:equivalent_definition}
	 
	\begin{lemma}
	\label{lemma:equivalence_definitions}
	Definition \ref{def:definition} and Definition \ref{def:equivalent_definition} are equivalent.
	\end{lemma} 
	\begin{proof}
	
	We show it for sub-solutions, the demonstration is the same for super-solutions.\\	
	
 Consider $U$ a USC sub-solution of $(\mathbf{HJB})_K$ in the sense of Definition \ref{def:equivalent_definition}. Now consider a test function $ \phi $ such that $0 = U(t_0, x_0) - \phi(t_0, x_0) = \underset{\mathcal{V}}{\sup}~ U -\phi$ for $\mathcal{V}$ a neighborhood of $(t_0, x_0)$ in $\mathcal{E}^K$. We show that 
 $$
 	F\big(t_0 , x_0, U(t_0, x_0), \nabla \phi(t_0, x_0) , \partial_{pp}^2 \phi(t_0, x_0), D^K\phi(t_0, x_0) \big) \leq 0.
 $$
Writing $x = (p, z)\in \mathbb{R}\times \mathcal{Z}^K_T$ we have $ \phi(t, p, z)  = \phi(t_0, p_0, z_0) + \partial_t \phi(t_0, x_0) (t-t_0) +  \partial_p \phi(t_0, x_0) (p-p_0) + \partial_{pp}^2 \phi(t_0, x_0)\frac{(p-p_0)^2}{2} + o(|p-p_0|^2 + |t-t_0|^2) + h(z-z_0). $
 where $h$ is a modulus of continuity of $\phi$. Thus we have
 $$
 (\nabla \phi(t_0, x_0), \partial_{pp}^2 \phi(t_0, x_0), h)\in \mathcal{J}^{+}u(t_0, x_0).
 $$
 Consequently
 $$
 F(t_0, x_0, U(t_0, x_0), \nabla \phi(t_0, x_0), \partial^2_{pp} \phi(t_0, x_0), D^KU(t_0, x_0)) \geq 0.
 $$
 So $U$ is a viscocity sub-solution of $(\mathbf{HJB})_K$ in the sense of Definition \ref{def:definition}.\\
 
 Now we show the opposite implication. Consider $U$ a USC function sub-solution of $(\mathbf{HJB})_K$ in the sense of Definition \ref{def:definition}. Consider $(d, A, h)\in \mathcal{J}^{+}U(t_0, x_0)$, we built a test function $\phi$ dominating $U$ with equality at point $(t_0, x_0)$ and such that
$$
\big(\nabla \phi(x_0), \partial_{pp}^2 \phi(x_0)\big) = (d, A).
$$
We will then get the expected inequality that will extend directly to $\overline{\mathcal{J}}^{+}U(t_0, x_0)$ by continuity of $F$.\\

Using the notation $(t, x) = (t, p, z)\in [0, T]\times \mathbb{R}\times \mathcal{Z}^K_t$ we have
$$
U(t, x) \leq U(t_0, x_0) + d_1(t-t_0) +  d_2(p-p_0) + \frac{1}{2}A(p-p_0)^2 + h(z-z_0) + o(|p-p_0|^2) + o(|t-t_0|).
	$$
	hence
	$$
	U(t, p, z) - h(z-z_0) \leq U(t_0, x_0) + d_1(t-t_0) +  d_2(p-p_0) + \frac{1}{2}A(p-p_0)^2 + o(|p-p_0|^2) + o(|t-t_0|).
	$$
	We take the supremum on $z$ over a compact neighborhood of $z_0$, and consider
	$$
	\tilde{U}(t, p) = \underset{z\in B_r( z_0)\cap \mathcal{Z}^K_t}{\sup }U(t, p, z) - h(z-z_0).
	$$
	Since $\tilde{U}(t_0, p_0) = U(t_0, x_0)$ we get
	$$
	\tilde{U}(t, p) \leq \tilde{U}(t_0, p_0) +  d_1(t-t_0) +  d_2(p-p_0) + \frac{1}{2}A(p-p_0)^2 + o(|p-p_0|^2) + o(|t-t_0|).
	$$
	We prove at the end that $\tilde{U}$ is a USC function and assume this is true. The last equation means that $(d, A)\in \mathcal{J}^+\tilde{U}(t_0, p_0)$. Then by an argument developped for the analysis of viscosity solutions on $\mathbb{R}^d$ (see for example \cite{fleming2006controlled} Lemma 4.1.) we have existence of a function $\phi\in C^{1, 2}$ such that
	 $$
	 \tilde{U}(t, p) - \phi(t, p) \leq \tilde{U}(t_0, p_0) -\phi(t_0, p_0)\text{ with } (\nabla\phi(t_0, p_0), \partial_{pp}^2 \phi(t_0, p_0)) = (d, A).
	 $$
	 So finally we have on a compact neighborhood of $x_0$:
	 $$
	 U(t, p, e) - \phi(t, p) - h(e - e_0) \leq U(t_0, p_0, e_0) - \phi(t_0, p_0) - h(e_0-e_0).
	 $$
	 This local domination can then be extended to the whole domain $\mathcal{E}^K$. \\
	 
	 Finally we show that $\tilde{U}$ is a USC function. Fix $\varepsilon>0$ and $(t, p)$. Since $U$ is USC and $h$ continuous, for any $e\in B_r(e_0)$ we can find $r_{e}$ such that on $B_{r_e}(t, p, e)$ we have
	 $$
	 U + h(\cdot - e_0)\leq U(t, p, e) + h(e-e_0) + \varepsilon.
	 $$
	 The collection $\big( B_{\frac{r_e}{2}}(t, p, e) \big)_{e\in B_{r}(e_0)}$ forms an open covering of $\{t\}\times \{p\} \times B_{r}(e_0)$ which is a compact set by Lemma \ref{lemma:topology_a}. Thus we may find a finite sequence $\big( B_{\frac{r_{e_i}}{2}}(t, p, e_i) \big)_{1 \leq i \leq N}$ that covers $ \{t\}\times \{p\} \times B_{r}(e_0)$. Consider $r_{*} = \underset{1 \leq i \leq N}{\min} \frac{r_{e_i}}{2}$. Now take any $(s, q) \in B_{r_*}(t, p)$, then fo any $e\in B_{r}(e_0) $ there is some $i\in \{1, \dots, N\}$ such that $(t, p, e) \in B_{r_{e_i}/2}(t, p, e_i)$. Hence we get
	 $$
	 \|(s, q, e) - (t, p, e_i)\| \leq \frac{r_{e_i}}{2} + r^* \leq r_{e_i}
	 $$
	 so $(s, q, e)\in B_{r_{e_i}}(t, p, e_i)$ and consequently
	 $$
	 U(s, q, e) - h(e-e_0)\leq U(t, p, e_i) + h(e_i-e_0) + \varepsilon \leq \tilde{U}(t, p) + \varepsilon.
	 $$
	Passing to the supremum in $e\in B_{r}(e_0)$ in the LHS we get that $\tilde{U}$ is USC.
	\end{proof}

\section{Crandall Ishi's lemma}
\label{appendix:crandall}
The most crucial point to prove comparison result for viscosity solutions is the Crandall-Ishi's lemma that allows to deal with the second order terms. In the general case the Crandall Ishi's lemma is proved for subset of $\mathbb{R}^n$, see \cite{crandall1992user}. Hence our particular domain requires an adaptation of the classic version of the Lemma.
	\begin{lemma}	
	\label{lemma:crandall_ishi}
	Let $\phi_1\in C^{2}([0, T]^2)$, $\phi_2\in C^{1}(\mathbb{R}^2)$ and $\phi_3\in C^{0}\big((\mathcal{Z}_T^K)^2\big)$, $u\in USC(\mathcal{E}^K)$ and $v\in LSC(\mathcal{E}^K)$. Suppose we have $(t_0, p_0, z_0)\in \mathbb{R}^2\times \mathbb{R}^2\times (\mathcal{Z}^K_T)^2$ such that
	\begin{equation}
	\label{eq:crandall_enonce_a}
	\begin{array}{rl}
	&u(t^1_0, p^1_0, z^1_0)-v(t^2_0, p^2_0, z^2_0)-\phi_1(t_0) -\phi_2(p_0) - \phi_3(z_0) \\
	&~~ = \underset{t, p, z \in   \mathbb{R}^2 \times [0, T]^2 \times (\mathcal{Z}^K_T)^2}{\sup}	u(t^1, p^1, z^1)-v(t^2, p^2, z^2)-\phi_1(t) -\phi_2(p) - \phi_3(z).	
	\end{array}
	\end{equation}
	Then for any $\varepsilon$ there is $(A_{\varepsilon}, h)$ and $(B_{\varepsilon}, h)$ in $\mathbb{R} \times C^0(\mathcal{Z}^K_T)$ such that 
	$$
	\big( (\nabla_1 \phi_1(t_0), \nabla_1 \phi_2(p_0)), A_{\varepsilon}, h)\in \overline{\mathcal{J}}^{+}u(t_0^1, p^1_0, z^1_0),~~ \big((-\nabla_2 \phi_1(t_0), -\nabla_2 \phi_2(p_0)), B_{\varepsilon}, h)\in \overline{\mathcal{J}}^{-}v(t^2_0, p^2_0, z^2_0)
	$$
	 and that
	\begin{equation}
	\label{eq:crandall_ishi_statement_a}	
	-(\varepsilon^{-1} + |H \phi_2(p_0)|)I_{2}\leq \begin{pmatrix}
	A_{\varepsilon} & 0 \\ 0 & -B_{\varepsilon}
	\end{pmatrix} \leq H \phi_2 (p_0) + \varepsilon H \phi_2(p_0)^2
	\end{equation}
	where $H$ is the Hessian operator and $|A|$ denotes the spectral radius of the matrix $A$.
	\end{lemma}
 Note that even thought this extension is not straightforward we benefit from the fact that in $(\mathbf{HJB})_K$ the second order derivative concerns a real variable. Therefore to prove this result we are going to benefit from the usual Crandall Ishi's lemma.

	\begin{proof}
	
	Suppose there exists $\mathcal{V}$ a compact neighborhood of $(t_0, p_0, z_0)$ in $\mathcal{E}^K$ such that on $\mathcal{V}\backslash \{t_0, p_0, z_0\} $ we have 
	\begin{align*}
	(u - v)(t_0, p_0, z_0) & \geq (u-v)(t, p, z)  - \phi_2(p) + \phi_2(p_0) - \phi_1(t) + \phi_1(t_0)- \phi_3(z) + \phi_3(z_0)\\
	&\geq (u-v)(t, p, z)- \phi_2(p) + \phi_2(p_0) + \nabla \phi_1(t_0)(t_0 - t) + \mathcal{O}(\|t-t_0\|^2) - \phi_3(z) + \phi_3(z_0) \\
	&>(u-v)(x, y, z)- \phi_2(p) + \phi_2(p_0) + \nabla \phi_1(t_0)(t_0 - t) - C\|t-t_0\|^2 - h(z_0^1 - z^1) - h(z_0^2 - z^2)
	\end{align*}	
	where $h$ is any modulus of continuity of the function $\phi_3$ and $C$ a positive constant. For $x \in \mathbb{R}$ consider $g_j(x) = \partial_j \phi_1(t_0)x - Cx^2$ for $j=1$ or $2$. Hence on $\mathcal{V} \backslash (t_0, p_0, z_0)$ we have:
	\begin{equation}	
	\label{eq:crandall_proof_a}
	\begin{array}{ll}
	(u - v)(t_0, p_0, z_0) >& u(t^1, p^1, z^1) - v(t^2, p^2, z^2) - \phi_2(p) + \phi_2(p_0)\\
	& - h(z^1_0-z^1) - h(z^2_0-z^2) + g_1(t^1_0-t^1) + g_2(t^2_0-t^2)	
	\end{array}
	\end{equation}
	with equality at $(t_0, p_0, z_0)$ and with $	h(0) = g_i(0) = 0,~~g'_j(0) = \partial_j \phi_1(t_0)$.\\

	We can always assume that there exists $r>0$ so that $\mathcal{V}$ is of the form
	$$
	\mathcal{V} = \big(B_r(t_0)\times B_r(p_0)\times B_r(z_0)\big) \cap \mathcal{E}^K
	$$
	where $B_r(x)$ denotes the closed ball with center $x$ and radius $r$.	We define the following functions
	$$
	\begin{array}{ll}
		\tilde{u}(p^1) & = \underset{t, z \in  B_r(t_0)\times B_r(z_0) }{\sup } u(t^1, p^1, z^1) -  h(z^1_0-z^1) +  g_1(t^1_0-t^1)  \\
		\tilde{v}(p^2) & = \underset{t, z \in  B_r(t_0)\times B_r(z_0) }{\inf } v(t^2, p^2, z^2) -  h(z^2_0-z^2) +  g_2(t^2_0-t^2) 
	\end{array}
  	$$
  	where the supremums above are taken for $(t, z)$ such that $(t, p, z)\in \mathcal{E}^K$. The functions $\tilde{u}$ and $\tilde{v}$ are respectively USC and LSC functions since the supremums are taken over compact subsets (see the proof of Lemma \ref{lemma:equivalence_definitions}). And we have
	$$
	\tilde{u}(p^1) - \tilde{v}(p^2) - \phi_2(p) \leq \tilde{u}(p^1_0) - \tilde{v}(p^2_0) - \phi_2(p_0).
	$$
	Thus by the Crandall Ishi's lemma (see for example Theorem 6.1. in \cite{fleming2006controlled}) there exists $(A_{\varepsilon}, B_{\varepsilon})$ satisfying \eqref{eq:crandall_ishi_statement_a} such that
	$$
	(\partial_1 \phi_2(p_0), A_{\varepsilon} )\in \overline{\mathcal{J}}^{+}\tilde{u}(p^1_0) \text{ and }(-\partial_2\phi_2(p_0), B_{\varepsilon} )\in \overline{\mathcal{J}}^{-}\tilde{v}(p^2_0).
	$$
	Consequently there exist a sequence $(q_n, A_n,p^1_n, \tilde{u}(p^1_n))_{n \in \mathbb{N}}$ such that
	$$
	\underset{n\rightarrow + \infty}{\lim}(q_n, A_n, p^1_n, \tilde{u}(p^1_n)) = (\partial_1 \phi_2(p_0), A_{\varepsilon} , p^1_0, \tilde{u}(p^1_0)),~\text{ and }\forall~ n\geq 0,~ (q_n, A_n, p^1_n, \tilde{u}(p^1_n))\in \mathcal{J}^{+}\tilde{u}(p^1_n).
	$$
	So for any $n$ we have
	$$
	\tilde{u}(p^1) \leq \tilde{u}(p_n^1) + q_n(p^1-p^1_n) +  \frac{1}{2}A_n(p^1 - p^1_n)^2 + o(|p^1-p^1_n|^2).
	$$
	Consider $t^1_n$ and $z^1_n$ such that 
	$$
	\tilde{u}(p^1_n) = u(t^1_n, p^1_n, z^1_n) - h(z^1_0 - z^1_n) + g_1(t_0^1 - t^1_n)
	$$
	such maximizers exist by compacity. We show that $(t^1_n, p^1_n, z^1_n)$ converges towards $(t^1_0, p^1_0, z^1_0)$, we assume it for now. Equation \eqref{eq:crandall_proof_a} implies that for any $(t, p, z)$ we have
	\begin{align*}
	u(t^1, p^1, z^1) \leq & ~u(t^1_n , p^1_n, z^1_n) + q_n(p^1-p^1_n) + \frac{1}{2}A_n(p^1 - p^1_n)^2 + o(|p^1 - p^1_n|^2) \\
	&~  - h(z^1_0 - z^1_n) + h(z^1_0 - z^1)  +  g_1(t^1_0 - t^1_n)  - g_1(t^1_0 - t^1) .	
	\end{align*}
	Consider the function $h_n(z^1) = - h(z^1_0 - z^1_n) +  h(z^1_0 - z^1_n - z^1)$ such that $h_n(0) = 0$ and
	$$
	h_n(z^1-z^1_n) = -h(z^1_0 - z^1_n) +  h(z^1_0 - z^1).
	$$
	Since $z_n^1 $ converges towards $z^1$ the sequence $(h_n)_{n\geq 0}$ converges uniformly towards $h$ because $h$ is continuous and because we are working on compact neighborhood. Consider $ q_n^1 = \partial_1\phi_1(t_0) - 2C(t^1_0 - t^1_n)$ that converges towards $\partial_1\phi_1(t_0) $ 
	\begin{align*}
	g_1(t^1_0 - t^1_n)  - g_1(t^1_0 - t^1) &= q_n^1(t^1 - t^1_n)+ C(t^1_n - t^1)^2.
	\end{align*}
	Thus we have 
	\begin{align*}
	u(t^1,p^1, z^1) \leq& u(t^1_n,p^1_n, z^1_n) + q_n(p^1-p^1_n) + \frac{1}{2}A_n(p^1 - p^1_n)^2 + o(|p^1 - p^1_n|^2)\\
	&  + q^1_n(t^1 - t^1_n) + o(|t^1 - t^1_n|) + h_n(z^1 - z^1_n)	
	\end{align*}
	hence $\big((q_n^1, q_n), A_n, h_n\big)\in \mathcal{J}^{+}u(t^1_n, p^1_n, z^1_n)$ and 
	$$
	\big((q_n^1, q_n), A_n, h_n\big) \rightarrow \big((\partial_1\phi_1(t_0),\partial_1\phi_2(p_0) ), A_{\varepsilon}, h \big)
	$$
	
	 Finally we show that $(t^1_n, p^1_n, z^1_n)\underset{n\rightarrow + \infty}{\rightarrow} (t^1_0, p^1_0, z^1_0)$ which will imply the conclusion that
	$$
	\big( ( \partial_1 \phi_1(t_0),\partial_2 \phi_2(p_0)) , A_{\varepsilon}, h) \in \overline{\mathcal{J}}^{+}u(t^1_0, p^1_0, z^1_0).
	$$ 
	We have for any $n\geq 0$:
	$$
	\tilde{u}(p^1_n) = u(t_n^1, p^1_n, z^1_n) - h(z_0^1 - z_n^1) - g_1(t_0^1 - t_n^1).
	$$
	Consider any $(t^1, z^1)\in \overline{(t^1_n, z^1_n)}_{n\geq 0}$. Since $\tilde{u}(p^1_n)\rightarrow\tilde{u}(p^1_0)$, by upper semi-continuity of $u$ and by the definition of $\tilde{u}$ we get
	$$
	u(t^1, p^1_0, z^1) - h(z^1_0 - z^1) + g_1(t_0^1 - t^1) \geq  u(t^1_0, p^1_0, z^1_0).
	$$
	Which implies that $(t^1, z^1) = (t^1_0, z^1_0)$ since everywhere else the above inequality is false because of Equation \eqref{eq:crandall_proof_a}. Hence we get
	$$
	\big( (q^1_n, q_n),A_n, h_n, (x^1_n, y^1_n, z^1_n) \big) \underset{n \rightarrow \infty}{\rightarrow} \big(( \partial_1\phi_1(t_0), \partial_1\phi_2(p_0)),A_{\varepsilon}, h,(t_0^1,p^1_0, z^1_0) \big)
	$$
	and so
	$$
	\big( (\partial_1 \phi_1(t_0), \partial_1 \phi_2(p_0)), A_{\varepsilon}, h \big)\in \overline{\mathcal{J}}^{+}u(t^1_0, p_0^1, z_0^1).
	$$
	Similarly we get 
	$$
	\big( (-\partial_2 \phi_1(t_0), -\partial_2 \phi_2(p_0) ), B_{\varepsilon}, h\big)\in \overline{\mathcal{J}}^{-}v(t^2_0, p_0^2, z_0^2).
	$$
	This concludes the proof.
	\end{proof}

\section{Existence of $\mathcal{R}^{\alpha, \gamma}$}
\label{appendix:existence_change_representation}

Consider $t\geq 0$ we have for any $j\geq0$
\begin{equation}
\label{eq:existence_change_a}
\theta^{a(j)}_t(T)=\sum_{i = 1}^{n} c^{a,i}_t(-\gamma_i)^j e^{-\gamma_i(T-t) }\text{ and }\theta^{b(j)}_t(T)=\sum_{i = 1}^{n} c^{b,i}_t(-\gamma_i)^j e^{-\gamma_i(T-t) }.
\end{equation}
So let $A$ be the matrix with coefficient $A_{ij} = (-\gamma_i)^j$. This is a Vandermonde matrix which is invertible. By Equation \eqref{eq:existence_change_a} we have
$$
c^{a,i}_t  = e^{\gamma_i(T-t)} \sum_{j = 1}^n (A^{-1})_{ij} \theta_t^{a(j)}(T)\text{ and }c^{b,i}_t  = e^{\gamma_i(T-t)} \sum_{j = 1}^n (A^{-1})_{ij} \theta_t^{b(j)}(T).
$$
So we define $\mathcal{R}^{\alpha, \gamma}$ for $(t, x) \in \mathcal{E}^{K_{\alpha, \gamma}}$ by
$$
\mathcal{R}^{\alpha, \gamma}(t, x) = \big(t, p, i, c^a(t, x), c^b(t,x)\big)
$$
where
$$
c^a(t, x) = \big( e^{\gamma_i(T-t)} \sum_{j = 1}^n (A^{-1})_{ij} \theta^{a(j)}(T) \big)_{1\leq j \leq n} \text{ and }c^b(t,x) = \big( e^{\gamma_i(T-t)} \sum_{j = 1}^n (A^{-1})_{ij} \theta^{b(j)}(T) \big)_{1\leq j \leq n}.
$$
By Lemma \ref{lemma:topology_a} the map $\mathcal{R}^{\alpha, \gamma}$ is continuous and by construction we have for any $t\geq 0$
$$
\mathcal{R}^{\alpha, \gamma}(t, X_t) = (t, Y^{\alpha, \gamma}_t).
$$

\section{Proof of Lemma \ref{lemma:existence_converging_sequence}}
\label{appendix:existence_converging_sequence}

We are going to approximate the integral in Equation \eqref{eq:completely_monotone_rep} by Riemann sum. We take $A_n = \sqrt{n}$  and $(a_i)_{0\leq i \leq n -1}$ a regular grid of $[0, A_n]$ with mesh $\frac{1}{\sqrt{n}}$. We set
	$$
	K_n(t) = \sum_{i = 0}^{n - 1}e^{- a_{i+1} t}\int_{a_i}^{a_{i + 1}}m(\mathrm{d}u)\leq K(t).
	$$
	For $t\in \mathbb{R}_+$, we have
	$$
	K(t) - K_n(t)  = \sum_{i = 0}^{n - 1} \int_{a_{i}}^{a_{i+1}}  m(\mathrm{d}u) \int_{u}^{a_{i+1}}  te^{-tv} \mathrm{d}v  - \int_{A_n}^{+\infty}e^{-tu} m( \mathrm{d}u).
	$$
	 Hence for any $T$ and $t\leq T$:
	\begin{align*}
	|K_n(t) - K(t)| &\leq \sum_{i = 0}^{n - 1} T \int_{a_i}^{a_{i + 1}}m(\mathrm{d}u) (a_{i + 1}-a_i)  + \int_{A_n}^{+\infty} m(\mathrm{d}u)	\\
	& \leq \frac{T}{\sqrt{n}}\int_{0}^n m(\mathrm{d}u)  + \int_{A_n}^{+\infty} m(\mathrm{d}u)\\
	& \leq \frac{T}{\sqrt{n}}\int_{0}^{+\infty} m(\mathrm{d}u)  + \int_{A_n}^{+\infty} m(\mathrm{d}u)
	\end{align*}
	which goes to $0$ when $n$ goes to infinity, uniformly on $t\in [0, T]$. Hence the sequence $K_n$ converges uniformly towards $K$ and is dominated by $K$ so $K_n$ converges in $\mathbb{L}_1$ towards $K$. \\
	
	Set $\alpha_n = K(0) - K_n(0)$ and $\beta_n = \frac{\alpha_n}{\|K\|_1 - \|K_n\|_1}$ and consider $\tilde{K}_n = K_n + \alpha_ne^{-\beta_n\cdot}$, we have for any $n$
$$
\tilde{K}_n(0) = K(0)\text{ and }\|\tilde{K}_n\|_1 = \|K\|_1
$$
and $\tilde{K}_n \rightarrow K$ in $\mathbb{L}_1$. Thus the sequence $(\tilde{K}_n)_{n\geq 0}$ gives the result.

\section{Choice of an approximating sequence for Section \ref{subsec:high_dimension_method}}
\label{appendix:large_dimension}

Consider $f^{\alpha, \lambda}$ the Mittag-Leffler density function, see \cite{haubold2011mittag} for details, and $\mathcal{L}$ the Laplace transform operator we have
$$
\mathcal{L}[f^{\alpha, \lambda}](t) = \frac{\lambda}{\lambda + t^{\alpha}}.
$$
Moreover for any $\beta\in (0, 1)$ we have
$$
\mathcal{L}[D^{\beta}f^{\alpha, \lambda}](t) = \frac{\lambda}{\lambda + t^{\alpha}}\frac{1}{t^{\beta}},
$$
where $D^{\beta}$ is the fractional derivative operator, see \cite{samko1993fractional} for details. Hence we have
$$
\frac{\lambda}{\lambda + (t+\varepsilon)^{\alpha}}\frac{1}{(t+\varepsilon)^{\beta}} = \int_0^{+\infty} e^{-pt} e^{-p\varepsilon} D^{\beta}f^{\alpha, \lambda}(p) \mathrm{d}p.
$$
Therefore to build the $(\alpha_n, \beta_n)_{n\geq 0}$ we simply use Riemman sums. More precisely for any $n$ we set
$$
K_n(t) = \sum_{i = 0}^{n-1} \big( e^{- c_{i+1} \varepsilon} D^{\beta}f^{\alpha, \lambda}(c_{i+1}) (c_{i+1} - c_{i}) \big) e^{-c_{i+1} t},
$$
where $c_i= i / \sqrt{n}$. Hence we have
$$
\alpha^n_i = \frac{i+1}{\sqrt{n}} \text{ and } \gamma^n_i = e^{- c_{i+1} \varepsilon} D^{\beta}f^{\alpha, \lambda}(c_{i+1}) (c_{i+1} - c_{i}).
$$
Finally we rescale the $\alpha^n$ to obtain equality of the $L^1$ norm.

\section{Probabilistic representation of PIDE in high dimension}
\label{sec:branching_method}

We are going to use a probabilistic representation based on branching processes. This method is insensitive to the dimension of the domain of the PIDE. Theoretically the method works for any semi-linear PIDE admitting a strong solution and with a generator that can be written as a power serie. Thought this is not the case for $(\mathbf{HJB})_{\alpha, \gamma}$, in order to implement this method we approximate the generator of the PIDE by a second order polynomial and assume that the approximated PIDE have a strong solution. Thus we are left with an PIDE of the form
	\begin{equation*}
	(\mathbf{HJB})'_{\alpha, \gamma}:~~-\partial_tU -\mathcal{L}U - f(U, D^{\alpha}U) = 0, ~~u(T, \cdot) = 0 \text{ on }\mathbb{Z} \times \mathbb{R}^n \times \mathbb{R}^n 	
	\end{equation*}
	where
	\begin{align*}
	f(U, D^{\alpha}U)(t, x) =& f_0(t, x) + f_1(t, x) U(t, x) \\
	&+ f^a_{2, 1}(t, x) D^{\alpha}_aU(t, x) + f^b_{2, 1}(t, x) D^{\alpha}_bU(t, x) \\
	&+ f^a_{2, 2}(t, x) D^{\alpha}_aU(t, x)^2 + f^b_{2, 2}(t, x) D^{\alpha}_bU(t, x)^2.
	\end{align*}
	The operator $\mathcal{L}$ is defined by $\mathcal{L}U(t, x) = - \langle \gamma, \nabla^c_a U(t,x) \rangle - \langle \gamma, \nabla^c_b U(t,x) \rangle $.\\
	
	Consider a process $\tilde{X}^{t,x}$ starting at time $t$ with initial state $x$ such that $(t, x) \in \mathcal{E}^{n}$ and whose dynamics is driven by the infinitesimal generator $\mathcal{L}$. The Feynman-Kac formula gives
	\begin{equation}
	\label{eq:branching_eq_b}
	U(t, x) = \mathbb{E}[ \frac{f(U, D^{\alpha}U)(\tau,\tilde{X}^{t,x}_{t+\tau})}{\rho(\tau)}\mathbf{1}_{t+\tau < T} ]	
	\end{equation}
	where $\tau$ is a positive random variable with density $\rho$.\\
	
	 We show in Appendix \ref{appendix:branching_proba} that there exists an appropriate probability measure $\mathbb{P}_{\mathcal{T}}$ on the set 
	$$
	\mathcal{T} = \big\{  0, 1, (2, j, d, \varepsilon), \text{ with }  d\in \{0, 1, 2\}, ~ j\in \{a, b\},~\varepsilon\in \{0, 1\}^d  \big\}
	$$
	and a set of functions $(g_{\tau})_{\tau \in \mathcal{T}}$ from $ [0, T]\times \mathbb{Z} \times \mathbb{R}_+^n \times \mathbb{R}_+^n $ such that for any random variable $\xi$ with law $\mathbb{P}_{\mathcal{T}}$ we have
	\begin{equation}
	\label{eq:branching_eq_a}
	f(U, D^{\alpha}U)(t, x) = \mathbb{E}[g_{\xi}(U, D^{\alpha}U)(t, x)].
	\end{equation}
	The set $(g_{\tau})_{\tau \in \mathcal{T}}$ is defined by
$$
	g_{0}(U, D^{\alpha}U)(t, x) = f_{0}(t, x) \mathbb{P}(l = 0)^{-1},~ g_{1}(U, D^{\alpha}U)(t, x) = f_{1}(t, x) \mathbb{P}(l = 1)^{-1} U(t, x)
$$
and
$$
g_{(2, j, d, \varepsilon)}(U, D^{\alpha}U)(t, x) = f^j_{2, d}(t, x)\mathbb{P}\big(l = (2, j, d, \varepsilon)\big)^{-1} \displaystyle \prod_{k = 1}^d U(t, x+\Delta^j\varepsilon_k)(-1)^{1-\varepsilon_k},
$$
	where $\Delta^a$ (resp. $\Delta^b$) is the jump corresponding to an ask (resp. bid) market order, namely $\Delta^a = (-1, \alpha, 0)\text{ and } \Delta^b = (1,0,  \alpha)$ (We recall that the price variable is no longer part of the domain).\\
	
We now define a branching process in the following way: any particle is noted by $(t,x, l_0, l_1, \dots, l_n)$ where $(x,t)\in \mathcal{E}^n$ and the $l_i$'s belong in $\mathcal{T}$. The variable $x$ denotes the initial position of the particle and $t$ its birth time, $l_n$ is the label of the particle, $l_{n-1}$ the label of its parent, and so on. The lifetime of the particles are i.i.d random variables with density $\rho$\\
 	
We now describe the evolution of the particle. Consider a particle born at time $s$ at the state $x$ with lifetime $\tau$. During its lifetime the particle state is described by its position: $\big((i^{s, x}_t, c^{s, x;a}_t, c^{s, x;b}_t)\big)_{s\leq t \leq s+\tau }$ in $\mathbb{Z}\times \mathbb{R}_+^n\times \mathbb{R}_+^n$. The dynamics of the particle position is given by
 	$$
 	\mathrm{d}c^{j, i}_t = -\gamma_i c^{j, i}_t\mathrm{d}t, \text{ for }i\in\{1, \dots, n\}\text{ and }j=a\text{ or }b.
 	$$
 	The other components are constants.	Note that this dynamics corresponds to the infinitesimal generator $\mathcal{L}$. When the particle dies it gave birth to independent particles. The number and type of children particles depend on the label $l_n$ of the particle:
	\begin{itemize}
	\item if $l_n = 0$: $0$ child
	\item if $l = 1$: $1$ child
	\item if $l_n = (2, d, j, \varepsilon)$: $d$ children \begin{itemize}
	\item if $j = a$ the initial state of the $i-th$ child particle is $X^{s;x}_{s + \tau} + \Delta^a \varepsilon_i$
	\item if $j = b$ the initial state of the $i-th$ child particle is $X^{s;x}_{s + \tau} + \Delta^b \varepsilon_i$
	\end{itemize}
	\end{itemize}
	The labels of the children particles are i.i.d. random variables with law $\mathbb{P}_{\mathcal{T}}$. We note $\mathcal{C}_p$ the set of the children particles.\\
	
	Considering a particle starting at point $(t, x)$, Equations \eqref{eq:branching_eq_a} and \eqref{eq:branching_eq_b} give
	$$
	U(t, x) = \mathbb{E}[\frac{a(l,t+\tau, X^{t;x}_{t+\tau })}{\rho(\tau)} \displaystyle\prod_{c\in \mathcal{C}_p} U(t+\tau, X_c) \mathbf{1}_{t+\tau < T} ]
	$$
	where $X_c$ denotes the initial position of the child particle $c$ and where $a$ is defined by
	\begin{align*}
	a(i, t, x) &= f_{0}(t, x) \mathbb{P}(l = i)^{-1}, \text{ for } i = 1 \text{or }2\\
	a\big((2, j, d, \varepsilon), t, x \big) &= f^j_{2, d}(t, x)\mathbb{P}\big(l = (2, j, d, \varepsilon)\big)^{-1} \prod_{k = 1}^d (-1)^{1-\varepsilon_k} 
	\end{align*}
	By iterating the above equality to the descendents of the particle and assuming that the number of descendent particles born before the time horizon $T$ is almost surely finite we can evaluate $U(t, x)$ using Monte Carlo simulation. For more details on this method we refer to \cite{henry2019branching}.\\

\subsection{Existence of a measure for the particle method}
\label{appendix:branching_proba}
We have 
$$
f(u, Du)(t, x) = \mathbb{E}[f_I(u, Du)(t, x)]
$$
where $I$ is a random variable with values in $\{ 0, 1, 2\}$, and
\begin{align*}
f_0(u, Du)(x) &= f_0(t, x) \mathbb{P}(I = 0)^{-1} \\
f_1(u, Du)(x) &= f_1(t, x) u(t, x) \mathbb{P}(I = 1)^{-1}\\
f_2(u, Du)(x) &= \mathbb{E}[f_{l}(u, Du)(t, x)] \mathbb{P}(I = 2)^{-1}
\end{align*}
	where $l$ is a random variable with values in $\{(a,1), (b, 1), (a, 2) , (b, 2)\}$ and
	\begin{align*}
	f_{(j, d)}(u, Du)(t, x) = f^j_{2, d}(t, x)D^ju(t, x)^d \mathbb{P}(l = (j, d))^{-1}.
	\end{align*}
	Finally we have
	$$
	D^ju(t, x)^d = 2^d ~\mathbb{E}[\displaystyle\prod_{k = 1}^d u(t, x+\Delta^j\varepsilon_k)(-1)^{1-\varepsilon_k} ]
	$$
	with $(\varepsilon_i)_{1\leq i \leq d}$ i.i.d. random variables with law $Ber(\frac{1}{2})$. Thus finally
	$$
	f(u, Du)(t, x) = \mathbb{E}[g_{l}(u, Du)(t, x)]
	$$
	with $l$ is a random variable whose law is the uniform probability measure on the set $\mathcal{L} = \big\{ 0, 1, (2, j, d, \varepsilon), \text{ with }  d\in \{0, 1, 2\}, ~ j\in \{a, b\},~\varepsilon\in \{0, 1\}^d  \big\}$ and where
	\begin{align*}
	g_{0}(u, Du)(t, x) &= f_{0}(t, x) ~\mathbb{P}(l = 0)^{-1}\\
	g_{1}(u, Du)(t, x) &= f_{1}(t, x)u(t, x) ~ \mathbb{P}(l = 1)^{-1} \\
	g_{(2, j, d, \varepsilon)}(u, Du)(t, x) &= f^j_{2, d}(t, x) \mathbb{P}\big(l = (2, j, d, \varepsilon)\big)^{-1} \displaystyle\prod_{k = 1}^d u(t, x+\Delta^j\varepsilon_k)(-1)^{1-\varepsilon_k} .
	\end{align*}

\bibliographystyle{plain}
\bibliography{biblio}

\begin{thebibliography}{10}

\bibitem{jaber2018lifting}
Eduardo Abi~Jaber.
\newblock Lifting the {H}eston model.
\newblock {\em Quantitative Finance}, 19(12):1995--2013, 2019.

\bibitem{alfonsi2016dynamic}
Aur{\'e}lien Alfonsi and Pierre Blanc.
\newblock Dynamic optimal execution in a mixed-market-impact {H}awkes price
  model.
\newblock {\em Finance and Stochastics}, 20(1):183--218, 2016.

\bibitem{avellaneda2008high}
Marco Avellaneda and Sasha Stoikov.
\newblock High-frequency trading in a limit order book.
\newblock {\em Quantitative Finance}, 8(3):217--224, 2008.

\bibitem{bacry2016estimation}
Emmanuel Bacry, Thibault Jaisson, and Jean-Fran{\c{c}}ois Muzy.
\newblock Estimation of slowly decreasing {H}awkes kernels: application to
  high-frequency order book dynamics.
\newblock {\em Quantitative Finance}, 16(8):1179--1201, 2016.

\bibitem{cartea2015algorithmic}
{\'A}lvaro Cartea, Sebastian Jaimungal, and Jos{\'e} Penalva.
\newblock {\em Algorithmic and high-frequency trading}.
\newblock Cambridge University Press, 2015.

\bibitem{cartea2014buy}
{\'A}lvaro Cartea, Sebastian Jaimungal, and Jason Ricci.
\newblock Buy low, sell high: A high frequency trading perspective.
\newblock {\em SIAM Journal on Financial Mathematics}, 5(1):415--444, 2014.

\bibitem{crandall1992user}
Michael~G Crandall, Hitoshi Ishii, and Pierre-Louis Lions.
\newblock User’s guide to viscosity solutions of second order partial
  differential equations.
\newblock {\em Bulletin of the American mathematical society}, 27(1):1--67,
  1992.

\bibitem{dayri2015large}
Khalil Dayri and Mathieu Rosenbaum.
\newblock Large tick assets: implicit spread and optimal tick size.
\newblock {\em Market Microstructure and Liquidity}, 1(01):1550003, 2015.

\bibitem{fleming2006controlled}
Wendell~H Fleming and Halil~Mete Soner.
\newblock {\em Controlled {M}arkov processes and viscosity solutions},
  volume~25.
\newblock Springer Science \& Business Media, 2006.

\bibitem{gueant2016financial}
Olivier Gu{\'e}ant.
\newblock {\em The Financial Mathematics of Market Liquidity: From optimal
  execution to market making}, volume~33.
\newblock CRC Press, 2016.

\bibitem{gueant2013dealing}
Olivier Gu{\'e}ant, Charles-Albert Lehalle, and Joaquin Fernandez-Tapia.
\newblock Dealing with the inventory risk: a solution to the market making
  problem.
\newblock {\em Mathematics and financial economics}, 7(4):477--507, 2013.

\bibitem{haubold2011mittag}
Hans~J Haubold, Arak~M Mathai, and Ram~K Saxena.
\newblock Mittag-{L}effler functions and their applications.
\newblock {\em Journal of Applied Mathematics}, 2011, 2011.

\bibitem{henry2019branching}
Pierre Henry-Labordere, Nadia Oudjane, Xiaolu Tan, Nizar Touzi, Xavier Warin,
  et~al.
\newblock Branching diffusion representation of semilinear {PDE}s and {M}onte
  {C}arlo approximation.
\newblock In {\em Annales de l'Institut Henri Poincar{\'e}, Probabilit{\'e}s et
  Statistiques}, volume~55, pages 184--210. Institut Henri Poincar{\'e}, 2019.

\bibitem{hewlett2006clustering}
Patrick Hewlett.
\newblock Clustering of order arrivals, price impact and trade path
  optimisation.
\newblock In {\em Workshop on Financial Modeling with Jump processes, Ecole
  Polytechnique}, pages 6--8, 2006.

\bibitem{bachouch2018deep}
C{\^o}me Hur{\'e}, Huy{\^e}n Pham, Achref Bachouch, and Nicolas Langren{\'e}.
\newblock Deep neural networks algorithms for stochastic control problems on
  finite horizon, part i: convergence analysis.
\newblock {\em arXiv preprint arXiv:1812.04300}, 2018.

\bibitem{jacod1975multivariate}
Jean Jacod.
\newblock Multivariate point processes: predictable projection,
  {R}adon-{N}ikodym derivatives, representation of martingales.
\newblock {\em Zeitschrift f{\"u}r Wahrscheinlichkeitstheorie und verwandte
  Gebiete}, 31(3):235--253, 1975.

\bibitem{jacod2013limit}
Jean Jacod and Albert Shiryaev.
\newblock {\em Limit theorems for stochastic processes}, volume 288.
\newblock Springer Science \& Business Media, 2013.

\bibitem{jaisson2016rough}
Thibault Jaisson, Mathieu Rosenbaum, et~al.
\newblock Rough fractional diffusions as scaling limits of nearly unstable
  heavy tailed {H}awkes processes.
\newblock {\em The Annals of Applied Probability}, 26(5):2860--2882, 2016.

\bibitem{lillo2004long}
Fabrizio Lillo and J~Doyne Farmer.
\newblock The long memory of the efficient market.
\newblock {\em Studies in nonlinear dynamics \& econometrics}, 8(3), 2004.

\bibitem{madhavan1997security}
Ananth Madhavan, Matthew Richardson, and Mark Roomans.
\newblock Why do security prices change? a transaction-level analysis of {NYSE}
  stocks.
\newblock {\em The Review of Financial Studies}, 10(4):1035--1064, 1997.

\bibitem{merkle2014completely}
Milan Merkle.
\newblock Completely monotone functions: A digest.
\newblock In {\em Analytic Number Theory, Approximation Theory, and Special
  Functions}, pages 347--364. Springer, 2014.

\bibitem{reny1999existence}
Philip~J Reny.
\newblock On the existence of pure and mixed strategy {N}ash equilibria in
  discontinuous games.
\newblock {\em Econometrica}, 67(5):1029--1056, 1999.

\bibitem{samko1993fractional}
Stefan~G Samko, Anatoly~A Kilbas, Oleg~I Marichev, et~al.
\newblock {\em Fractional integrals and derivatives}, volume~1.
\newblock Gordon and Breach Science Publishers, Yverdon Yverdon-les-Bains,
  Switzerland, 1993.

\bibitem{sokol2013optimal}
Alexander Sokol et~al.
\newblock Optimal {N}ovikov-type criteria for local martingales with jumps.
\newblock {\em Electronic Communications in Probability}, 18, 2013.

\bibitem{touzi2012optimal}
Nizar Touzi.
\newblock {\em Optimal stochastic control, stochastic target problems, and
  backward SDE}, volume~29.
\newblock Springer Science \& Business Media, 2012.

\bibitem{wyart2008relation}
Matthieu Wyart, Jean-Philippe Bouchaud, Julien Kockelkoren, Marc Potters, and
  Michele Vettorazzo.
\newblock Relation between bid--ask spread, impact and volatility in
  order-driven markets.
\newblock {\em Quantitative finance}, 8(1):41--57, 2008.

\bibitem{chen2020optimal}
Ge~Zhang Ying~Chen, Zexin~Wang and Chao Zhou.
\newblock Optimal high frequency trading with thinned {H}awkes process.
\newblock {\em Working paper}.

\end{thebibliography}
	
\end{document}